\documentclass[11pt]{article}
\usepackage{arxiv2-packages}
\usepackage{arxiv-macros}
\usepackage[square,numbers]{natbib}

\usepackage{color-edits}
\usepackage[colorlinks=true, 
            linkcolor=blue,     %
            citecolor=blue,     %
            urlcolor=blue]      %
            {hyperref}
\usepackage{cleveref}
\usepackage{comment}
\usepackage{thm-restate} %
\addauthor{cp}{red}
\renewcommand{\thefootnote}{\fnsymbol{footnote}}

\title{Can Users Fix Algorithms? A Game-Theoretic Analysis of Collective Content Amplification in Recommender Systems}

\author{Ekaterina Fedorova\footnotemark[1] \and Madeline Kitch\footnotemark[2] \and Chara Podimata\footnotemark[3]}

\date{Feb 2026}
\begin{document}
\maketitle
\footnotetext[1]{Massachusetts Institute of Technology, \texttt{fedorova@mit.edu}}
\footnotetext[2]{Carnegie Mellon University, \texttt{mckitch@andrew.cmu.edu}}
\footnotetext[3]{Massachusetts Institute of Technology, \texttt{podimata@mit.edu}}
\renewcommand{\thefootnote}{\arabic{footnote}}

\begin{abstract}
Users of social media platforms based on recommendation systems ({RecSys}) (e.g. TikTok, X, YouTube) \emph{strategically} interact with platform content to influence future recommendations. On some such platforms, users have been documented to form large-scale grassroots movements encouraging others to purposefully interact with algorithmically suppressed content in order to \emph{counteractively} ``boost'' its recommendation. However, despite widespread documentation of this phenomenon, there is little theoretical work analyzing its impact on the platform or users themselves. We study a game between users and a RecSys, where users (potentially strategically) interact with the content available to them, and the RecSys---limited by preference learning ability---provides each user her approximately most-preferred item. We compare recommendations and social welfare when users interact with content according to their personal interests and when a collective of users intentionally interacts with an otherwise suppressed item. We provide sufficient conditions to ensure a \emph{pareto} improvement in recommendations and \emph{strict} increases in user social welfare under collective interaction, and provide a robust algorithm to find an effective collective strategy. Interestingly, despite the intended algorithmic protest of these movements, we show that for commonly assumed recommender utility functions, effective collective strategies also improve the utility of the RecSys. Our theoretical analysis is complemented by empirical results of effective collective interaction strategies on the GoodReads dataset and an online survey on how real-world users attempt to influence others' recommendations on RecSys platforms. Our findings examine how and when platforms' recommendation algorithms may incentivize users to collectivize and interact with content in algorithmic protest as well as what this collectivization means for the platform.
\end{abstract}
\section{Introduction}\label{sec:intro}
Social networking and media platforms such as TikTok, Instagram, and YouTube learn about users by tracking their interactions with content via likes, comments, views, etc. Using these interactions, they can infer user preferences and provide new content recommendations. From the user perspective, while the actual algorithm is a black box, the basic idea of how it functions is well-known \cite{devito_folk, eslami_facebook_folk_theory}. Given their heuristic knowledge, some users strategically interact with content to purposefully tailor the recommendations they receive \cite{cen2024measuringstrategizationrecommendationusers, karizat_algo_resistance, lgbtq_tiktok,haupt_recommending_2023}. Consider, for example, a user who enjoys a type of niche content. Though this content is not generally popular, the user would like it to feature very prominently in their recommended feed. Knowing that engaging with content of this type will likely cause it to be recommended again in the future, the user purposefully likes, comments, and watches this type of content \emph{more often} than they personally would like in the moment. From a theory point of view and to analyze the dynamics that emerge from this user adaptation, there are several game theoretic models of this self-interested strategic behavior \cite{haupt_recommending_2023,cen_user_2024}.

However, existing modeling of strategic user interactions do not acknowledge that recommendations are informed by both a user's own interactions with platform content and the interactions of \emph{other} users with said content. \emph{Collaborative filtering} (CF), a widely-adopted recommendation methodology \citep{cf_popular_evaluate, schafer_collaborative_2007}, allows the recommender to infer the preferences of one user from another because users and contents (aka \emph{items}) are assumed to be representable as vectors in respective user or item latent spaces. This enables the following association: when one user interacts with a piece of content, the recommender may predict how much interaction the item would elicit from other users based on the similarity of users' vectors in latent space. Users also have a heuristic understanding of how their interactions impact others' recommendations and some leverage this information to interact strategically to impact \emph{other users'} platform experiences~\cite{karizat_algo_resistance, eslami_facebook_folk_theory}. While the exact prevalence of such \emph{interpersonal} content interaction strategies is still unknown, in our survey of 100 recommendation platform users on Prolific, 32 had intentionally interacted or avoided interaction with content to influence other users' feeds (see Appendix \ref{app:survey} for more details). Interpersonal content interaction strategies are thus likely widespread, but at the same time, the constitute an under-studied component to understanding the outcome and impact of algorithmic content recommendation.

While there are many reasons a user may want to impact others' recommended feeds, interpersonal content interaction strategies have been particularly relevant to users who want to combat algorithmic suppression and injustice. Indeed,~\citet{karizat_algo_resistance} document that users believe certain types of content are suppressed by the algorithm and thus activist 
users \emph{intentionally} interact with such content to ``boost'' it via algorithmic recommendation to others. On an individual level, it is unlikely that \emph{one} user attempting to boost the popularity of suppressed content would impact algorithmic recommendations on a large scale. However, as many marginalized user and creator groups report recommendation inequity or suppression \cite{melo_booktok, camille_tiktok_vendetta, biddle_tiktok_2020}, users have organized \emph{large-scale grassroots movements} encouraging other users to purposefully engage with the content of those who are suppressed by the algorithm \cite{tiktok_blm,mccall_booktoks_2022}. Such movements have the potential to have a significant impact on the outcome of algorithmic content recommendation and the platform as a whole. Yet, existing game-theoretic recommendation models of self-interested strategic user behavior cannot account for the relevant inter-user effects, making it difficult to understand if/when users' collective efforts are actually theoretically justified and what guarantees the movements provide for user or platform welfare. 

To analyze both the incentives users have to participate in these movements and the resulting recommendations from the system, we present a novel game-theoretic model of \emph{collaborative filtering-based} algorithmic recommendation, which we use to answer two questions:
\begin{center}
\begin{enumerate}
    \item[{\bf (Q1)}] \emph{Effectiveness and sufficient conditions}: Are the documented collective interaction strategies of real-world movements theoretically effective for users seeking to improve the system and under what conditions are these movements successful?
    \item[{\bf (Q2)}] \emph{Platform/mechanism design implications}: Given that some collection of users take part in these movements, how is the platform's welfare impacted? Is it poor design to incentivize collective user strategies on the part of the recommender?
\end{enumerate}
\end{center}

Within our model of a collaborative-filtering recommendation system (RecSys), we compare the recommendations users receive when they interact with content according to their \emph{personal} interest in the item to those they receive when some users form a \emph{collective movement} to purposefully interact with a chosen suppressed type of content. We find that when there is a minority group of users whose preferences are \emph{sufficiently disparate} from those of a mainstream group, a collective can improve social welfare and recommendations. These results answer {\bf Q1} in the affirmative and provide formal conditions under which the movements are effective. 

The implication of movement effectiveness is that as long as relevant users are even minimally activist, 
the recommender implicitly incentivizes a collective to form and interact with content in a way that does not perfectly match their \emph{personal} preferences. And though this means users would \emph{not} interact with the algorithm \emph{truthfully} (according to personal preference), we highlight that under certain platform utility functions, since recommendations are improved and users are forced to increase their platform engagement (often allowing increased ad revenue for the platform), such collective movements---despite their algorithmic protest nature---actually also benefit the platform.

\subsection{Our Contributions}

\xhdr{Model.} In Section~\ref{sec:model}, we formally present our model of a CF RecSys collective strategic interactions. The model is an abstraction of an \emph{online matrix completion} RecSys protocol. In an online matrix completion (MC) setting, at each round, a recommender (equivalently, the platform, principal, or learner) randomly queries a selection of ratings, which are the numerical formalization of interactions (e.g., likes, comments, views, etc) $m$ users may have with $n$ items. A standard representation of this process is an $m \times n$ ratings matrix whose entries are slowly filled-in over rounds of recommendation exploration. Upon sufficient exploration, the recommender estimates a complete ratings matrix using a best low-rank approximation of the partially-known ratings matrix and uses this to provide each user their top item(s). However, as we prove formally, this process can pose issues for recommendations to a minority group whose preferences are significantly different the the majority. In particular, if just a few key ratings are not queried, then minority group preferences might not be learned at all. Unfortunately, rank-minimization is an NP-hard problem and generally not tractable to analyze. Thus, we present our model, which condenses the \emph{online} MC protocol into an abstracted \emph{one-round low-rank approximation} multi-agent game between users and the RecSys.

The model is built around a transformation across three matrices: $\bR^\star \rightarrow \tbR \rightarrow \widehat{\bR}$. $\bR^\star \in \R^{m\times n}_{\geq 0}$ is an unknown true \emph{personal} preference matrix, whose entries, $r^\star_{u,i}$ represent how much personal interest each of $m$ users have for $n$ items. However, users can choose to strategically interact with items in a way that does not align with their personal interest, which is represented as $\tbR \in \R^{m \times n}_{\geq 0}$, the revealed preference matrix. This is the information that gets passed onto the learner (it), but it, as in online MC, cannot actually know/learn every rating; we encompass these issues under a variable called the \emph{exploration limit}. The learner's item recommendation to each user must instead be made using a particular low-rank approximation of $\tbR$, called $\widehat{\bR}$, which is a function of $\tbR$ and the exploration limit. %
We use the \emph{social welfare} to capture how well the RecSys serves users.

To our knowledge, our model is the first tractable game-theoretic model created to study strategic users leveraging the inter-user rating effects of real-world RecSys; as such, it is an important contribution to the perspective of recommender systems as multi-agent games more broadly. While we use our model for the specific purpose of analyzing users' attempts at coordinated amplification of marginalized content, it is of independent interest for the analysis of many other cases in which users strategically interact with content to affect others' feeds. For example, the multi-agent game perspective is necessary to understand how bots or otherwise coordinated mass negative reviewing may impact other users' recommendations (see discussion in Section \ref{sec:discussion}).

\xhdr{Informal results: Users.} In Section \ref{sec:users} we formally study {\bf Q1} and show that when there is \textbf{(1)} a ``majority'' class of users who like \emph{``popular''} items; \textbf{(2)} a ``minority'' class who like \emph{``niche''} items, and \textbf{(3)} limited RecSys exploration, it may result in too little information learned by the RecSys about minority preferences. This results in good item recommendations for the majority, but popular item recommendations to the minority, who would prefer the niche items. However, if a collective of majority users purposefully interact with niche content, they may force the learner to be sensitive to minority preferences and strictly improve recommendations, thus answering {\bf Q1} in the affirmative.

To show that under conditions \textbf{(1)}-\textbf{(3)} recommendations and user welfare is poor, we represent \textbf{(1)} and \textbf{(2)} as a setting within our model where the true personal preference matrix is what we call a \emph{majority-minority matrix}. This is a matrix such that the majority and minority users prefer completely different items, meaning we can re-order items and users to form a block matrix with a majority user- and a minority user-block. The majority is called such because it holds greater weight in the system: its block has greater singular values than those of the minority block. In our model, when users submit ratings for items truthfully to the majority-minority personal preferences, the RecSys provides a top item to each user using the \emph{rank-reduced version} of the majority-minority matrix. Hence, for a whole range of exploration limits, the ratings information that is reduced out will be exactly the minority users' preferences because their preferences are linearly independent from the majority and they make up the smallest singular values of the matrix. Thus, the platform will be unable to provide good recommendations to the minority. From the online MC setting perspective, the minority users are ``harder to learn'' because they form a relatively small (information-wise) part of the matrix and so if too few minority user-item pairs are queried, the estimated complete matrix will zero-out minority user preferences for niche items.

Next, we provide sufficient conditions such that a collective movement of purposeful interaction with a minority item, improves this poor state of recommendation and welfare. Choosing an item and corresponding minority user subgroup whom we call ``picky'' that would be \emph{especially} hard for the RecSys to learn otherwise, a collective of majority users uprates the picky item by a collaboratively chosen value, $\eta$, in effort to improve recommendations to the particularly under-served users. As long as $\eta$ is enough to throw the largest singular value of the minority matrix above the exploration limit, but not so large that the participating majority users start getting mistaken for minorities, the majority collective is essentially able to artificially increase the power of the minority in the system such that recommendations improve and welfare increases. As we show in Appendix~\ref{app:model2}, these results hold beyond majority-minority matrices, even when the majority and minority groups are not disjoint in their preferences. Finally, we provide an algorithm that, with collaboration and information sharing, outputs an $\eta$ (or $0$ if the conditions are not satisfiable) for a given collective to use to guarantee improved recommendations.

\xhdr{Informal results: Platform.} In Section \ref{sec:learner}, we consider the perspective of the RecSys and analyze {\bf Q2} on the implications of user collective action for the platform. Notably, whenever the conditions we present in Section \ref{sec:users} can be satisfied, it means that forming a collective and strategically increasing interaction yields better results for users than simply interacting with items according to personal preferences. As such, in these cases, there is incentive for users to follow this collective rating strategy as long as the users of an effective collective care even slightly about the welfare of minority users. 
Naively, through a mechanism design lens, we may say that this is not a very good mechanism because there is a non-truthful strategy that dominates the truthful one. And through a social lens, recall that these movements are, for users in the real-world, instances of algorithmic resistance or protest to a system they perceive as suppressing content. While both of these intuitions suggest that the RecSys would be better off designing a different system that would not create such incentive for collective strategization, we take a different approach and prove that designing a new mechanism is not necessarily best for the learner.

First, we analyze a very natural learner utility function: one based in \emph{personalization accuracy}. That is, the learner is most accurate if it recommends the highest preferred item (according to personal interest) to each user and the utility function itself is simply the sum of the personal interest each user has for her respective recommendation. Notably, for such a learner, given all that matters is good recommendations, effective collective strategic rating improves the learner's welfare as it yields a Pareto improvement in recommendations. %

More interesting to consider is the ``engagement'' recommender. A common assumption in RecSys modeling and construction is that recommender maximizes \emph{engagement} (e.g., \cite{engagement_maxxing}). Engagement, like interaction, refers to the actions users have with items (e.g., likes, comments, views, etc). Intuitively, inducing greater engagement is good for a RecSys as it is analogous to having more space on which to sell ads. In our setting, a natural engagement-based utility is simply the sum of the absolute value of all \emph{realized} ratings. As such, an engagement-based recommender generally wants to inflate ratings (equivalently, the interactions users have) whether those interactions line up with user's true personal preferences or not. Thus, while the collective strategy is meant to be an act \emph{against} the RecSys in the eyes of the users, in this regime, it actually increases its welfare beyond the personal preference baseline.

\xhdr{Empirical validation.}
In \Cref{sec:experiment} we evaluate \textbf{Q1} and \textbf{Q2} on real data. Inspired by the real movements in book communities on TikTok where users have organized to increase their interactions with marginalized authors' content to fight algorithmic suppression~\cite{mccall_booktoks_2022}, we study collective action on the Goodreads user/book interaction dataset~\cite{goodreads1, goodreads2}. As a baseline, we do supervised fine-tuning on the Qwen2.5-1.5B-Instruct base model~\cite{qwen2025qwen25technicalreport} to train an LLM-based recommender that suggests which 2015 young-adult (YA) books users might be interested based on their true interactions with previous years' YA books. To simulate varying degrees of collective action, we train 3 additional models using the same method on strategically perturbed data. Choosing $80\%$ of users who are primarily interested in the $30$ most popular books to be our collective and $10$ niche book/author pairs to be the target of action, each collective user does the following: (1) adds one of the target books to their shelf, representing their boosting of the book directly; (2) adds an interaction with a pre-2015 book by $j$ target authors to represent boosting the authors. To create the three perturbed datasets on which we train the models representing collective action, we vary $j$ across $1, 3,$ and $7$. 

Our findings validate our theoretical results: despite not perfectly meeting the theoretical assumptions in the more realistic setting, the baseline algorithm primarily recommends popular books/authors to non-mainstream users who would instead prefer niche content. However, the the recommendations of the targeted niche users can be significantly improved by our simulated collective strategic interaction. Further, we find that, despite focusing on a very large collective, the model's overall accuracy is very minimally impacted.

\subsection{Related works} \label{sec:related_work}
Our work connects several diverse threads. Our focus on users' collective resistance to RecSys content suppression is inspired by the Human Computer Interaction (HCI) and Algorithmic Collective Action (ACA) literatures. Our model is informed by standard RecSys modeling via MC together with the broader paradigm of strategic learning. We outline each thread below.

\xhdr{HCI and Algorithmic Collective Action.} HCI researchers have long studied how people form mental models of algorithms---often described as \emph{folk theories}---to make sense of algorithmic systems~\citep{mental_models_chapter_2003, gelman_concepts_2011, eslami_facebook_folk_theory, self-present_folk}, and have conducted numerous interview-based studies about users' experiences with recommendation algorithms~\citep{salehi_youtube, lgbtq_tiktok, tfem_folk, queer_ad_targeting, omhc_folk, camille_tiktok_vendetta}. Across these studies and in a recent large-scale experiment \cite{cen2024measuringstrategizationrecommendationusers}, results indicate that people use folk theories to \emph{strategically} interact (e.g., via liking, commenting) with platform content to tailor their \emph{personal} recommendation feeds. Importantly, strategic interaction is not always purely self-interested. \citet{karizat_algo_resistance}, for example, document the folk theorization that the algorithm suppresses marginalized users' content, along with resulting user strategies intended to counteract such suppression. Similar behavior has also been reported in grassroots movements supporting BLM \cite{tiktok_blm} and in communities amplifying marginalized authors \cite{mccall_booktoks_2022}. Our theoretical model serves as a mathematical proof-of-concept for these behaviors and for the system-level dynamics that can arise from them.

An emerging, primarily theoretical, literature is also directly relevant to our interest in non-self-interested strategies: Algorithmic Collective Action (ACA). First introduced by \citet{hardt_aca}, ACA studies settings in which agents coordinate their actions to influence or manipulate a learning algorithm. Subsequent work has expanded the agent-side modeling assumptions, including multiple and heterogeneous collectives~\cite{karan_aca_two_groups} and richer (e.g., combinatorial) action spaces~\cite{sigg_decline_now}. Other work complicates the learner's problem by considering, e.g., differentially private learning~\cite{solanki2025crowdingnoisealgorithmiccollective} or distributionally robust optimization~\cite{bendov_aca_learning_algos}. Our contribution is to study ACA in the specific setting of RecSys, which to our knowledge has not yet been formalized. On the empirical side, \citet{baumann_aca_playlists} examine a closely related phenomenon by simulating collective playlist reordering to promote songs on platforms such as Spotify. Interpreting playlist position as an implicit rating signal, our model can be viewed as providing theoretical grounding for the dynamics observed in their simulations.

\xhdr{Matrix Completion and RecSys Modeling.} We abstract a class of collaborative filtering (CF) recommendation algorithms via online matrix completion MC. In CF algorithms, estimates of a user's preferences depend on their past ratings as well as the ratings of other users~\cite{schafer_collaborative_2007, koren_advances_2021, su_survey_2009}. Capturing this dependence is important for our setting, since we study how a collective's content-interaction strategies can affect \emph{others'} feeds. In MC, a matrix represents user-item preferences: after observing a subset of entries, the remaining unknown preferences are imputed using a low-rank approximation. For tractability, we model MC at a higher level of abstraction, but there is a large technical literature proving theoretical guarantees in this setting~\cite{candes_power_2010, candes_exact_2008, recht_simpler_2011, candes_matrix_2010, keshavan_matrix_2009, baby_online_2024, jain_online_2023}.

There is a substantial body of work modeling recommender systems through game-theoretic and economic lenses. Two papers in particular study strategic (i.e., self interested) users.~\citet{haupt_recommending_2023} and~\citet{cen_user_2024} model interactions between a single user and a RecSys as Stackelberg games, focusing respectively on unfair equilibria and on trustworthy algorithms. Our model differs by studying a multi-agent game and incentives beyond self-interest. Other work derives game-theoretic insights without explicitly modeling strategic users. \citet{reachability_dean} and \citet{guo_stereotyping_2021} study settings in which some items, despite being present in the system, would never be recommended to any user type. We obtain a similar outcome and show that it can be ameliorated by collective rating strategies. \citet{peng2023reconciling} analyze the accuracy-diversity trade-off as it relates to users' consumption constraints; in contrast, our results focus on top-1 recommendation under a fixed preference matrix and do not impose user cost constraints, so this trade-off does not arise for our learner. \citet{pca_fairness_liu} study an item-based fairness measure of recommender systems based in principal component analysis (PCA) as the limit of users and items goes to infinity, while we also study a dimensionality-reducing recommender, we focus on finite settings and user-based measures of welfare (more details in \Cref{app:connection_fair_pca}).

Another line of work focuses on welfare objectives or mechanism design from the learner's perspective. \citet{engagement_maxxing} study optimal information design to keep a user engaged with the platform; \citet{ben-porat_attrition_2022} model the learner as a multi-armed bandit seeking to avoid user attrition; and \citet{churn} maximize engagement to reduce churn. While we do not focus on the recommender's perspective as directly as these works, we briefly discuss engagement maximization as a candidate learner utility in Section~\ref{sec:learner}. Additionally, several papers analyze strategic \emph{content creators}, characterizing equilibria in content-creation marketplaces~\cite{supplyside, ben-porat_content_provider_2020}, modeling incentives for content quality~\cite{clickbait}, designing algorithms to incentivize high-quality content~\cite{hu2023incentivizing}, and designing fair and stable recommenders under creator strategization~\cite{ben-porat_game-theoretic_2018}. While content creators are a type of user, our focus is on interactions with \emph{existing} items.

Finally, our users' activist rating behavior in RecSys resembles public-spirited agents in voting settings~\cite{bailey_distort_ec, public_spirit3, public_sprit2}. As in that model, agents intentionally distort reports to benefit others; however, our setting differs fundamentally because each user receives an individualized recommendation (rather than a single collective outcome such as an election winner) and the strategy space is continuous. As a result, the public-spirited voting regime is not directly comparable to ours.

\xhdr{Strategic learning.}
Our work relates to the broad literature on strategic agents in learning and performative prediction~\cite{pp} (see \citet{podimata2025incentive} for a review). The classic strategic classification setting, first introduced by \citet{hardt_strategic_2016}, considers a learner who publishes a classifier to agents whose utilities depend on their resulting classifications. A large body of subsequent work takes a learner-centric perspective, developing algorithms that are robust to agent strategization, including incentive-aware learning~\cite{hardt_strategic_2016, levanon_strategic_2021, chen_learning_2020}, truthfulness~\cite{harris_strategic_2023}, and strategyproofness~\cite{ball_scoring_2024}.

Although we address learner welfare in Section~\ref{sec:learner}, the strand of strategic classification most relevant to our contribution adopts a user-centric perspective, particularly in its treatment of fairness. \citet{estornell_fair_strat_class} highlights that otherwise fair classifiers may not be fair in strategic settings. Several works study the unequal outcomes resulting from strategically robust algorithms: \citet{hu_disparate_2018} with misclassification rates, \citet{milli2019social} with manipulation effort, \citet{anticiptating_gaming} with true positive rates and acceptance rates, \citet{maggie_penn_strat_class} with ``aligned incentives'' to manipulate, and \cite{diana_minimax_fairness} with minimax group error. In general, all analyze learner-side interventions and their effects on outcomes across user groups. By contrast, we study a \emph{user}-side intervention: we show that activist users are incentivized to act in this way and that, under our model, no user is made worse off.

\section{Model \& Preliminaries}\label{sec:model}
In this section, we propose a tractable multi-agent game that captures the dynamics of the interplay between users and the real-world CF recommendation algorithms. We begin by providing some background for the reader on Online Matrix Completion (Section~\ref{sec:MC}), which is the classical paradigm at the heart of CF algorithms. Then, in Section~\ref{sec:abstraction}, we present a generalized abstraction of Online Matrix Completion, and leverage it to instantiate our model for the user-platform interactions.

\xhdr{Notation.} Matrices are capital and bolded (i.e., $\mathbf{X} \in \R^{m\times n}$), vectors are lower-case and bolded (i.e., $\mathbf{z}\in \R^{d}$), and one-dimensional variables are lower-case (i.e., $y \in \R$). Of a matrix, $\mathbf{X}$, the $i$th \emph{column} is $\mathbf{X}_i$ (an exception to lowercase vectors), the $j$th \emph{row} is $\mathbf{x}_j$, and the $j$th row, $i$th column element is $x_{j,i}$; $\sigma_k(\mathbf{X})$ is the $k$-th largest singular value. For matrix norms, $\|\mathbf{X}\|_1:= \max_{i \in [n]}\sum_{u \in [m]}r_{u,i}$ $\|\mathbf{X}\|_2:= \sigma_1(\mathbf{X})$ Sets are capital calligraphic letters (i.e., $\calU$), and their complement is $\calU^{C}$. To assist the reader, a notation table can be found in Appendix \ref{app:notation_table}.

\subsection{Background: Online Matrix Completion}\label{sec:MC}

Online Matrix Completion (MC) is the characterization of the standard MC setting as an online learning problem in which a recommender (equivalently, learner, or platform) learns $m$ users' preferences over $n$ items over a series of rounds~\cite{baby_online_2024, jain_online_2023}. 
Users' preferences are represented as an unknown $m\times n$ \emph{``ratings''} matrix, $\bR^\star\in \R^{m\times n}$, wherein each entry, $r^\star_{u,i}$ s.t. $u \in [m], i \in [n]$, is user $u$'s \emph{rating} of item $i$. We use ``rating'' to encompass the numerical formalization of all the interactions user $u$ has with item $i$ such as likes, comments, views, etc together with how much the user inherently likes an item. At each round $t$, the recommender must [randomly] select some collection of $(u,i)$ pairs whose matrix entries---the analogous collection of $r^\star_{u,i}$---will be revealed. In the language of RecSys, for every selected $(u,i)$, the recommender provides user $u$ item $i$ and sees user $u$'s rating of item $i$. Thus, as rounds continue, the recommender develops an increasingly completed ratings matrix whose completed entries match that of the unknown $\mathbf{R}^\star$. We elaborate on and illustrate this process using matrix notation in Appendix \ref{app:MC_illustration}.

The recommender cannot recommend \emph{every} item to \emph{every} user and complete the whole matrix $\bR^\star$; indeed, there are millions of users and billions of items available on modern platforms. Thus, the standard MC setting assumes that $\bR^\star$ is an unknown \emph{low-rank} matrix~\cite{nguyen_mc_survey, panish2025matrixcompletionsurveytheory, li2019surveymatrixcompletionperspective} and uses this low-rank property to estimate a $\widehat{\bR}$, such that $\widehat{\bR}$ retains all the known entries. Then  $\widehat{\bR}$ can be used to estimate what the top item(s) are for each user, without querying every single user-item rating. This low-rank estimation with known ratings is what gives online MC its collaborative filtering properties; a user's \emph{estimated} ratings are dependent not only on their own known ratings, but also on the ratings of everyone else. To formalize how matrix $\widehat{\bR}$ is calculated, we define the set of $(u,i)$ pairs that have been observed by the learner, $\Omega$:
\begin{definition}[Observed pairs, $\Omega$] $\Omega :=\{(u,i) \in [m] \times [n]: r^\star_{u,i}\text{ has been observed by the learner}\}.$
\end{definition}
The recommender finds the lowest rank matrix that constrains itself to the known entries, which is only a linear constraint. We say that $\widehat{\bR}$ is a solution to the following optimization problem\footnote{Computationally, there are various convex approximations of the objective that are used in the literature, such as nuclear norm~\cite{candes_exact_2008, recht_simpler_2011, candes_power_2010} or max-norm~\cite{fang_max-norm_2018, cai_max-norm_2013}, but we analyze the rank minimization version to avoid the specificity of an approximation.}:
\begin{equation} \label{eqn:rank_min}
    \begin{aligned}
        \text{minimize}_{\mathbf{X}\in \R^{m\times n}} &\quad \mathrm{rank}(\mathbf{X}) \\
        \text{subject to} &\quad x_{u,i} = r^\star_{u,i} \quad \forall (u,i) \in \Omega.
    \end{aligned}
\end{equation}
We have now established how, given some amount of exploration, an estimated preference matrix, $\widehat{\bR}$, is calculated, but we have purposefully eluded the specifics of how $(u,i)$ are randomly selected in each round (e.g., how many, with or without replacement, etc.), and the stopping condition at which exploration ends and $\widehat{\bR}$ is computed. All these aspects can vary across the literature and we want to avoid this level of specificity for our abstraction. Instead, we point out a general observation about rank-minimization that only relies on an assumption about the observed ratings, $\Omega$. When there exists a minority group whose preferences are significantly distinct from the mainstream, the group's preferences may be ``harder to learn'' and thus, due to randomness or poor exploration, these users may be poorly served. Formally, we prove this statement using \emph{majority-minority} preferences.

\begin{definition}[Majority-minority matrix]\label{def:majority-minority-matrix}
    A matrix $\bR \in \R_{\geq 0}^{m\times n}$ is a \emph{majority-minority matrix} if there exists a partition of users $\calUmaj \cup \calUmin = [m]$ (where $\calUmaj, \calUmin \subseteq [m]$) and a partition of items $\calImaj \cup \calImin = [n]$ (where $\calImaj, \calImin \subseteq [m]$) such that if $u \in \calUmaj$ and $i \in \calImin$ or  $u' \in \calUmin$ and $i' \in \calImaj$, then $r_{u,i}, r_{u'i'}=0$. Further, no user has ratings of all $0$s (i.e., $\forall u \in [m]$, $\sum_{i \in [n]}r_{u,i}>0$). 
\end{definition}

Put simply, a majority-minority preference structure is one in which it is possible, through re-ordering of users and items, to form the matrix into a 2-block-diagonal structure. Note that the items in $\calImin$ are not necessarily liked by \emph{all} minority users. Now under the condition that some key elements (of the minority group's preferences) have not been queried, minority preferences are potentially\footnote{We focus on the \emph{sparsest} solution to \eqref{eqn:rank_min}, though there may be many solutions (Appendix \ref{app:less_sparse}), all solutions can be reduced to the sparser one of ~\Cref{thm:rank_min_bad} and sparser matrices are desirable for storage and computational reasons.} not learned at all. This is similar to the concept of matrix \emph{coherence} (see~\Cref{app:MC_coherence} for discussion), though this is not necessary to understand for our work. %

\begin{restatable}[Estimated Minority Item Ratings are Zero]{proposition}{RankMinBad}
\label{thm:rank_min_bad}
Assume that $\bf{R}^\star$ satisfies Definition~\ref{def:majority-minority-matrix}, and
that for any $(u,i) \in \Omega$ s.t. $u \in \calUmin$ and $i \in \calImin$, $r^\star_{u,i} = 0$.
Then, the sparsest $\widehat{\bR}$ solving Equation~\eqref{eqn:rank_min} is such that
$\hat{r}_{u,i} = 0$ for all $u\in [m]$ and $i \in \calImin$.
\end{restatable}

The proof is simple and comes from the linear independence between majority and minority user groups; we defer its formal statement to Appendix \ref{app:MC_proofs}. The intuitive meaning of Proposition~\ref{thm:rank_min_bad} is much like what the collective movement of users we model attempts to combat: if the estimated preference matrix leaves as zero all the minority users' actual favorite items, then these will not be recommended to any users in the future, thus creating a problem of suppressed minority content. 

In this stated form, Proposition~\ref{thm:rank_min_bad} relies on a lack of exploration and a discrete minority group, both of which may seem overly strong at first glance. On the first assumption, notice that if the minority group is relatively small, in a randomly ordered version of the matrix, it is not unreasonable to imagine the key elements are unexplored. Examples of this in Appendix \ref{app:MC_ass} highlight that this could be as few as 2 of 100 elements. For the second assumption, we chose to state the Proposition in this strong assumption format for cleanliness and to avoid probabilistic statements that would require commitment to specific random selection processes of how entries are queried at each round. More generally, it is easy to see that---in even looser conditions---the same (or similar) outcomes can happen. For example, if the groups' preferences are not completely disjoint, but some users of each group only have sparse preferences for the other group's preferred items, then if those ratings are randomly missed, the same result arises. We also illustrate this in Appendix~\ref{app:MC_ass}.

To summarize, in an online MC recommendation setting, when a minority group's preferences differ from the mainstream preferences, minority preferences can be harder to learn, thus leading to a suppression of the minority content and poor recommendations. This is a direct byproduct of the minority group being small, but different from the majority in a way that the recommender ignores. The goal of our paper is to study how collective \emph{intentional} interaction with the minority items could alleviate this. Intuitively, in the MC setting, one could see this being effective as, if the collective of majority users is large enough, it breaks linear independence between the minority and majority with high probability and no longer makes setting all estimated minority item ratings to 0, an obviously rank-minimizing choice. Unfortunately, exact rank-minimization is NP-hard and the online nature of the setting introduces a repeated game dynamic that is difficult to analyze, especially with many agents. Thus, to answer the questions posed about the real-world collective algorithmic protests in Section \ref{sec:intro}, we present an abstracted model that collapses the sequential MC game into a one-round low-rank approximation. In the following section, we present this tractable model, highlighting the connection to the online MC version. The tractable model yields the same negative outcomes as Proposition~\ref{thm:rank_min_bad}, but enables us to study user strategy without the NP-hardness of an exact rank-minimization objective and choice of specific MC algorithm.

\subsection{Abstraction}\label{sec:abstraction}
\subsubsection{Model summary} To get around the intractabilities of NP-hardness and repeated game analysis in standard online MC, we instead analyze an analogous one-round low-rank-minimization-based multi-agent game. In particular, our model is built upon a transformation across three matrices: $\bR^\star \rightarrow \tbR \rightarrow \widehat{\bR}$. Like before, the unknown true personal preference matrix is $\bR^\star \in \R^{m\times n}_{\geq 0}$. To parallel the evolving partially complete matrices of online MC but in a tractable way, we use $\tbR \in \R^{m \times n}_{\geq 0}$, a single complete \emph{revealed preference matrix}, which is then transformed via low-rank approximation into $\widehat{\bR}$, which represents the information a learner could have reasonably learned from the full $\tbR$. %
 In the remainder of the section, we detail the full process of $\bR^\star \rightarrow \tbR \rightarrow \widehat{\bR}$.

\subsubsection{Learner's protocol (Protocol~\ref{prot:learner})}
The learner (it) wants to recommend each agent (she, when referred to individually) her top item\footnote{We focus on top-1 recommendation for simplicity of exposition. Our results go through for top-$k$ recommendations too (for a range of $k$ values); we present a longer discussion in Section \ref{sec:discussion}.} according to her unknown personal preferences. We construct the tractable abstracted protocol of online MC by collapsing multiple rounds of interaction into a \emph{low-rank approximation} on a complete revealed preference matrix $\tbR$, rather than \emph{rank-minimization} constrained to a randomly partially completed matrix formed over the multiple rounds. $\widehat{\bR}$ is this low-rank approximation of $\tbR$. The learner never directly sees $\tbR$, as $\tbR$ reflects all $m\times n$ reported ratings, which would be unrealistic to have access to. Instead, the learner uses the low-rank $\widehat{\bR}$, which represents the estimated preference matrix it could have learned from repeated rounds of agents' item interactions, to provide top-1 recommendations to each user. %

\begin{protocol}[t!]
\caption{Learner's protocol}\label{prot:learner}
\begin{algorithmic}
\State \textsc{Learning phase:}
\State \hspace{10pt} Learner has access to $\widehat{\mathbf{R}}$, the $k^\star$-truncated SVD of $\tbR$. \Comment{``learned'' preferences}
\State \textsc{Recommendation phase:}
\State \hspace{10pt} Learner shows agents $u \in [m]$ their top item $\Topk(u) \in [n]$ according to $\widehat{\mathbf{R}}$. 
\end{algorithmic}
\end{protocol}

\xhdr{Learning phase.}
The recommender gets $\widehat{\bR}$, a reduced information version of the complete (potentially strategically manipulated) ratings of $\tbR$. Specifically, $\widehat{\bR}$ is a $\kst$-truncated Singular Value Decomposition (SVD), which means it consists of the $\kst$ heaviest weight (in terms of singular value) principal components of $\tbR$. Formally, $\widehat{\mathbf{R}} = \sum_{j \in [\kst]}\sigma_j\mathbf{u_j}\mathbf{v_j}^\top$ where $\tbR = \sum_{j \in [\mathrm{rank}(\tbR)]}\sigma_j\mathbf{u_j}\mathbf{v_j}^\top$. As a result, $\widehat{\bR}$ is the best rank-$\kst$ approximation of the revealed preference matrix. 

\emph{What is $k^\star$?} In MC, $k^\star$ is the rank of the estimated complete preference matrix after a number of rating exploration rounds. Ideally, $k^\star$ should be such that $\widehat{\bR}$ represents the variation of preferences well, while not requiring too many queries. In practice, there are various exploration-stopping conditions an algorithm might use. We model this process as a learner that has a marginal amount of rating ``variation'' it can afford to lose. Formally, we define $\mathrm{TVR}$ as the measure of this variation:
\begin{definition}%
The total variation retained ($\mathrm{TVR}$) when doing a $k^\star$-truncated SVD on $\tbR$ is: 
\begin{equation*}
    \mathrm{TVR}(k^\star):= \frac{\sum_{j \in [k^\star]}\sigma_j(\tbR)}{\sum_{j \in [\mathrm{rank}(\bR)]}\sigma_j(\tbR)}
\end{equation*}
\end{definition}
$\mathrm{TVR}$ is essentially the ``explained variance ratio''~\cite{pedregosa2011scikit} or Cummulative Percentage of Total Variation~\cite{joliffe_principal_2002}, which is used when doing Principal Component Analysis (PCA). Intuitively, $\mathrm{TVR}$ is simply a measure of the proportion of $\tbR$'s information that $\widehat{\mathbf{R}}$ keeps in the low-rank approximation process. To represent the process of exploring up to a point of sufficient information learning, we define the \emph{$\alpha$-loss tolerant learner}.

\begin{definition}[$\alpha$-loss tolerant learner]\label{alpha_learner}
An $\alpha$-loss tolerant learner has access to $\widehat{\bR}$ of the minimum rank, $k^\star$, such that $\mathrm{TVR}(\kst+1) - \mathrm{TVR}(\kst) \leq \alpha \cdot (\sum_{j \in [\mathrm{rank}(\tbR)]} \sigma_j(\tbR))^{-1}$. Equivalently,
    \begin{center}
        $\min_{k\in [\mathrm{rank}(\tbR)]} k \quad 
        \text{s.t., } \quad \sigma_{k+1}(\tbR) \leq \alpha$.
    \end{center}
\end{definition}

The $\alpha$-loss tolerant learner is referred to as such because it is \emph{tolerant} of losing this much information in the rank reduction. An $\alpha$-loss tolerant learner has the lowest rank $\widehat{\bR}$, such that increasing rank by 1 will not improve the proportion of information, $\mathrm{TVR}$, enough such that it is worthwhile to increase the dimension of the data. Thus, we refer to $\alpha$ as the ``\emph{exploration limit}'' because it represents the stopping condition at which the learner cannot afford to keep exploring to better learn preferences; if a learner is tolerant of experiencing only a small amount of loss, it must be willing to sacrifice time and/or compute to achieve this. Since $\alpha$ represents learning ability/capacity, we generally assume there could exist learners across the full spectrum of extremely strong computation/time constraints (high $\alpha$) to weak computation/time constraints (low $\alpha$).

\xhdr{Recommendation phase.}
Using $\widehat{\bR}$, the learner estimates the top item(s) for each user, $\topitems(u, \alpha):= \mathrm{arg}\max_{i \in [n]} \hat{r}_{u,i}$. $\topitems(u, \alpha)$ may be a set. We assume that the learner breaks ties in favor of the \emph{most popular} of the top-rated items $\topitems^{\pop}(u, \alpha)$, i.e.,
\begin{center}
        $\Topk(u, \alpha) \sim \textrm{Unif}(\topitems^{\pop}(u, \alpha)), \quad \text{where} \quad \topitems^{\pop}(u, \alpha):= \arg\max_{i \in \topitems(u,\alpha)} \quad \|\hR_i\|_1$
\end{center}
Note that each of these top item(s)-related functions, depends on $\alpha$ via the use of $\widehat{\bR}$. When using top items-related notation, we omit the dependence on $\alpha$ whenever clear from context.

We use the \emph{Social Welfare} as the measure of how good recommendations are overall:
    \begin{center}
        $\SW(\tbR, \alpha) = \sum_{u \in [m]}{r^\star_{u,\Topk(u)}}$.
    \end{center}

\subsubsection{Users' agency}
As with online MC, we focus primarily on \emph{majority-minority} personal preference matrices. 
While the main body of this paper focuses on these discrete user groups, in Appendix~\ref{app:model2}, we present analogous results under a more complex class of non-exclusive preference matrices; see also~\Cref{sec:discussion} for further discussion. 
We refer to majority-minority matrices in the block-diagonal ordering; this is without loss of generality because:

\begin{restatable}[]{observation}{PermInv}
\label{obs:permutation_invariance}
The welfare of each user is invariant under any re-ordering of users and items.
\end{restatable}

The proof of this observation is in Appendix \ref{app:abstraction_proofs}. We order the rows and columns of the true preference matrices such that for some 
$\bar{m} \in [m]$ and  $\bar{n} \in [n]$, $\calUmaj = [\bar{m}]$ and $\calImaj = [\bar{n}]$. Therefore, $\bR^\star$ is a block-diagonal matrix where the blocks are:
$\bR_{\maj}^\star \in \R_{\geq 0}^{\bar{m}\times \bar{n}}$ and $\bR_{\minor}^\star \in \R_{\geq 0}^{(m-\bar{m}) \times (n-\bar{n})}$. Because we analyze the collective combating of minority suppression, we add mathematical formalization of what it means for a group to be ``majority'' or ``minority'':
\begin{assumption}[Singular Value Gap]\label{assumption:singular_value_gap} Let $\bR$ be a majority-minority matrix, $k_{\maj}=\mathrm{rank}(\bR_{\maj})$, and $\mathcal{G}(\bR):= (\sigma_{1}(\bR_{\minor}), \sigma_{k_{\maj}}(\bR_{\maj}))$.
If $\mathcal{G}(\bR) \neq \emptyset$, then we say that $\bR$ has a \emph{singular value gap}; we assume that our majority-minority matrices has a singular value gap. 
\end{assumption}
Roughly speaking, this assumption corresponds to the majority/minority sub-matrices (and thus preferences) being more/less dominant in the system, respectively; see Appendix~\ref{app:maj_min_ex} for an example. When the learner is $\alpha$-loss tolerant such that $\alpha \in \mathcal{G}(\tbR)$ (i.e., the exploration limit, $\alpha$, falls between the information values of the majority and minority groups), $k^\star$, the rank of the learned preference matrix is \emph{exactly equal} to the rank of the majority matrix (full proof in Appendix~\ref{app:rhat_nbar_proof}):

\begin{restatable}[]{proposition}{RHatNBar}
\label{rhat_equals_nbar1}
For any revealed majority-minority preference matrices $\tbR$ with a singular value gap and any $\alpha$-loss tolerant learner where $\alpha \in \mathcal{G}(\tbR)$, it must be the case that $k^\star = k_{\maj}$.
\end{restatable}

\xhdr{Strategic rating \& paper's approach.} In our model of the RecSys, while users have true personal preferences, $\bR^\star$, that determine how much they would rate (i.e., interact with) an item to personal taste in the moment, this does not mean they must actually rate the item this much. Our users have one primary action space which they may use to communicate with the recommender: rating items via $\tbR$. To analyze the phenomenon of users' protest movements against algorithmic suppression, in this work we compare the social welfare in two cases: (1) agents truthfully interact with content according to their [personal] interest each item: $\tilde{r}_{u,i} = r^\star_{u,i}$; and (2) a collective of agents purposefully increase interactions with minority items (despite not personally liking the content). 
\section{When are collectives effective?}\label{sec:users}
In this section, we use our model abstraction to study when strategic collective action effectively helps minority users. First, we show that if (1) majority users prefer popular items, (2) minority users prefer niche items, and (3) the RecSys has limited exploration, the system learns too little about minority preferences, yielding good majority recommendations but poor minority ones, and low social welfare. We then analyze collective uprating strategies where majority users uprate minority-preferred items, present sufficient conditions for Pareto improvement in recommendations and strict welfare improvement, and provide a robust algorithm to find such strategies.

\subsection{Poor minority recommendation without strategic interactions} \label{sec:poor_minority}

Intuitively, in majority-minority preference matrices (\Cref{def:majority-minority-matrix}) with a gap in singular values between the two groups (Assumption \ref{assumption:singular_value_gap}), the majority is more ``important'' because its higher singular values indicate that the majority matrix block is \emph{more informative}. Additionally, because the groups' preferences are disjoint, each majority user's preference vector is orthogonal to all minority users' preference vectors; i.e., majority and minority users' preferences live in different subspaces. Without any strategic collective action, an $\alpha$-loss tolerant learner may reduce out the subspace of the minority user preferences, while perfectly capturing the majority user preferences.
\begin{restatable}[Truthfulness is good for majority, bad for minority]{theorem}{GMajBMin}\label{g_maj_b_min2}
    Let the ground truth preference matrix $\bR^\star$ be a majority-minority matrix satisfying Assumption~\ref{assumption:singular_value_gap} and $\bR^{\star} = \tbR$ (i.e., simple personal interest rating)%
    . If $\alpha$ is in the singular value gap (i.e., $\alpha \in \mathcal{G}(\bR^\star)$), then all majority users are accurately given their top item, while minority users are given popular items they do not like. Formally:
   \begin{center} 
        $r_{u,\Topk(u,\alpha)}^\star = \begin{cases}\max_{i \in [n]}r_{u,i}^\star & u \in \calUmaj \\
        0 & u \in \calUmin
        \end{cases}
        ,\quad \Rightarrow \quad 
        \SW(\bR^\star,\alpha) = \sum_{u \in \calUmaj}\argmax_{i \in [n]}r_{u,i}^\star$
    \end{center}
\end{restatable}
The proof uses the fact that the learner has access only to some lower dimensional version of preferences, and because the minority preferences are both linearly independent from the majority and make up a smaller portion of the information in the complete \emph{revealed} preference matrix, they are the ``low-hanging fruit'' to rank-reduce out. We defer the full proof to Appendix \ref{app:users_proofs_gmajbmin}.

\subsection{Improving recommendations via simple collective rating strategies} \label{sec:effective_collective}

\Cref{g_maj_b_min2} highlights that in general for $\alpha$-loss tolerant learners against users that are nonstrategic and simply interact with content according to personal interest, information about minority preferences is not retained by the learner. Can collaborative rating distortion force the learner to retain minority information? We focus on aiding ``picky'' minority users, whose preferences are \emph{especially} hard to learn for the recommender; picky users only like a specific item, meaning a \emph{single} rating in the matrix must be queried to avoid the pitfalls of \Cref{thm:rank_min_bad}.

\begin{definition}[Picky Users]\label{defn:picky_items}
    We say that item $i^\star$ is a picky item with picky user group $\calU_{i^\star}$ if
    \begin{center}
    $r_{u,i^\star}>0 \iff u \in \calU_{i^\star}, \forall u \in [m]\quad$ and
    $\quad r_{u,i}=0$ $\forall u \in \calU_{i^\star}$, $\forall i \neq i^\star$.
    \end{center}
\end{definition}

Because picky users are especially hard to learn, they are particularly in danger of being poorly served by the algorithm, thus we suppose collectives of majority users seeking to help those in the minority may particularly target these groups. Our goal is to model collectives of majority users \emph{uprating} a picky minority item by a collaboratively selected amount, $\eta>0$. %

\subsubsection{Collective uprating} Let $\bR^\star$ be a majority-minority matrix ordered into block-diagonal form where the first $\barm<m$ users are the majority users and the first $\barn < n $ items are the majority items. $\calU_\COLL \subseteq \calUmaj$ is a nonempty subset of majority users who form the collective. All users $u \in \calU_\COLL$ report $\eta \in \R_{>0}$ (instead of $0$, which is their personal interest in the item) for a target picky item $i^\star > \barn$, while others continue to interact according to their personal interest. At a high level, this modeling captures some majority users explicitly choosing to interact more with minority content to ``boost'' it. This yields a strategically manipulated matrix, $\tbR$, instead of the true $\bfR^\star$; $\tbR$ is the same as $\bR^\star$ except for elements indexed $(u, i^\star), \forall u \in \calU_\COLL$. We present such a matrix below, with {\color{red}red} denoting the uprated items from the collective.
%
%
%

\begin{equation}
\bR^\star = 
\begin{bmatrix}
& & & 0& 0 & \hdots & 0\\
& &  & 0& 0 & \hdots & 0\\
& & & \vdots & \vdots & \vdots & \vdots\\
& \bfR_{\maj}^\star &  & 0 &0& \hdots & 0\\
& & & \vdots & \vdots & \vdots& \vdots\\
&  & & 0 & 0 & \hdots & 0\\
&  &  & 0 & 0 & \hdots & 0\\
0 & \hdots & 0 & &&\\
\vdots & \vdots & \vdots & &&\bfR_{\minor}^\star\\
0 & \hdots & 0 & &&
\end{bmatrix}
\quad
\rightarrow
\quad
\tbR = 
\begin{bmatrix}
& & & {\color{red}\eta}& 0 & \hdots & 0\\
& &  & 0& 0 & \hdots & 0\\
& & & \vdots & \vdots & \vdots & \vdots\\
& \bfR_{\maj}^\star &  & 0 &0& \hdots & 0\\
& & & \vdots & \vdots & \vdots& \vdots\\
&  & & 0 & 0 & \hdots & 0\\
&  &  & {\color{red}\eta} & 0 & \hdots & 0\\
0 & \hdots & 0 & &&\\
\vdots & \vdots & \vdots & &&\bfR_{\minor}^\star\\
0 & \hdots & 0 & &&
\end{bmatrix}
\end{equation}

This poses a natural question: when are these collective strategies effective at \emph{actually} improving recommendations? We derive sufficient conditions on the collective of users, $\calU_\COLL$, and the amount of uprating, $\eta$, such that all majority users \emph{and} picky users receive correct recommendations from the $\alpha$-loss tolerant learner. To do so, we must define another useful singular value gap.%
\begin{definition}[$(\eta, \calU_\COLL)$-Sufficient Singular Value Gap]
    For a given $\eta \in \R_{>0}$, $\calU_\COLL\subseteq \calUmaj$, and majority-minority preference matrix, $\bR$, define the following space, $\mathcal{G}(\bR,\calU_\COLL, \eta)$:

\begin{center}
    $\mathcal{G}(\bR, \calU_\COLL, \eta):= \left(\sigma_1(\bR_{\minor}), \sqrt{\min \{\sigma_{k_{\maj}}(\bR_{\maj})^2, \eta^2|\calU_\COLL|+\|\bR_{i^\star}\|_{2}^2\}-\eta\sqrt{\bar{n}}\AV(\bR, \calU_{\COLL})}\right)$
\end{center}
where %
$\AV(\bR, \calU_{\COLL}) = \max_{i \in [\bar{n}]}\sum_{u \in \calU_\COLL}r_{u,i}$ is the largest aggregate value (i.e., rating) of a popular item for users in the collective. 
\end{definition}

Much like how the simpler singular value gap, $\mathcal{G}(\bR)$, is the gap between the largest singular value of the minority matrix and the smallest (nonzero) singular value of the majority matrix, $\mathcal{G}(\bR, \calU_\COLL, \eta)$ is just an approximation of the gap between the largest nonzero singular gap of the minority matrix and the smallest (nonzero) singular value of a new majority matrix defined by considering the picky user(s) and the strategically rated picky item as a part of the majority matrix. This (approximated) gap is important in defining sufficient conditions on said strategic ratings. Notice that the upper bound in the right-hand term of $\mathcal{G}(\bR, \calU_\COLL, \eta)$ is less than or equal to simply $\sigma_{k_{\maj}}$, which is the upper bound in the right-hand term of $\mathcal{G}(\bR)$. Thus we see that $\mathcal{G}(\bR, \eta, \calU_\COLL) \subseteq \mathcal{G}(\bR)$, and so given that $\alpha$ is already in the vanilla singular value gap (the assumption of \Cref{g_maj_b_min2}), as long as $\eta$ is collaboratively selected such that the right-hand term of $\mathcal{G}(\bR, \calU_\COLL, \eta)$ gets big enough, $\alpha$ is also in the singular value gap of a new majority-minority matrix where the picky minority users become majority and can benefit from good recommendations. The following Theorem states formally that, under these conditions, collective uprating strategies allow all majority users' \emph{and} a number of minority users' social welfare to be maximized.

\begin{restatable}[Improving Recommendations]{theorem}{SWImprove}\label{thm:SW_EA_Model1}
    Let $\bR^\star$ be a majority-minority matrix with a picky item $i^\star>\bar{n}$. Suppose there is a collective uprating to value $\eta$ with collective $\calU_\COLL$ s.t., $\bR^\star$ gets a $(\eta, \calU_\COLL)$-sufficient singular value gap. If $\eta < \min_{u \in \calUmaj}\max_{i \in [n]}r_{u,i}^\star$ and $\alpha \in \mathcal{G}(\bR^\star, \calU_\COLL, \eta)$, then
    \begin{center}$
    \Topk(u,\alpha) \in
    \begin{cases} 
    \arg\max_{i \in [n]} r_{u,i}^\star & u \in \calUmaj \cup \calU_{i^\star}\\
    [\barn +1] & u \in \calUmin\setminus \calU_{i^\star}
    \end{cases}
    ,\quad \Rightarrow \quad
    \SW(\tbR, \alpha) = \sum_{u \in (\calUmaj \cup \calU_{i^\star})}\max_{i \in [n]}r_{u,i}^\star$
    \end{center}
\end{restatable}

The intuition for this proof is exactly the singular value gap approximation described above. In Appendix \ref{app:users_proofs_sw_improve} we provide details on how to prove that $\mathcal{G}(\bR^\star, \calU_\COLL, \eta)$ is a new singular value gap formally using known lower bounds on the singular values of a matrix after appending a column. Notably, the outcome of \Cref{thm:SW_EA_Model1} is clearly a Pareto improvement in recommendation and a strict social welfare improvement.

Because \Cref{thm:SW_EA_Model1} may be difficult to use when checking whether an uprating strategy is guaranteed to be effective or not, it is useful to restate as more explicit sufficient conditions that $\eta$, the uprating value, and $\calU_{\COLL}$, the strategic user collective, must satisfy:
\begin{restatable}{corollary}{SuffCond}\label{cor:suff_cond} Let $\bR^{\star}$ be a majority-minority matrix with a picky item $i^\star>\bar{n}$, $\alpha > \sigma_{1}(\bR_{\minor}^\star)$, and $\kappa := \min_{u \in \calUmaj}\max_{i \in [n]}r_{u,i}^\star$. Then, a collective uprating improves social welfare if the following hold:
    \begin{center}
    $0<\eta < \kappa
    ,\quad
    \alpha < \sqrt{\min \{\sigma_{k_{\maj}}(\bR_{\maj}^\star)^2, \eta^2|\calU_\COLL|+\|\bR^\star_{i^\star}\|^2_2\}-\eta\sqrt{\bar{n}}\AV(\bR^\star, \calU_\COLL)}$.
    \end{center}
\end{restatable}

While this means that for any $\alpha$-loss tolerant learner one may check if a particular collective action will be genuinely effective, it is perhaps more enticing to consider how agents may find an effective collective. Fortunately, \Cref{cor:suff_cond} enables the construction of an effective $\eta$ finder.

\subsection{Algorithms to find effective collective strategies} \label{sec:coll_algo}
Using \Cref{cor:suff_cond}, we construct Algorithm \ref{alg:find_eta_aglo}, which returns an effective $\eta$ using only arithmetic operations, suggesting that, given sufficient information sharing, it is computationally reasonable that users find effective strategies. Algorithm 1 follows naturally from the sufficient conditions on $\eta$. In particular, the first condition $(0 < \eta < \kappa)$ is easy to handle since, given a $\kappa$, we may just choose any positive value less than it. Making sure the suggested $\eta$ satisfies the second condition is harder. Let us fix the collective, $\calU_\COLL$, and break down this condition into two. One of these (the easier to handle) is of the form $\eta\in \{\eta: \eta< h(\mathbf{z})\}$ (where $\bz$ is a vector of the algorithm's parameters). The other can be stated in the form: $\eta \in \{\eta: g(\eta; \bz)> 0\}$ where $g$ is quadratic in terms of $\eta$. This space can be computed using the discriminant of $g$, which we calculate in \Cref{alg:find_eta_aglo} line 2. The rest of the algorithm simply outputs an $\eta$ in the middle of all these intersecting spaces, or $0$ if such an intersection does not exist. We defer the formal proof of~\Cref{thm:find_eta_algo} to Appendix \ref{app:users_proofs_find_eta}.
\begin{restatable}[Algorithm \ref{alg:find_eta_aglo} returns an effective $\eta$]{theorem}{FindEta} \label{thm:find_eta_algo}
Let $\bfR^\star$ be a majority-minority matrix satisfying Assumption \ref{assumption:singular_value_gap} with a picky item at index $i^\star> \barn $ and $\alpha \in \mathcal{G}(\bR^\star)$. Then Algorithm \ref{alg:find_eta_aglo}, using $(\sigma_{k_{\maj}}(\bR_{\maj}^\star)$, $\alpha$, $\bar{n}$, $\|\bR^\star_{i^\star}\|^2_2$, $\AV(\bR^\star, \calU_\COLL)$, $\kappa$, $|\calU_\COLL|)$ as parameters, returns either:
\begin{enumerate}
    \item $\eta \in \R_{>0}$, such that social welfare is improved if all $u \in \calU_\COLL$ uprate $i^\star$ by $\eta$.
    \item $0$, if and only if there is no $\eta$ correlated strategy that satisfies our sufficient conditions.
\end{enumerate}
\end{restatable}
\begin{algorithm}[htbp]
\caption{Effective $\eta$ Finder}\label{alg:find_eta_aglo}
\begin{algorithmic}[1]
\Require $\sigma_{k_{\maj}}(\bR_{\maj}^\star)$, $\alpha$, $\bar{n}$, $\|\bR^\star_{i^\star}\|^2_2$, $\AV(\bR^\star, \calU_\COLL)$, $\kappa$, $|\calU_\COLL|$
\State $N_\up \gets \min\bigg\{ \left(\sigma_{k_{\maj}}(\bR_{\maj}^\star)^2-\alpha^2\right)\left(\sqrt{\bar{n}}\cdot\AV(\bR^\star, \calU_\COLL)\right)^{-1}, \kappa \bigg\}$ 
\Comment{upper bound on feasible $\eta$}
\State $d \gets \bar{n}[\AV(\bR^\star, \calU_\COLL)]^2+4|\calU_\COLL|(\alpha^2-\|\bR^\star_{i^\star}\|^2_2)$ \Comment{find discriminant}
\If{$d < 0$}
\State $N_\lo \gets N_{\up}/2$ \Comment{no real $\eta$ lower bound exists}
\ElsIf{$d \geq 0$}
\State $N_\lo\gets \left(\sqrt{\bar{n}} \cdot \AV(\bR^\star, \calU_\COLL)+\sqrt{d}\right)\left(2|\calU_\COLL|\right)^{-1}$ \Comment{lower bound on feasible $\eta$}
\EndIf
\If{$N_\lo < N_\up$}
\Return $(N_\lo+N_\up)/2$ \Comment{return if exists feasible $\eta$}
\EndIf
\State $N_\up \gets \min\bigg\{\left(\sqrt{\bar{n}}\AV(\bR^\star, \calU_\COLL)-\sqrt{d}\right)\left(2|\calU_\COLL|\right)^{-1},N_\up\bigg\}$ \Comment{new upper bound on feasible $\eta$}
\If{$N_\up > 0$}
    \Return ${N_\up}/{2}$ \Comment{return an $\eta$ if upper bound $> 0$}
\EndIf
\Return $0$ \Comment{sufficient conditions can't be satisfied}
\end{algorithmic}
\end{algorithm}

%

In order for the collective to compute an effective $\eta$, the algorithm requires knowledge of several key quantities (e.g., characteristics of majority and minority preferences etc). While perfect information about $\bR^\star$ would suffice, this level of information may sometimes be unrealistic to have in real life. However, reasonable estimates of these parameters---obtained through partial preference knowledge and platform experience (as suggested by users' algorithmic folk theories in HCI literature)---may suffice. Since our primary focus is analyzing the effectiveness of collective strategies rather than the precise collaboration mechanisms, and our experiments in Section \ref{sec:experiment} show that simpler, less coordinated collaborative strategies can work in realistic systems, we defer extensive analysis of data-sharing mechanisms to future work. Instead, we provide robustness results showing that $\eta$ remains effective even with imperfect information and parameter misspecification.

\subsubsection{Algorithm \ref{alg:find_eta_aglo} robustness}\label{sec:robust}
How far off can the estimate of the required parameters be from the true setting parameters such that \Cref{alg:find_eta_aglo}'s suggested collaborative uprating is still guaranteed to improve social welfare and recommendations?

Recall from the sufficient conditions of \Cref{cor:suff_cond}, there is a \emph{set} of $\eta$ that satisfy the conditions under the true parameters. However, taking a closer look at the algorithm, given it is not zero, the $\eta$ outputted by a particular instantiation of parameters is simply a value within this sufficient space. If sufficient $\eta$ spaces of different parameterizations overlap, then the output of the algorithm under one set of parameters may also be effective for the other parametrization. We formalize this notion in the following Theorem, which provides conditions such that the output of Algorithm \ref{alg:find_eta_aglo} still satisfies the true sufficient conditions while using potentially incorrect estimates of \emph{many parameters at once}.
\begin{restatable} [Robustness of Algorithm \ref{alg:find_eta_aglo} to misspecifications]{theorem}{RobustAlgo}\label{thm:perturbed_params}
    Define the vector of required \Cref{alg:find_eta_aglo} parameters as $\bz:= (\sigma_{k_{maj}}(\bR_{\maj}), \alpha, \bar{n}, \|\bR_{i^\star}\|_2^2, \AV(\bR, \calU_\COLL), |\calU_\COLL|)$ and two functions:
    \begin{enumerate}
        \item A function, $f$, of the algorithm parameters and parameterized by $\eta$, 
        \begin{center}
        $f(\bz;\eta):= \min(\sigma^2, \eta^2 |\calU_\EA|+\|\bR_{i^\star}\|_2^2) - \eta\sqrt{\bar{n}}AV -\alpha^2$
        \end{center}
        \item A function, $L$, of preference matrix, $\bR$ and parameterized by $\eta$, 
        \begin{center}
        $L(\bR; \eta):= \sqrt{4\|\bR\|_2^2+(\eta \|\bR\|_1^2)/4+\eta^2n+\max \{4\|\bR\|_2^2, 1+\eta^4\}}$
        \end{center}
    \end{enumerate}
    Under the assumptions of Theorem \ref{thm:SW_EA_Model1}, further assume $\kappa$, $n$ and $\|\bR^\star\|_1, \|\bR^\star\|_2$ are known. Let $\mathbf{z}^{\star}$ be the vector of true parameters, $\hat{\bz}$ the vector of estimates, and $\hat{\eta}$ the value returned by the algorithm. When $\hat{\eta}>0 $, if $\|\hat{\bz}-\bz^{\star}\|_2 < \left(f(\hat{\bz};\hat{\eta})\right)\left(L(\bR^\star; \hat{\eta})\right)^{-1}$,
    then the conclusions of Theorem \ref{thm:SW_EA_Model1} still hold when $\hat{\eta}$ is selected as the uprating value by users of the collective. 
    \end{restatable}
    The proof ensures that $\hat{\eta}$ still falls within the true sufficient condition range, which is defined by functions very similar to $f$ and $L$, using Lipschitz continuity; see Appendix \ref{app:users_proofs_robust} for details.
\section{What do collectives mean for RecSys design?}\label{sec:learner}
We have thus far analyzed users' incentives for collective action and its effectiveness against an $\alpha$ loss tolerant learner, without explicitly modeling the learner's utility function. In this regime, we provide sufficient conditions under which simple collective movement strategies guarantee both a Pareto improvement in recommendations and a strict improvement in social welfare. When these conditions hold and a collective of (majority) users care about minority users, collective action dominates the baseline where users interact with content based solely on personal interest. This raises an important question: is it ``bad'' RecSys design to incentivize collective strategization over ``truthful'' personal interest-based interactions? After all, in practice, engaging in these collectives is often a form of algorithmic protest \emph{against} the RecSys. In this section, we analyze the learner's welfare under different (learner) utility functions that are similar to those in existing modeling approaches. Interestingly, we show that the collective strategies we study improve its utility.

\subsection{Personalization accuracy learner}\label{sec:learner_ben}
We first analyze a learner whose welfare is solely determined by how good (i.e., accurate) each user's recommendations are in terms of what that user is \emph{personally} interested in receiving. Because the definition of users' social welfare is also a measure of how good recommendations are, formally, we call this a \emph{personalization accuracy} $\alpha$-loss tolerant learner, but note that equivalently it has a utility that is equal to the social welfare of all users. Conceptually, such a learner utility captures both benevolence towards users and natural definitions of personalization accuracy.
\begin{definition}[Personalization Accuracy Learner]
    Given a true preference matrix, $\bR^\star$, a personalization accuracy learner is an $\alpha$-loss tolerant learner such that:
    \begin{center}
        $U_{\BEN}:=\sum_{u \in [m]}{r^\star_{u,\Topk(u,\alpha)}} = \SW(\bR^\star,\alpha)$
    \end{center}
\end{definition}
Given that this learner prioritizes agents' social welfare, then from the results of the previous sections, we have that the collective action of some users increases the learner's utility.

\begin{restatable}{corollary}{BenLearner}\label{cor:ben_learner}
    Under the assumptions for Theorem \ref{thm:SW_EA_Model1}, an $\alpha$-loss tolerant learner with a personalization accuracy-based utility function would achieve 
    \begin{center}
    $U_{\BEN}^{\TRUE} = \sum_{u \in \calUmaj}\max_{i \in [n]}r_{u,i}^\star
    ,\quad 
    U_{\BEN}^{\COLL} = \sum_{u \in (\calUmaj \cup \calU_{i^\star})}\max_{i \in [n]}r_{u,i}^\star$
    \end{center}
    when agents are truthful or follow an effective collective strategy, respectively, and $U_{\BEN}^{\COLL} > U_{\BEN}^{\TRUE}$.
\end{restatable}

\subsection{Engagement learner}\label{sec:learner_en}
Next, we consider a learner whose welfare is based on how much \emph{engagement} it elicits on the platform. Recall that while we frequently refer to $\widetilde{\bR}$ as a ``revealed preferences'' matrix, preferences/ratings are usually inferred over time via user interactions such as watch time, views, likes, comments, etc. Alternatively, these interactions can be characterized as ``engagement''~\cite{engagement_maxxing}). A large $\tilde{r}_{u,i}$ indicates that a user had significant engagement with the item and, importantly, a platform can sell ad space on users' content engagement. We call a learner an ``engagement learner'', if their utility is equal to the total reported ratings across all users.
\begin{definition}[Engagement Learner]
     Given a true preference matrix, $\bR^\star$, and a revealed preference matrix, $\tbR$, an engagement-based learner is an $\alpha$-loss tolerant learner s.t., $U_{\EN}:= \sum_{i \in [n]}\sum_{u \in [m]}|\tilde{r}_{u,i}|$.
\end{definition}
Recall that the collective strategies we study are uprating schemes: a set of users submit rating $\eta$ instead of $0$. Thus, because engagement-based utility is simply the sum of ratings, a collective strategy yields higher utility than truthfulness.
\begin{restatable}[Collective increases utility for the engagement learner]{proposition}{EnLearner}\label{prop:en_learner}
Under assumptions for Theorem \ref{thm:SW_EA_Model1}, and $\alpha$-loss tolerant learner with engagement-based utility would achieve
\begin{center}
$U_{\EN}^{\TRUE} = \sum_{i \in [n]}\sum_{u \in [m]}|r_{u,i}^\star|
,\quad
U_{\EN}^{\COLL} = \sum_{i \in [n]}\sum_{u \in [m]}|r_{u,i}^\star| + \eta|\calU_\COLL|$
\end{center}
when agents are truthful or follow an effective collective strategy, respectively, and $U_{\EN}^{\COLL} > U_{\EN}^{\TRUE}$.
\end{restatable}

\xhdr{Protest against the algorithm helps the platform.}
Mathematically, both~\Cref{cor:ben_learner} and~\Cref{prop:en_learner} follow directly from \Cref{sec:users}. However, conceptually, they formalize a notable idea: even though this collective action is, from some perspectives \emph{untruthful} and in the real-world meant as a resistance \emph{against} the platform, depending on platform priorities, it may be the case that the platform benefits from it. As such, it is not inherently necessary that a ``good'' RecSys designs away such incentives. Naturally, there are some utilities that may be harmed by collective actions, which does suggest space for interesting future work; we discuss this further in Section \ref{sec:discussion}.
\section{Empirics}\label{sec:experiment}

We expect that our theoretical conditions to improve recommendations are not strictly \emph{necessary} and that more general strategies may work in realistic settings. To this end, we present an empirical simulation of collective action against an LLM-based book recommender. Our empirics validate the theory: even though the formal assumptions are not met, because the model still recommends popular items to users that would prefer something more niche, a collective of majority users can improve minority welfare and only very minimally hurt test accuracy.

\xhdr{Method overview.}
We briefly summarize the dataset and methods, leaving details to \Cref{app:experiment_details}. We use the Goodreads user/book interactions dataset~\cite{goodreads1, goodreads2}. Interactions include (1) adding to shelf (``shelving'') to indicate interest, (2) marking as read, (3) rating (out of 5), and (4) reviewing. We fine-tune four LLM-based book recommenders: given a user's profile of past interactions with young-adult (YA) books and four new YA book options, each recommends two (of the four) new YA books that the user should add to their shelf. That is, the model recommends two new books the user would be interested in. To simulate collective action, we train each model on different versions of $85,238$ users' data. The baseline model is trained on users' true history of pre-2015 interactions and uses two of each user's actually shelved 2015 books as the ``ground-truth'' recommendation, with two other randomly generated 2015 book options. We train the other three models using the same structure, but with perturbed user data that represents collective action to varying levels, which we discuss in the following section. 

\xhdr{Results.}
We first look at recommendations without collective action. On a hold-out test set of $21,310$ users, the baseline model is $98.8\%$ partially or fully correct in its recommendations of two books each user should add to their shelf (Table \ref{tab:accuracies}). However, this is driven by correct recommendations of popular authors/books (Figure \ref{fig:baseline_model_acc}). 
\begin{figure}[ht]
\centering

\begin{subfigure}{0.48\textwidth}
  \centering
  \includegraphics[width=\linewidth]{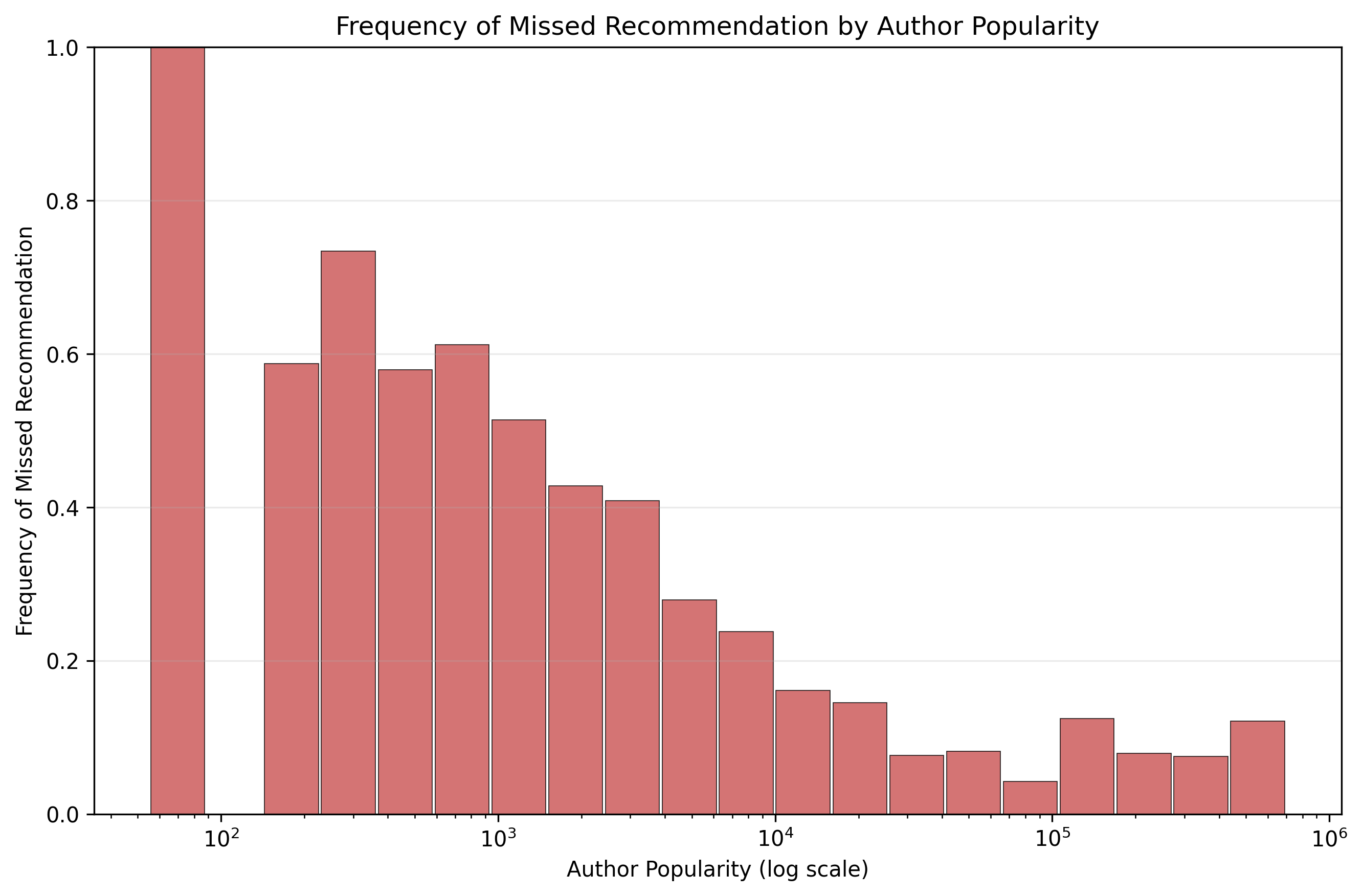}
  \caption{Avg frequency with which books are incorrectly not recommended by popularity of their author}
  \label{fig:sub1}
\end{subfigure}
\hfill
\begin{subfigure}{0.48\textwidth}
  \centering
  \includegraphics[width=\linewidth]{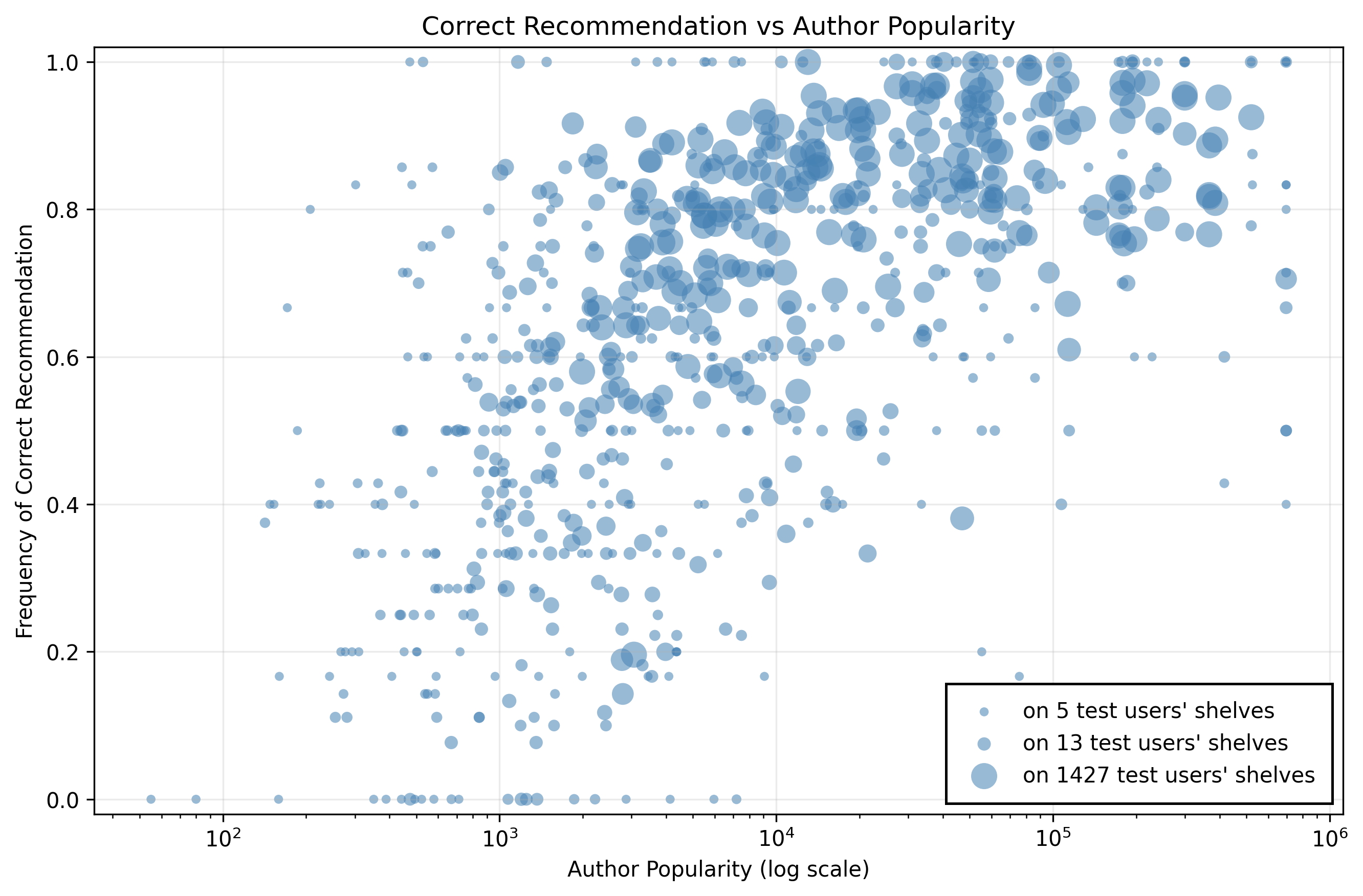}
  \caption{Avg frequency with which books are correctly recommended by popularity of their author}
  \label{fig:sub2}
\end{subfigure}

\caption{[In]correct recommendation of books by their author popularity, measured by $\#$ of total interactions an author gets. Books shelved by fewer than $5$ users are dropped. More popularity metrics in \Cref{app:experiment_additional_results}.}
\label{fig:baseline_model_acc}
\end{figure}
In fact, on test users who shelve less popular books, if given popular alternative options, the baseline model almost always recommends the user add the \emph{popular} books to their shelf (Column 1, Table \ref{tab:improve_rec}), matching the ideas of \Cref{g_maj_b_min2}. From a purely predictive perspective, this behavior is reasonable. A randomly drawn user is likely to be interested in and want to add a popular book to their shelf. However, this has downstream effects. First, non-mainstream users have poor welfare in this RecSys. Second, the niche authors/books are \emph{disproportionately} less promoted. On real platforms, these are not just unfortunate but accepted outcomes; users notice and resist. On TikTok, for example, users have noticed similar outcomes with books by marginalized authors~\cite{melo_booktok} and purposefully increased interaction with said authors' content to correct for algorithms' lack of promotion \cite{mccall_booktoks_2022}. Focusing on $10$ ``target'' YA authors in our dataset whose 2015 books are otherwise rarely accurately recommended for shelving against more popular books, we simulate collective action by perturbing the training data. As our collective, we choose $80\%$ of users whose correct shelf recommendation would include a top-30 most shelved book. As the strategic action, each user in the collective adds one of the target 2015 books to their shelf and increases pre-2015 interactions with $j$ of the target authors. Across $j \in \{1,3,7\}$, we train 3 new models and find major improvement in (correct) recommendation of target authors to test users who like their books (analogous to \Cref{thm:SW_EA_Model1}). Table \ref{tab:improve_rec} presents the frequency with which each model recommends shelving the targeted book to test users whose true 2015 book shelf includes said target book. We present results where the models recommended two books from a set of 2 truly shelved books and 2 random books from 100 authors with the greatest number of past interactions. In Appendix \ref{app:experiment_additional_results}, we have similar results for other thresholds of popularity.

\begin{table}[ht]
\centering
\footnotesize
\begin{tabular}{lrrrrr}
\toprule
Target Author & Baseline & $j=1$ & $j=3$ & $j=7$ & Test user \# \\
\midrule
K.A.\ Holt       
& 0.5769 
& 0.9615 \,(+66.7\%) 
& 0.9615 \,(+66.7\%) 
& 0.9615 \,(+66.7\%) 
& 26 \\

Allan Stratton   
& 0.0952 
& 1.0000 \,(+950.4\%) 
& 1.0000 \,(+950.4\%) 
& 0.9048 \,(+850.4\%) 
& 21 \\

Julie Mayhew     
& 0.3333 
& 0.9444 \,(+183.3\%) 
& 0.9444 \,(+183.3\%) 
& 0.9444 \,(+183.3\%) 
& 18 \\

Matthew Crow     
& 0.0667 
& 1.0000 \,(+1399.3\%) 
& 1.0000 \,(+1399.3\%) 
& 1.0000 \,(+1399.3\%) 
& 15 \\

Edwidge Danticat 
& 0.0667 
& 0.9333 \,(+1299.3\%) 
& 0.9333 \,(+1299.3\%) 
& 0.9333 \,(+1299.3\%) 
& 15 \\

S.\ Walden       
& 0.3077 
& 1.0000 \,(+225.0\%) 
& 1.0000 \,(+225.0\%) 
& 1.0000 \,(+225.0\%) 
& 13 \\

Addison Moore    
& 0.4545 
& 1.0000 \,(+120.0\%) 
& 1.0000 \,(+120.0\%) 
& 1.0000 \,(+120.0\%) 
& 11 \\

Alan Gratz       
& 0.4000 
& 1.0000 \,(+150.0\%) 
& 1.0000 \,(+150.0\%) 
& 0.9000 \,(+125.0\%) 
& 10 \\

Anna Martin      
& 0.5000 
& 1.0000 \,(+100.0\%) 
& 1.0000 \,(+100.0\%) 
& 1.0000 \,(+100.0\%) 
& 10 \\

Gabriella Lepore 
& 0.4000 
& 0.9000 \,(+125.0\%) 
& 0.9000 \,(+125.0\%) 
& 0.9000 \,(+125.0\%) 
& 10 \\
\bottomrule
\end{tabular}
\caption{Models' frequency of target author recommendation to interested test users. Percent change is computed relative to the baseline model.}
\label{tab:improve_rec}
\end{table}

\xhdr{Limitations.}
Because we cannot apply our guarantees of \emph{Pareto} improvement, more general collective action may come at a cost in accuracy. However, Table \ref{tab:accuracies} highlights that, for this dataset, test accuracy loss is minimal (no more than a $-.2\%$ change in combined partial and full correct percentage). Additionally, we assume a large amount, $80\%$, of popular-book users do collective action. Due to the computational cost of retraining many models, we are limited in the experimentation we present, however, as the improvement was so large, we expect smaller percentages to still work.
\section{Discussion} \label{sec:discussion}

Our work is the first to theoretically study and formalize the real-world phenomenon of collective algorithmic resistance movements on recommendation system-based online platforms. %
To conclude, we highlight some of our key assumptions and avenues for future work.

\xhdr{Assumptions and limitations.}
In the main body of our work, we make three somewhat strong assumptions: disjoint preferences between the majority and minority users, uprating strategies such that all users in a collective uprate a chosen item by a single shared value, $\eta$, and picky items exclusively liked by a disjoint subset of minority users. Though we validate our main body's theory with experimental results in Section \ref{sec:experiment} that confirm realistically that these assumptions are not required to see any benefits for minorities, in Appendix \ref{app:model2}, we also present analogous guarantees to our main body using alternative (and looser) assumptions. However, due to the complexity of a generalized collective strategy, we do not provide an algorithm to find it. Throughout the paper, we also assume a setting of top-1 recommendation rather than the more realistic top-$k$ recommendation setting. This is primarily for mathematical and conceptual simplicity. The essence of our results can be generalized to a broader range of $k$, which we detail in Appendix \ref{app:topk}. Finally, it is important to restate one limitation: the ``bad'' outcome we focus on in \Cref{thm:rank_min_bad} is just one solution to Equation \eqref{eqn:rank_min}. It is a very notable solution because it is the sparsest, which is a desirable property for storage/computation, and it represents a ``core'' solution to which all other solutions can be reduced (i.e., made sparser until they become that one), but given a solver that outputs \emph{any} solution with some probability, it is possible to \emph{not} get the bad one.

\xhdr{Future work.}
While our paper's contribution is primarily theoretical, we also study the phenomenon of collective movements via empirics (Section \ref{sec:experiment}) and survey (Appendix \ref{app:survey}). Likewise, we believe there are interesting avenues to further explore across all three of these subfields. On the theory side, using our online MC abstraction or others, there are many alternative multi-agent strategic interactions that can be studied. For example, what are the conditions under which bots--whose interactions, like those of the collective, are intended to impact other users feeds--can significantly impact social welfare? On the HCI side, while we cited many interview studies and our own survey, there are no (to our knowledge) large-scale behavioral experiments of users' interpersonally-motivated content interaction strategies. It would be useful to understand whether interpersonal content interaction strategies are more prevalent on some platforms over others, which actual interactions (e.g., comments, likes, etc) are most often used, and, as is relevant to the data-sharing and collaboration required for our algorithm, how users organize and coordinate. Finally, on the empirical side, collective strategic interactions could affect other algorithms. Our SFT post-training on an LLM inadvertently highlights a kind of collective ``alignment hacking''. This may be interesting to study more broadly for AI safety.
\section{Acknowledgments}
Research reported in this paper was supported by an Amazon Research Award Fall 2023. Any opinions, findings, and conclusions or recommendations expressed in this material are those of the author(s) and do not reflect the views of Amazon.

\bibliographystyle{plainnat}
\bibliography{bib}
\newpage
\appendix
\section*{\LARGE Supplementary Material}
\section{Use of AI Tools}
The authors used ChatGPT (GPT-5.2) for assistance with minor text editing and Latex formatting. Claude Code (Sonnet 4.5) and Codex (GPT-5.3) were used for coding assistance. All output has been reviewed and edited by the authors.

\section{Table of notation}\label{app:notation_table}

To assist the reader, we include a table summarizing our paper notation in \Cref{table:notation}.
\renewcommand{\arraystretch}{1.2}

\begin{table}[H]\caption{Paper notation in the main body}\label{table:notation}
\centering
\begin{tabular}{|c | c|} 
\hline 
Notation & Explanation \\ [0.5ex] 
\hline \hline
$\bR$ & generic preference matrix (elements $r_{u,i}$) \\ 
 \hline
$\bR^\star$ & personal interest matrix \\ 
 \hline
$\tbR$ & revealed preference matrix \\ 
 \hline
$\hR$ & low-rank approximation of preference matrix \\ 
 \hline
$\bR_{\maj}, \bR_{\minor}$ & submatrices for majority and minority groups \\ 
 \hline
$\alpha$ & learner's loss tolerance parameter \\ 
 \hline
$k^\star$ & learner's chosen rank \\ 
 \hline
$k_{\maj}$ & rank of the majority submatrix \\ 
 \hline
$\mathrm{TVR}(k^\star, \bR)$ & total variation retained after $k^\star$-truncated svd on $\bR$\\ 
 \hline
$\Topk(u)$ & top item recommended to user $u$ \\ 
 \hline
$\topitems(u)$ & set of user $u$'s top items under $\hR$ \\ 
 \hline
$\topitems^{\pop}(u)$ & most popular top item(s) for user $u$ \\ 
 \hline
$\calU$ & set of users \\ 
 \hline
$\calUmaj, \calUmin$ & majority and minority user sets \\ 
 \hline
$\calImaj, \calImin$ & majority and minority item sets \\ 
 \hline
$\mathcal{G}(\bR)$ & singular value gap range for matrix $\bR$ \\ 
 \hline
$i^\star$ & picky item index preferred by minority users \\ 
 \hline
$\calU_{i^\star}$ & group of users who like only item $i^\star$ \\ 
 \hline
$\calU_\COLL$ & subset of majority users who form collective \\ 
 \hline
$\eta$ & positive uprating amount used by collective users \\ 
 \hline
$\hat{\eta}$ & suggested strategy returned by Algorithm \ref{alg:find_eta_aglo} \\ 
 \hline
$\AV(\bR, \calU_{\COLL})$ & the largest aggregate value of a popular item for users in the collective \\ 
 \hline
$\mathcal{G}(\bR, \calU_\EA, \eta)$ & uprating-aware singular value gap \\ 
 \hline
$\kappa$ & smallest top rating among majority users \\ 
 \hline
$\SW(\bR, \alpha)$ & social welfare under matrix $\bR$ and parameter $\alpha$ \\ 
 \hline
$\hat{\bz}$, $\bz^\star$ & estimated and true parameter vectors for Algorithm \ref{alg:find_eta_aglo}\\ 
\hline
\end{tabular}
\end{table}

\section{Supplementary material for Section \ref{sec:intro}}\label{app:intro}

\subsection{Connection to the PCA item-fairness}\label{app:connection_fair_pca}
While connections between our main results and \citet{pca_fairness_liu} are mentioned in Section \ref{sec:related_work}, we include here a more in-depth discussion. In short, our work and theirs are highly complementary. While we reach similar initial negative conclusions about the unfairness of majority-minority matrices (our \Cref{thm:rank_min_bad}, \Cref{g_maj_b_min2}, and \Cref{g_maj_b_min}, their Theorem 1), both the specific mathematical formalization of these results (i.e., our shared focus on rank-reducing mechanisms) and the solutions put forth we study are significantly different.

Theoretically, \citet{pca_fairness_liu} establish a type of item-based unfairness in PCA-based recommenders. Their primary result (Theorem 1) pertains to (item recommendation) unfairness \emph{in the limit} as the system of users and items grows. That is, defining popular items using the growth of a relevant ratings matrix singular value and unpopular items using lack of growth in a relevant singular value, under some assumptions on growth rate of items and users in the system, in the limit, PCA will zero-out the ratings of the unpopular items. This result is conceptually similar to the results we present on the unfairness of recommendations for minority users \emph{prior} to users engaging in collective action (\Cref{thm:rank_min_bad}, \Cref{g_maj_b_min2}, and \Cref{g_maj_b_min}). However, importantly, our result is in the \emph{finite} user and item case, thus more closely capturing realistic recommendation systems. Additionally, because we more explicitly root our model as an abstraction of matrix completion recommendation, rather than discussing this as an \emph{item} unfairness, we directly present what the zeroing out means for \emph{users'} recommendation and social welfare in our model. Finally, we use our negative results to study very different solutions. To increase fairness, \citet{pca_fairness_liu} present a novel item-weighted PCA methodology. In the case of our paper, we analyze a real-world \emph{user} phenomenon: \emph{collective strategic interaction}. 

Overall, our results, which are initially similar, and distinct solutions create significantly different contributions that validate each other.

\section{Supplementary material for Section \ref{sec:model}}\label{app:model}

\subsection{Details on online matrix completion}\label{app:MC}
\subsubsection{An illustration of matrix completion}\label{app:MC_illustration}

The basic idea of matrix completion is very simple: the algorithm starts with an empty matrix and repeatedly queries random entries (or collections of entries) at each round; this is the ``exploration'' phase. At some stopping point, the algorithm ends exploration and estimates a completed matrix that minimizes the rank while retaining all known entries (formally, Equation \eqref{eqn:rank_min}). It is useful to visualize this process using matrix forms. Consider the following true, but unknown ratings matrix,
\[\bR^\star =
    \begin{pmatrix} 
    1 & 0 & 0 & 4\\
    0 & 1 & 1 & 2\\
    0 & 1 & 1 & 2\\
    5 & 0 & 0 & 20
\end{pmatrix}\]
This matrix is $ 4\times 4$, but rank 2, since all row vectors could be created via linear combinations of 2 row vectors (rows 1 and 2). Every row is a user and every column is an item. Each entry, $r^\star_{u,i}$, represents user $u$'s rating of item $i$. In round 1, the algorithm randomly queries some entries (e.g., one random item per user) and ends up with the following partially known matrix:
\[\widetilde{\bR}(1) =
    \begin{pmatrix} 
    - & 0 & - & -\\
    - & - & 1 & -\\
    - & - & 1 & -\\
    - & - & - & 20
\end{pmatrix}\]
At this point, the recommender still does not know very much about the users' preferences. Notably, it may be the case that they all have the same exact preference vector. One way to see this is to ask: what is the lowest possible rank that $\widehat{\bR}$ could be? Clearly, rank 1. Because it is still possible to write a completed version of this matrix where every row vector is the same one. Here is one such solution to Equation \eqref{eqn:rank_min}. {\color{red}Red} represents entries that were completed:
\[\widehat{\bR}(1) =
    \begin{pmatrix} 
    {\color{red}0} & 0 & {\color{red}1} & {\color{red}20}\\
    {\color{red}0} & {\color{red}0} & 1 & {\color{red}20}\\
    {\color{red}0} & {\color{red}0} & 1 & {\color{red}20}\\
    {\color{red}0} & {\color{red}0} & {\color{red}1} & 20
\end{pmatrix}\]
Of course, this is not unique, alternatively the 1st and the 4th row vectors would be multiples of the 2nd and 3rd, which would still ensure rank 1:
\[\widehat{\bR}(1) =
    \begin{pmatrix} 
    {\color{red}0} & 0 & {\color{red}20} & {\color{red}20}\\
    {\color{red}0} & {\color{red}0} & 1 & {\color{red}1}\\
    {\color{red}0} & {\color{red}0} & 1 & {\color{red}1}\\
    {\color{red}0} & {\color{red}0} & {\color{red}20} & 20
\end{pmatrix}\]
And we could change the 1st column zero to be any common number that follows the required row-wise linear combination rules:
\[\widehat{\bR}(1) =
    \begin{pmatrix} 
    {\color{red}77*20} & 0 & {\color{red}20} & {\color{red}20}\\
    {\color{red}77} & {\color{red}0} & 1 & {\color{red}1}\\
    {\color{red}77} & {\color{red}0} & 1 & {\color{red}1}\\
    {\color{red}77*20} & {\color{red}0} & {\color{red}20} & 20
\end{pmatrix}\]
As the recommender queries more, the matrix is easier to estimate. Suppose instead some round $t$ (note that the recommender does not have to query one random item per user at each round), the recommender has the following matrix.
\[\widetilde{\bR}(t) =
    \begin{pmatrix} 
    1 & 0 & - & -\\
    0 & 1 & - & 2\\
    - & 1 & 1 & -\\
    - & 0 & 0 & 20
\end{pmatrix}\]
At this point, it is no longer possible for any feasible completed matrix to reach as low as rank 1. One way to see this is that row 1 and row 2 cannot possibly be multiples (by a nonzero constant) of each other. Rank 2 is possible however, so any feasible and optimal completed matrix must achieve this rank. While Rank 2 is the true rank of the matrix, it does not mean the recommender will definitely complete the matrix perfectly. Consider the following completion, which is feasible and optimal according to~\eqref{eqn:rank_min}, but not exact according to $\bR^\star$:
\[\widehat{\bR}(t) =
    \begin{pmatrix} 
    1 & 0 & {\color{red}0} & {\color{red}20}\\
    0 & 1 & {\color{red}1} & 2\\
    {\color{red}0} & 1 & 1 & {\color{red}2}\\
    {\color{red}1} & 0 & 0 & 20
\end{pmatrix}\]
In summary, the results of matrix completion generally improve as more entries are queried and often, the optimal solution is not unique. Even further, rank minimization as an objective is actually not very easy to solve at a large scale and it is common to replace this objective with a (convex) approximation such as nuclear norm~\cite{candes_exact_2008, recht_simpler_2011, candes_power_2010}. This approximation is outside the scope this paper, we refer the reader to surveys from \citet{nguyen_mc_survey, li2019surveymatrixcompletionperspective, panish2025matrixcompletionsurveytheory} for further discussion.
\subsubsection{Connections to coherence} \label{app:MC_coherence}
\citet{candes_exact_2008} discuss the difficulty (i.e., how many entries should be queried) of perfectly recovering a matrix. Consider their example of a difficult-to-perfectly-recover matrix:
\[\bR^\star =
    \begin{pmatrix} 
    0 & 0 & \dots & 0 & 1\\
    0 & 0 & \dots & 0 & 0\\
    \vdots & \vdots & \vdots & \vdots & \vdots\\
    0 & 0 & \dots & 0 & 0
\end{pmatrix}\]
To perfectly recover $\bR^\star$ (i.e., to have reached an estimate $\widehat{\bR}$ s.t. $\widehat{\bR} = \bR^\star$) the learner must sample \emph{nearly all} entries. This stems from the issue that this matrix, even though it is low-rank, is incredibly sparse. In particular, the singular vectors of $\bR^\star$ (recall that these are the vectors that make up the singular value decomposition of a matrix, $M$, in the form $M = \sum_{i \in \textrm{rank}(M)}\sigma_i \mathbf{u}_i\mathbf{v}_i$ where $\sigma_i$ are singular values) are highly concentrated. \citet{candes_exact_2008} define a formal measurement of this singular value spread called \emph{``coherence''}. Matrices with a very high coherence (such as $\bR^\star$ in the example) are very hard to exactly complete.

While we do not explicitly focus on it as we primarily introduce MC to establish our own model, the majority-minority matrices we discuss clearly have a similar-structured sparsity given their block-diagonal structure. Additionally, the relevant failing of MC for the matrices we study is not its inexactness, but its potential zeroing out of minority preferences specifically. It may be of interest to future work that is entirely in the (online) MC perspective to analyze matrix completion for majority-minority (likely highly coherent) group structured matrices. There may be interesting guarantees on exact completion of the whole majority-minority matrix and other more group-disparate guarantees on the (higher) difficulty of exactly recovering the minority part of the matrix.

\subsubsection{Are the assumptions of Proposition~\ref{thm:rank_min_bad} strong?} \label{app:MC_ass}

We begin by discussing what Proposition~\ref{thm:rank_min_bad}'s assumption on $\Omega$ actually looks like in practice. Suppose we have a majority-minority ratings matrix made up of 10 users and 10 items. 2 users are the minority users and each of them only likes a specific item that no one else likes. Ordered such that it is directly in diagonal form, such a matrix may look like:

\[
\begin{tikzpicture}[baseline=(A.center)]
  \node (L) {$\bR^{\star}=$};

  \node[right=2mm of L] (A) {$
  \begin{pmatrix}
  0.2 & 0.7 & 0.4 & 0.9 & 0.1 & 0.6 & 0.3 & 0.8 & 0.0 & 0.0\\
  0.5 & 0.1 & 0.9 & 0.2 & 0.7 & 0.4 & 0.8 & 0.3 & 0.0 & 0.0\\
  0.8 & 0.3 & 0.6 & 0.1 & 0.9 & 0.2 & 0.5 & 0.7 & 0.0 & 0.0\\
  0.1 & 0.9 & 0.2 & 0.5 & 0.8 & 0.3 & 0.7 & 0.4 & 0.0 & 0.0\\
  0.6 & 0.4 & 0.7 & 0.3 & 0.2 & 0.9 & 0.1 & 0.8 & 0.0 & 0.0\\
  0.3 & 0.8 & 0.1 & 0.7 & 0.4 & 0.5 & 0.9 & 0.2 & 0.0 & 0.0\\
  0.9 & 0.2 & 0.5 & 0.4 & 0.3 & 0.8 & 0.6 & 0.1 & 0.0 & 0.0\\
  0.4 & 0.5 & 0.8 & 0.6 & 0.1 & 0.7 & 0.2 & 0.9 & 0.0 & 0.0\\
  0.0 & 0.0 & 0.0 & 0.0 & 0.0 & 0.0 & 0.0 & 0.0 & {\color{red}0.8} & 0.0 \\
  0.0 & 0.0 & 0.0 & 0.0 & 0.0 & 0.0 & 0.0 & 0.0 & 0.0 & {\color{red}0.9} \\
  \end{pmatrix}
  $};

  \begin{scope}[on background layer]
    \fill[teal!25, rounded corners, fill opacity=0.45]
      ($(A.north east)+(-1.50cm,-0.25cm)$) rectangle
      ($(A.south east)+(-0.10cm, 0.25cm)$);

    \fill[orange!30, rounded corners, fill opacity=0.45]
      ($(A.south west)+( 0.10cm, 0.90cm)$) rectangle
      ($(A.south east)+(-0.10cm, 0.25cm)$);
  \end{scope}
\end{tikzpicture}
\]
{The elements in {\color{red}{red}} are the only ones which the theorem assumes are \emph{not} queried}; all other elements may or may not be queried. Thus, we see that the theorem statement holds as long as only these specific 2 out of 100 elements are not randomly selected for query. Meanwhile, 2 out of 10 (a much larger proportion) users might be poorly served by a recommender as a result of this.

The block diagonal structure of $\bR^\star$ (what we call a majority-minority matrix) could also be perceived as a (strong) assumption of Proposition~\ref{thm:rank_min_bad}. To this end, we point out that technically, it does not have to be the true matrix, $\bR^\star$, that is a block-diagonal, but that simply it should look like the partially completed matrix \emph{could be} from a block-diagonal. Concretely, consider the following alternative true ratings matrix:

\[
\begin{tikzpicture}[baseline=(A.center)]
  \node (L) {$\bR^{\star'}=$};

  \node[right=2mm of L] (A) {$
  \begin{pmatrix}
  0.2 & 0.7 & 0.4 & 0.9 & 0.1 & 0.6 & 0.3 & 0.8 & {\color{red}0.1} & 0.0\\
  0.5 & 0.1 & 0.9 & 0.2 & 0.7 & 0.4 & 0.8 & 0.3 & 0.0 & 0.0\\
  0.8 & 0.3 & 0.6 & 0.1 & 0.9 & 0.2 & 0.5 & 0.7 & 0.0 & 0.0\\
  0.1 & 0.9 & 0.2 & 0.5 & 0.8 & 0.3 & 0.7 & 0.4 & 0.0 & {\color{red}0.7}\\
  0.6 & 0.4 & 0.7 & 0.3 & 0.2 & 0.9 & 0.1 & 0.8 & 0.0 & 0.0\\
  0.3 & 0.8 & 0.1 & 0.7 & 0.4 & 0.5 & 0.9 & 0.2 & 0.0 & 0.0\\
  0.9 & 0.2 & 0.5 & 0.4 & 0.3 & 0.8 & 0.6 & 0.1 & 0.0 & 0.0\\
  0.4 & 0.5 & 0.8 & 0.6 & 0.1 & 0.7 & 0.2 & 0.9 & 0.0 & 0.0\\
  0.0 & {\color{red}0.5} & 0.0 & 0.0 & 0.0 & {\color{red}0.2} & 0.0 & 0.0 & {\color{red}0.8} & 0.0 \\
  {\color{red}0.8} & 0.0 & 0.0 & 0.0 & 0.0 & 0.0 & 0.0 & {\color{red}0.3} & 0.0 & {\color{red}0.9} \\
  \end{pmatrix}
  $};

  \begin{scope}[on background layer]
    \fill[teal!25, rounded corners, fill opacity=0.45]
      ($(A.north east)+(-1.50cm,-0.25cm)$) rectangle
      ($(A.south east)+(-0.10cm, 0.25cm)$);

    \fill[orange!30, rounded corners, fill opacity=0.45]
      ($(A.south west)+( 0.10cm, 0.90cm)$) rectangle
      ($(A.south east)+(-0.10cm, 0.25cm)$);
  \end{scope}
\end{tikzpicture}
\]

Now instead, the minority users' preferences and their preferred items are not entirely discrete, but just \emph{sparse}. As long as the red elements, which are now greater in number than in $\bR^\star$, are not queried, Proposition~\ref{thm:rank_min_bad} still holds because these elements could be filled in as zero and a low rank approximation is achieved. Of course, as the sparsity decreases, this kind of assumption of $\Omega$ not containing any of the red entries becomes less reasonable. Overall, it is not the case that perfect block-diagonal structure must be true; sparse niche users/items can create the same effect.
\subsubsection{Proof of Proposition~\ref{thm:rank_min_bad}} \label{app:MC_proofs}
\RankMinBad*
\noindent \begin{proof}[Proof of~\Cref{thm:rank_min_bad}]

The rank of a matrix is defined as the number of linearly independent columns or rows and it is invariant to column or row permutations. Without loss of generality, consider the true (unknown) ratings matrix $\bR^\star$ in the following (potentially permuted) form:
\[\bR^\star =
    \begin{pmatrix} 
    \bR_{\maj}^\star & \mathbf{0} \\
    \mathbf{0} & \bR_{\minor}^\star
\end{pmatrix}\]
where $\bR_{\maj}^\star \in \R^{\barm \times \barn}$ and $\bR_{\minor} \in \R^{(m-\barm) \times (n-\barn)}$. In other words, we reorder the columns and rows such that the first $\barm$ users are all $u \in \calUmaj$ while the first $\barn$ items are all $i \in \calImaj$. Similarly, we refer to the partially completed matrix in these row and column orderings: 
\[\widetilde{\bR} =
    \begin{pmatrix} 
    \widetilde{\bR}_\maj & \tilde{\mathbf{0}} \\
    \tilde{\mathbf{0}} & \tilde{\mathbf{0}}
\end{pmatrix}\]
where $\widetilde{\bR}_{\maj}$ is the partially completed version of $\bR_{\maj}^\star$ matrix (that is, just $\bR_{\maj}^\star$, but with unknown entries represented as $-$) and $\tilde{\mathbf{0}}$ refers to partially completed blocks of 0s (that is, blocks of $0$, but with unknown entries represented as $-$). Note the bottom right block is a partially complete block of zeros by~\Cref{thm:rank_min_bad}'s assumption on $\Omega$, that the only $\bR_{\minor}^\star$ entries that are known (i.e., constrained for the $\widehat{\bR}$ solution) are zeros. Lastly, we refer to the completed version (i.e., a solution to Equation~\eqref{eqn:rank_min}), in the same reordering:
\[\widehat{\bR} =
    \begin{pmatrix} 
    \widehat{\bR}_{\maj} & \widehat{\mathbf{A}} \\
    \widehat{\mathbf{B}} & \widehat{\bR}_{\minor}
\end{pmatrix}\]

In order to prove the proposition, all we must do is show the following claim: 
\begin{claim}\label{claim:MC} Given any $\widehat{\bR}$ that is a feasible optimal solution to Equation \eqref{eqn:rank_min}, we can make $\widehat{\mathbf{A}} = \mathbf{0}$, $\widehat{\mathbf{B}} = \mathbf{0}$, and $\widehat{\bR}_{\minor} = \mathbf{0}$ and have a new $\widehat{\bR}'$ that is also a feasible optimal solution.
\end{claim}
Then, because $\widehat{\bR}'$ is a feasible optimal solution solution but also sparser than $\widehat{\bR}$, it must be the case that the sparsest feasible optimal solution is one that has $\widehat{\mathbf{B}} = \mathbf{0}$, and $\widehat{\bR}_{\minor} = \mathbf{0}$. This concludes the proof of the proposition. Before we delve into the proof, one may wonder why would it ever be the case that any of the $\widehat{\mathbf{A}}, \widehat{\mathbf{B}}, \widehat{\mathbf{\bR}}_{\minor}$ would be different than $\mathbf{0}$. We discuss this in the next subsection.

To prove Claim~\ref{claim:MC}: Recall that by definition of matrix rank, given any matrix $\mathbf{X}$, if we define $\mathbf{X}'$ as the same matrix but with any appended row or column vector, $\textrm{rank}(\mathbf{X}') \geq \textrm{rank}(\mathbf{X})$. Now, consider any $\widehat{\bR}$ that is a feasible and optimal solution to Equation \eqref{eqn:rank_min}. We can recover $\widehat{\bR}_\maj$ by iteratively removing columns and rows from $\widehat{\bR}$. Therefore, $\textrm{rank}(\widehat{\bR}) \geq \textrm{rank}(\widehat{\bR}_\maj)$. Padding out $\widehat{\bR}_\maj$ with zeros such that it looks like:
\[\widehat{\bR}' =
    \begin{pmatrix} 
    \widehat{\bR}_{\maj} & \mathbf{0} \\
    \mathbf{0} & \mathbf{0}
\end{pmatrix}\]
such that $\widehat{\bR}' \in \R^{m\times n}$, we now have $\textrm{rank}(\widehat{\bR}) \geq \textrm{rank}(\widehat{\bR}_\maj) = \textrm{rank}(\widehat{\bR}')$. But clearly, $\widehat{\bR}'$ still satisfies all the $\Omega$ constraints (recall  that by the~\Cref{thm:rank_min_bad}'s assumption that the only $(u,i) \in \Omega$ where $u \in \calUmin$ and $i \in \calImin$ are s.t. $r^\star_{u,i} = 0$) and it achieves the same or lower rank as $\widehat{\bR}$, which is a rank minimizing solution. Therefore $\widehat{\bR}'$ must also a feasible optimal solution.
\end{proof}

\subsubsection{A note about the ``less sparse'' solutions to Equation \ref{eqn:rank_min}}\label{app:less_sparse}
It is not immediately obvious how there even exist the less-sparse feasible and optimal solutions where $\widehat{\mathbf{A}} \neq \mathbf{0}$, $\widehat{\mathbf{B}} \neq \mathbf{0}$, and/or $\widehat{\bR}_{\minor} \neq \mathbf{0}$ because it may seem that to find a feasible minimum rank $\widehat{\bR}$ under the assumption on $\Omega$ in~\Cref{thm:rank_min_bad}, one solving method that would work and is feels very natural due to the block nature is to immediately set $\widehat{\mathbf{A}} = \mathbf{0}$, $\widehat{\mathbf{B}} = \mathbf{0}$, and $\widehat{\bR}_{\minor} = \mathbf{0}$, and then proceed to filling in $\widetilde{\bR}_\maj$. However, generally, it is not necessary that $\widehat{\mathbf{A}} = \mathbf{0}$, $\widehat{\mathbf{B}} = \mathbf{0}$, and $\widehat{\bR}_{\minor} = \mathbf{0}$ for $\widehat{\bR}$ to be feasible and optimal. To illustrate this, consider the following example where {\color{red}red} represents $\widetilde{\bR}_\maj$, {\color{ForestGreen}green} represents $\mathbf{A}, \mathbf{B}$, and {\color{RoyalPurple}purple} represents $\widetilde{\bR}_\minor$:
\[\widetilde{\bR} =
    \begin{pmatrix} 
    {\color{red}1} & {\color{red}-} & {\color{red}0} &{\color{ForestGreen}-} & {\color{ForestGreen}0} & {\color{ForestGreen}-}\\
    {\color{red}0} & {\color{red}-} & {\color{red}0} & {\color{ForestGreen}0} & {\color{ForestGreen}-} & {\color{ForestGreen}-}\\
    {\color{red}1} & {\color{red}0} & {\color{red}1} & {\color{ForestGreen}-} & {\color{ForestGreen}0} & {\color{ForestGreen}0}\\
    {\color{ForestGreen}-} & {\color{ForestGreen}0} & {\color{ForestGreen}-} & {\color{RoyalPurple}-} & {\color{RoyalPurple}-} & {\color{RoyalPurple}-}\\
    {\color{ForestGreen}0} & {\color{ForestGreen}-} & {\color{ForestGreen}-} & {\color{RoyalPurple}-} & {\color{RoyalPurple}-} & {\color{RoyalPurple}-}\\
    {\color{ForestGreen}0} & {\color{ForestGreen}0} & {\color{ForestGreen}-} &{\color{RoyalPurple}-} & {\color{RoyalPurple}0} & {\color{RoyalPurple}-}
\end{pmatrix}\]
Because of the red items, there is no way to fill in rows such that \emph{all} rows are linearly dependent with each other; as a result, we cannot make this matrix only rank 1. However, rank 2 is possible, therefore this is minimum rank. Now, one obvious way to fill this in is to do the following:
\[\widehat{\bR} =
    \begin{pmatrix} 
    {\color{red}1} & {\color{red}0} & {\color{red}0} &{\color{ForestGreen}0} & {\color{ForestGreen}0} & {\color{ForestGreen}0}\\
    {\color{red}0} & {\color{red}0} & {\color{red}0} & {\color{ForestGreen}0} & {\color{ForestGreen}0} & {\color{ForestGreen}0}\\
    {\color{red}1} & {\color{red}0} & {\color{red}1} & {\color{ForestGreen}0} & {\color{ForestGreen}0} & {\color{ForestGreen}0}\\
    {\color{ForestGreen}0} & {\color{ForestGreen}0} & {\color{ForestGreen}0} & {\color{RoyalPurple}0} & {\color{RoyalPurple}0} & {\color{RoyalPurple}0}\\
    {\color{ForestGreen}0} & {\color{ForestGreen}0} & {\color{ForestGreen}0} & {\color{RoyalPurple}0} & {\color{RoyalPurple}0} & {\color{RoyalPurple}0}\\
    {\color{ForestGreen}0} & {\color{ForestGreen}0} & {\color{ForestGreen}0} &{\color{RoyalPurple}0} & {\color{RoyalPurple}0} & {\color{RoyalPurple}0}
\end{pmatrix}\]
Clearly this is rank 2 and it keeps all known entries, so it is feasible. However, it is not a unique solution. Notice that we can easily change some green and purple entries, keep the same rank, and still remain feasible.
\[\widehat{\bR}' =
    \begin{pmatrix} 
    {\color{red}1} & {\color{red}0} & {\color{red}0} &{\color{ForestGreen}0} & {\color{ForestGreen}0} & {\color{ForestGreen}0}\\
    {\color{red}0} & {\color{red}0} & {\color{red}0} & {\color{ForestGreen}0} & {\color{ForestGreen}0} & {\color{ForestGreen}0}\\
    {\color{red}1} & {\color{red}0} & {\color{red}1} & {\color{ForestGreen}1} & {\color{ForestGreen}0} & {\color{ForestGreen}0}\\
    {\color{ForestGreen}1} & {\color{ForestGreen}0} & {\color{ForestGreen}1} & {\color{RoyalPurple}1} & {\color{RoyalPurple}0} & {\color{RoyalPurple}0}\\
    {\color{ForestGreen}0} & {\color{ForestGreen}0} & {\color{ForestGreen}0} & {\color{RoyalPurple}0} & {\color{RoyalPurple}0} & {\color{RoyalPurple}0}\\
    {\color{ForestGreen}0} & {\color{ForestGreen}0} & {\color{ForestGreen}0} &{\color{RoyalPurple}0} & {\color{RoyalPurple}0} & {\color{RoyalPurple}0}
\end{pmatrix}\]
This is still only a rank 2 matrix, which is the minimum rank, but we now have some of the green and purple entries set as non-zeros. Proposition~\ref{thm:rank_min_bad} proves that any feasible optimal solution can be ``reduced'' to one that leaves the green and purple entries as all zeros.

\subsection{Details on our abstraction}
\subsubsection{Proof of observation 
\ref{obs:permutation_invariance}} 
\label{app:abstraction_proofs}
\PermInv*
\begin{proof}[Proof of Obs \ref{obs:permutation_invariance}]
    Let $\bR$ be the original (i.e., non-permuted) ratings matrix. Let $\pi_R$ be a permutation of users (rows) and $\pi_C$ be a permutation of items (columns) and $\mathbf{P}_R$, $\mathbf{P}_C$ the corresponding permutation matrices. Then the permuted ratings matrix is given by 
    \[\bR' = \mathbf{P}_R \bR \mathbf{P}_C.\]
    \textbf{Claim 1}: Let $\widehat{\bR}$ be the rank-$k$ PCA of $\bR$. Then 
    $\widehat{\bR}' = \mathbf{P}_R \widehat{\bR} \mathbf{P}_C$
    is the rank-$k$ PCA of $\bR'$. 

Recall that $\widehat{\bR}$ is a rank-$k$ matrix minimizing the sum of squared errors. For any rank-$k$ matrix $\mathbf{X}$:
\begin{align*}
    \|\bR-\mathbf{X}\|_F^2 & = \sum_{u \in [m]}\sum_{i \in [n]}(r_{ui}-x_{ui})^2\\
    & = \sum_{u \in [m]}\sum_{i \in [n]}(r_{\pi_R(u), \pi_C(i)}-x_{\pi_R(u), \pi_C(i)})^2\\
    & = \|\mathbf{P}_R\bR \mathbf{P}_C -\mathbf{P}_R\mathbf{X} \mathbf{P}_C\|_F^2 \\
    & = \|\bR' -\mathbf{P}_R\mathbf{X} \mathbf{P}_C\|_F^2
\end{align*}
Thus:
\[\widehat{\bR} \in \argmin_{\mathbf{X}: \mathrm{rank}(\mathbf{X})=k}\|\bR-\mathbf{X}\|_F \iff \widehat{\bR}' \in \argmin_{\mathbf{X}: \mathrm{rank}(\mathbf{X})=k}\|\bR'-\mathbf{X}\|_F\]

    \textbf{Claim 2}: Let $\mu$ and $\mu'$ be the recommendations based on $\widehat{\bR}$ and $\widehat{\bR}'$, respectively. Then for all $u \in [m]$:
    \[r_{u, \mu_u} = r'_{\pi_R(u), \mu_{\pi_R(u)}'}.\]
By construction, $\hat{r}'_{\pi_R(u), \pi_C(i)} = \hat{r}_{u, i}$ for all $(u,i) \in [m] \times [n]$. Thus, 
\begin{align*}
    & \mu_u \in \argmax_{i \in [n]}\hat{r}_{u,i} \\
    & \iff \mu_u \in \argmax_{i \in [n]}\hat{r}_{\pi_R(u),\pi_C(i)}'\\
    & \iff \pi_C(\mu_u) \in \argmax_{i \in [n]}\hat{r}_{\pi_R(u),i}'
\end{align*}
Further, $\|\mathbf{R}_{:,i}\|_1=\|\mathbf{R}_{:,\pi_C(i)}\|_1$. Therefore, the recommendation will be the same regardless of ordering. 
\end{proof}
\subsubsection{Example of a majority-minority matrix with a singular value gap}\label{app:maj_min_ex}
\begin{example}\label{ex:matrix}
A very simple majority-minority matrix is a binary matrix of $m$ users and $4$ items where every user likes just one item: $2$ popular items are liked by $m_{\maj}$ users and $2$ less-popular items are liked by $m_{\minor}<m_{\maj}$ users. Ordering users by which item they like, and listing the popular items first the we can write $\bR$ as
\[\bR = \begin{pmatrix}\mathbf{1}_{m_\maj} &  0 & 0 & 0 \\
0 & \mathbf{1}_{m_\maj}  & 0 & 0\\
0 & 0 & \mathbf{1}_{m_\minor} & 0 \\
0 &  0 & 0 & \mathbf{1}_{m_\minor}
\end{pmatrix}\]
where $\mathbf{1}_{m} \in \R^{m}$ is a vectors of all $1$s, one for each of the users who like that item. Likewise:
\[\bR =\begin{pmatrix} \bR_{\maj} & \mathbf{0} \\

\mathbf{0} & \bR_{\minor}\end{pmatrix}\]
where $\bR_{\maj} \in \R^{2m_{\maj}\times 2}=\begin{pmatrix}\mathbf{1}_{m_\maj} &  0 \\
0 & \mathbf{1}_{m_\maj} 
\end{pmatrix}$ has the ratings of all users who like the popular items and $\bR_{\minor} \in \R^{2m_{\minor}\times 2}=\begin{pmatrix}\mathbf{1}_{m_\minor} &  0 \\
0 & \mathbf{1}_{m_\minor} 
\end{pmatrix}$ has the ratings of all users who like the less-popular items. The Singular Value Gap space is $\mathcal{G}(\bR) =(\sqrt{m_{\minor}},\sqrt{m_{\maj}} )$
\end{example}

\subsubsection{Proof of Proposition \ref{rhat_equals_nbar1}}\label{app:rhat_nbar_proof}
\RHatNBar*

\begin{proof}
    Recall that $k_\maj$ is the rank of $\widetilde{\bR}$. Since $\tbR$ is a block diagonal matrix, its singular values are simply the singular values of the blocks. Further, because it has a singular value gap, by definition, all the nonzero singular values of $\tbR_\maj >$ the nonzero singular values of $\tbR_{\minor}$. Thus we have that $\sigma_{k_{\maj}+1}(\tbR) = \sigma_{1}(\tbR_\minor)$. But since $\alpha \in \mathcal{G}(\tbR)$, this means that $\alpha$ is in the singular value gap between $\tbR_\maj$ and $\tbR_\minor$, equivalently, $\alpha > \sigma_1(\tbR_\minor)$. This implies that $k^\star \le k_{\maj}$, since $k^\star$ is the solution to:
\begin{center}
        $\min_{k\in [\mathrm{rank}(\tbR)]} k \quad 
        \text{s.t., } \quad \sigma_{k+1}(\tbR) \leq \alpha$.
\end{center}
so clearly $k_\maj$ as the solution is feasible, but it might not be the smallest possible.

    But likewise, from the fact that $\alpha$ is in the singular value gap between $\tbR_\maj$ and $\tbR_\minor$, we have that:
    $\sigma_{k_{\maj}}(\tbR) = \sigma_{k_{\maj}}(\tbR_{\maj}) > \alpha$, which implies $k^\star \ge k_{\maj}$ (recall that singular values are ordered, so any $\sigma_{k_\maj -i}$ are weakly greater than $\sigma_{k_\maj}$). Thus $k^\star = k_{\maj}$. 
\end{proof}

\section{Supplementary material for Section \ref{sec:users}}\label{app:users}
\subsection{Missing Proofs}\label{app:users_proofs}
\subsubsection{Proof of Theorem~\ref{g_maj_b_min2}} \label{app:users_proofs_gmajbmin}
\GMajBMin*
\begin{proof}
By \Cref{rhat_equals_nbar1}, $\bfR^\star$ is reduced to rank $k_{\maj}$, meaning $\widehat{\bfR}= \sum_{i \in [k_{\maj}]}\sigma_i\mathbf{u_i}\mathbf{v_i}^\top$, where this is a sum over the $k_{\maj}$ largest singular values. Let $U_{\maj}\Sigma_{\maj}V_{\maj}^\top$, $U_{\minor}\Sigma_{\minor}V_{\minor}^\top$ be a SVD for $\bR_{\maj}^\star$ and $\bR_{\minor}^\star$, respectively. Then the following is a SVD for $\bR^\star$:
    \[\begin{pmatrix}U_{\maj} & \mathbf{0} \\
    \mathbf{0} & U_{\minor} \end{pmatrix}\begin{pmatrix}\Sigma_{\maj} & \mathbf{0} \\
    \mathbf{0} & \Sigma_{\minor} \end{pmatrix}\begin{pmatrix}V_{\maj}^\top & \mathbf{0} \\
    \mathbf{0} & V_{\minor}^\top \end{pmatrix} = \begin{pmatrix}U_{\maj}\Sigma_{\maj}V_{\maj}^\top & \mathbf{0} \\
    \mathbf{0} & U_{\minor}\Sigma_{\minor}V_{\minor}^\top \end{pmatrix}.
    \]
Note that
\[ 
    \mathbf{\Sigma} := \begin{pmatrix}\Sigma_{\maj} & \mathbf{0} \\
    \mathbf{0} & \Sigma_{\minor} \end{pmatrix}
    \]
is not necessarily a perfectly ordered diagonal of singular values because even if $\Sigma_{\maj}$ and $\Sigma_{\minor}$ are ordered, we do not assume full rank of $\bR_{\maj}^\star$, meaning that some diag elements of $\Sigma_{\maj}$ may be $0$. 

However by Assumption \ref{assumption:singular_value_gap}, the first $k_{\maj}$ singular values belong to $\Sigma_{\maj}$. Also using the definition of $k_{\maj}$ as the rank of $\bR_{\maj}^\star$,
\[\widehat{\bR} = \begin{pmatrix}U_{\maj}\Sigma_{\maj}V_{\maj}^\top & \mathbf{0} \\
    \mathbf{0} & \mathbf{0} \end{pmatrix} = \begin{pmatrix} \bR_{\maj}^\star & \mathbf{0} \\
\mathbf{0} & \mathbf{0}\end{pmatrix}\]

Firstly, we will show that for all $u \in \calUmaj, r_{u,i_u^\star}^\star=\max_{i \in [n]}r_{u,i}^\star$.
Because all ratings for majority users are preserved, for all $u \in \calUmaj$:
\[\argmax_{i \in [n]}\hat{r}_{u,i}=\argmax_{i \in [n]}r_{u,i}^\star.\]
Hence, $i_u^\star \in \argmax_{i \in [n]}r_{u,i}^\star$ and $r_{u,i_u^\star}^\star=\max_{i \in [n]}r_{u,i}^\star$.

Secondly, we will show that for all $u \in \calUmin, r_{u,i_u^\star}^\star=0$. For all minority users:
\[\argmax_{i \in [n]}\hat{r}_{u,i} = [n]\]
since their ratings are represented by a vector of $0$s. 
Therefore, they will be recommended an item from the arg max of 
\begin{equation*}
    \begin{aligned}
        \text{maximize}_{i \in [n]} \|\widehat{\bR}_{i}\|_1
    \end{aligned}
\end{equation*}
But clearly (because $\widehat{\bR}$ is simply the zero-padded majority matrix) this can be rewritten as 
\begin{equation*}
    \begin{aligned}
        \text{maximize}_{i \in [\barn]} \|\bR_{\maj_i}^\star\|_1
    \end{aligned}
\end{equation*}
By assumption, $\sum_{u \in \calUmin}r_{u,i}^\star=0$ for all $i \in [\bar{n}]$. Thus for all $u \in \calUmin$, $i_u^\star \in [\bar{n}]$ and $r_{u,i_u^\star}^\star=0$.
\end{proof}

\subsubsection{Proof of Theorem~\ref{thm:SW_EA_Model1}} \label{app:users_proofs_sw_improve}
\SWImprove*
We first refer to known lower-bounds on matrix singular values when appending a column:
\begin{lemma}[Corollary 3.5 of~\citet{chretien_perturbation_2014}]\label{thm:singular_val_LB1}
    Let $d$ be a positive integer and let $\mathbf{M} \in \mathbb{R}^{d \times d}$ be a positive semi-definite matrix with rank $k\leq d$, whose eigenvalues are $\lambda_1 \ge \dots \ge \lambda_d$. For some $\mathbf{a} \in \R^d$, and $c \in \R$ let $\mathbf{A}$ be given by 
    \[\mathbf{A} = \begin{pmatrix}c & \mathbf{a}^\top\\
    \mathbf{a} & \mathbf{M}\end{pmatrix}\]
    Then 
    \[\lambda_{k+1}(\mathbf{A}) \ge \min(c,\lambda_k)-\|\mathbf{a}\|_2.\] 
\end{lemma}

And bounds on matrix singular values when removing columns:
\begin{lemma}[Corollary 7.3.6 of~\citet{Horn_Johnson_2012}] \label{thm:remove_cols}
Let $\mathbf{A} \in \mathbb{C}^{m\times n}$ be a hermitian matrix and let $\hat{\mathbf{A}}\in \mathbb{C}^{m\times (n-1)}$ or $\in \mathbb{C}^{(m-1)\times n}$ be a hermitian matrix obtained by deleting any column or row from $\mathbf{A}$. Define $r:=\mathrm{rank}(\mathbf{A})$ and $\sigma_r(\hat{\mathbf{A}}) = 0$ if $m\geq n$ and a column is deleted or if $m \leq n$ and a row is deleted. Then:
\begin{equation*}
    \sigma_1(\mathbf{A}) \geq \sigma_1(\hat{\mathbf{A}}) \geq \sigma_2(\mathbf{A}) \geq \sigma_2(\hat{\mathbf{A}}) \geq \dots \geq \sigma_r(\mathbf{A}) \geq \sigma_r(\hat{\mathbf{A}})
\end{equation*}
\end{lemma}
\begin{proof}[Proof of Theorem \ref{thm:SW_EA_Model1}]

WLOG we shall assume that $i^\star = \bar{n}+1$ and $\calUmaj\cup \calU_{i^\star} = [m_1]$ (see \Cref{obs:permutation_invariance}).  
In order to prove that social welfare is the sum of majority AND picky item users' top ratings, we shall go first prove the following claims:
\begin{claim}\label{claim:SW_EA_Model1_1}
    Given $\tbR$, the learner will pick $k^\star = k_{\maj}+1$.
\end{claim}
\begin{claim}\label{claim:SW_EA_Model1_2}
    Let $\widehat{\bR}$ be the rank $k^\star$ SVD truncation of $\tbR$. We have that 
 \[\hat{r}_{u,i}=\begin{cases} r_{u,i}^\star & u \in \calUmaj, i \in [\bar{n}]\\
 \tilde{r}_{u,i} & u \in (\calUmaj\cup \calU_{i^\star}), i=i^\star \\
 0 & ow
 \end{cases}.\]
\end{claim}
\begin{proof}[Proof of Claim~\ref{claim:SW_EA_Model1_1}]

Recall that the learner picks $k^\star$ by solving

\begin{align*}
    & \min_{k \le \min\{m,n\}}k \\
    & \text{s.t. }\sigma_{k+1}(\tbR) \le \alpha
\end{align*}

First, we show that $\sigma_{k_{\maj}+1}(\tbR)> \alpha$ which implies that $k^\star \ge k_{\maj}+1$. 

We shall show this by invoking Lemma \ref{thm:singular_val_LB1} to bound the $(k_\maj +1)$th singular value of a matrix $\tilde{\mathbf{X}} \in \R^ {m\times 
(\barn +1)}$ made up of the first $\barn +1$ columns of $\tbR$. We will then show that this is weakly smaller than $\sigma_{k_{\maj}+1}(\tbR)$ using lemma \ref{thm:remove_cols}. 

Let matrix $\mathbf{X} \in \R^ {m\times 
\barn}$ be a matrix made up of the first $\barn$ columns of $\bR^\star$ (or equivalently, $\tbR$).

Construct $\mathbf{A} \in \R^{(\barn+1)\times (\barn+1)}$ according to Lemma \ref{thm:singular_val_LB1}: with $\mathbf{X}^\top\mathbf{X} \in \bR^{\barn \times \barn}$, $c = \tbR_{i^\star}^\top\tbR_{i^\star} \in \R$, $\mathbf{a} = \mathbf{X}^\top \tbR_{i^\star} \in \R^m$.  This satisfies the conditions of Lemma \ref{thm:singular_val_LB1} when $k=k_{\maj}$ and $d = \bar{n}$. Evaluating for each value in the bound of Lemma \ref{thm:singular_val_LB1}: 
\begin{align*}
    c & = \tbR_{i^\star}^\top\tbR_{i^\star} = \sum_{u \in [m]}\tilde{r}_{u,i^\star}^2 \\
    & = \sum_{u \in \calU_\COLL}\eta^2 +\sum_{u \in \calU_{i^\star}}(r_{u,i^\star}^{\star})^2 &\quad \text{Definition \ref{defn:picky_items}}\\
    & = \eta^2 |\calU_\COLL| + \|\bR_{i^\star}^\star\|_2^2 &\quad \text{By construction of collective strategy} \\
\end{align*}

Additionally:
\begin{align*}
    \|\mathbf{a}\|_2^2 & = \sum_{i \in [\bar{n}]}\left(\mathbf{X}^\top \tbR_{i^\star}\right)_i^2  = \sum_{i \in [\bar{n}]}\left(\sum_{u \in [m]}r_{u,i}^\star\tilde{r}_{u,i^\star}\right)^2 \\
    & = \sum_{i \in [\bar{n}]}\left(\sum_{u \in \calUmaj}r_{u,i}^\star\tilde{r}_{u,i^\star}\right)^2  &\quad \text{exclusivity of $\calImaj$ and $\calImin$}\\
    & =  \sum_{i \in [\bar{n}]}(\eta \sum_{u \in \calU_\COLL}r_{u,i}^\star)^2 =\eta^2 \sum_{i \in [\bar{n}]}(\sum_{u \in \calU_\COLL}r_{u,i}^\star)^{2} &\quad \text{By construction of collective strategy}\\
    & \le \eta^2 \bar{n} \max_{i \in [\bar{n}]} (\sum_{u \in \calU_\COLL}r_{u,i}^\star)^2\\
    & = \bar{n}(\eta \AV(\bR^\star, \calU_\COLL))^2
\end{align*}

Thus $\|\mathbf{a}\|_2 \le  \eta \sqrt{\bar{n}}\AV(\bR^\star, \calU_\COLL)$. 

To get the bound of Lemma \ref{thm:singular_val_LB1}, we also need singular values of $\mathbf{X}$ (equivalently, eigenvalues of $\mathbf{M}:=\mathbf{X}^\top\mathbf{X}$). Clearly, the non-zero singular values of $\mathbf{X}$ and $\bR_{\maj}^\star$ are the same because $\mathbf{X}$ is simply $\bR_{\maj}^\star$ but padded with zeroes, thus:
\[\lambda_{k_{\maj}} = \sigma_{k_{\maj}}(\mathbf{X})^2 =\sigma_{k_{\maj}}(\bR_{\maj}^\star)^2\]

We have then that 
\begin{align*}
    \sigma_{k_{\maj}+1}(\tbR)^2 & \ge \sigma_{k_{\maj}+1}(\tilde{\mathbf{X}})^2 &\quad \text{Lemma \ref{thm:remove_cols}}\\
    & = \lambda_{k_{\maj}+1}(\mathbf{A}) &\quad \text{By construction of $\mathbf{A}$}\\
    & \ge \min(c,\lambda_{k_{\maj}}) -\|\mathbf{a}\|_2 &\quad \text{Lemma \ref{thm:singular_val_LB1}}\\
    & \ge \min (\eta^2 |\calU_\COLL| + \|\bR_{i^\star}^\star\|_2^2, \sigma_{k_{\maj}}(\bR_{\maj}^\star)^2) - \eta \sqrt{\bar{n}}\AV(\bR^\star, \calU_\COLL) \\
    & > \alpha^2 &\quad \text{Theorem \ref{thm:SW_EA_Model1} assumption}
\end{align*}
This implies that $\sigma_{k_{\maj}+1}(\tbR) > \alpha$ and that $k^\star \ge k_{\maj}+1$.

Now we will show that $k^\star \leq k_{\maj}+1$. Note that for all $i > \bar{n}+1$ and $u \in \calUmaj \cup \calU_{i^\star}$, $\tilde{r}_{u,i}=0$. Likewise, for all $i \le \bar{n}+1$
and $u \notin \calUmaj \cup \calU_{i^\star}$, $\tilde{r}_{u,i}=0$. WLOG and because $i^\star$ is picky, we can represent $\tbR$ as a block diagonal matrix: 
\[\tbR = \begin{pmatrix}\tbR_{\maj '} & \mathbf{0} \\ \mathbf{0} & \tbR_{\minor '} \end{pmatrix}\]
where $\tbR_{\maj '} \in \R^{m_1\times \bar{n}+1}$ has reported ratings for users $u \in \calUmaj \cup \calU_{i^\star}$ and items $i \le \bar{n}+1$ and $\tbR_{\minor '} \in \R^{(m-m_1)\times n-(\bar{n}+1)}$ has reported ratings for users $u \in (\calUmin \setminus \calU_{i^\star})$ and items $i > \bar{n}+1$.

Recall that the singular values of a block diagonal matrix are simply a concatenation of the singular values of the two blocks. Since $\mathrm{rank}(\tbR_{\maj '}) \le \mathrm{rank}(\bR_{\maj}^\star)+1= k_{\maj}+1$, it has no more than $k_{\maj}+1$ nonzero singular values. It follows that at least one of the $k_{\maj}+2$ largest singular values of $\tbR$ are a singular value of $\tilde{\bR}_2$. 
    Therefore:
    \begin{align*}
    \sigma_{k_{\maj}+2}(\tbR) & \le \sigma_{1}(\tbR_{\minor '}) \\
    &\leq \sigma_{1}(\bR_{\minor}^\star) &\quad \text{Lemma \ref{thm:remove_cols}}\\
    & < \alpha &\quad \text{Theorem \ref{thm:SW_EA_Model1} assumption}
\end{align*}
    
This implies that $k^\star \le k_{\maj}+1$. So $k^\star = k_{\maj}+1$ as desired. 
\end{proof}

\begin{proof}[Proof of Claim~\ref{claim:SW_EA_Model1_2}]

Recall that the $k^\star$-truncated SVD is $\sum_{j\in[k^\star]}\sigma_j\mathbf{u_j}\mathbf{v_j}^\top$ where $[k^\star]$ are the $k^\star$ largest singular values. In the above claims we showed that $\sigma_{k_{\maj}+1}(\tbR) \geq \min \{\sigma_{k_{\maj}}(\bR_{\maj}^\star)^2, \eta^2|\calU_\COLL|+\|\bR^\star_{i^\star}\|_2^2\}-\eta\sqrt{\bar{n}}\AV(\bR^\star, \calU_\COLL)$ and so by assumption, the $k_{\maj}+1$ largest singular values are all strictly greater than $\sigma_1(\tbR_{\minor '})$.

In the proof of Theorem \ref{g_maj_b_min2}, we showed that the SVD of a block diagonal matrix can be decomposed into a sum of terms for each block. 
Therefore, because the $k_{\maj}+1$ largest singular values are all strictly greater than $\sigma_1(\tbR_{\minor '})$, it must be case that for $\tbR$, $\sum_{j\in[k_{\maj}+1]}\sigma_j\mathbf{u_j}\mathbf{v_j}^\top$, where $[k_{\maj}+1]$ are the $k_{\maj}+1$ largest singular values, form the following matrix:
\[\begin{pmatrix}\tbR_{\maj '} & \mathbf{0} \\ \mathbf{0} & \mathbf{0} \end{pmatrix}\]

Of course, from Claim~\ref{claim:SW_EA_Model1_1}, we have that $k^\star = k_{\maj}+1$, thus this completes Claim~\ref{claim:SW_EA_Model1_2}.
\end{proof}

Now, we shall use Claims \ref{claim:SW_EA_Model1_1} and \ref{claim:SW_EA_Model1_2} to prove our Theorem result. 

For all $u \in \calUmaj$ and all $i \neq i^\star$, $r_{u,i}^\star=\hat{r}_{u,i}$. To show that user $u$ will be recommended a top item it therefore suffices to show that picky item, $i^\star$, will not become a top item. This is true by construction:
\[\hat{r}_{u,i^\star} \le \eta < \min_{u' \in \calUmaj}\max_{i \in [n]}r_{u'i}^\star \le   \max_{i \in [n]}r_{u,i}^\star \quad \forall u \in \calUmaj\]
For all $u \in \calU_{i^\star}$, $r_{u,i}^\star=\hat{r}_{u,i}$ for all $i \in [n]$. Thus $\argmax_{i \in [n]}\hat{r}_{u,i}=\argmax_{i \in [n]}r_{u,i}^\star=i^\star$. 

Lastly, for all $u \in \calUmin \setminus \calU_{i^\star}$:
\[\argmax_{i \in [n]}\hat{r}_{u,i}=[n]\]
Therefore, we will have 
\[\Topk(u) \in \argmax_{i \in [n]}\|\widehat{\bR}_{i}\|_1 = \argmax_{i \in [\bar{n}+1]}\|\widehat{\bR}_{i}\|_1 \subseteq [\bar{n}+1]\]
and recall that $r_{u,i}^\star = 0$ $\forall i\in [\barn +1], u\in \calUmin \setminus \calU_{i^\star}$\\
As such: 
\begin{align*}
    \SW(\tbR, \alpha) & = \sum_{u \in \calUmaj}\max_{i \in [n]}r_{u,i}^\star + \sum_{u \in \calU_{i^\star}}\max_{i \in [n]}r_{u,i}^\star + \sum_{u \in \calUmin}0\\
    & = \sum_{u \in \calUmaj\cup \calU_{i^\star}}\max_{i \in [n]}r_{u,i}^\star.
\end{align*}

\end{proof}

\subsubsection{Proof of Corollary~\ref{cor:suff_cond}}
\SuffCond*
\begin{proof}[Proof of Corollary \ref{cor:suff_cond}]
    Suppose we have an altruistic rating $(\eta, \calU_\COLL)$ for the picky item such that the conditions above hold. Then we must have that 
    \[
    \alpha \in (\sigma_{1}(\bR_{\minor}^\star), \sqrt{\min \{\sigma_{k_{\maj}}(\bR_{\maj}^\star)^2, \eta^2|\calU_\COLL|+\|\bR^\star_{i^\star}\|^2_2\}-\eta\sqrt{\bar{n}}\AV(\bR^\star, \calU_\COLL)})
    \]
    The space in this interval is $\mathcal{G}(\bR^\star, \calU_\COLL, \eta)$ by definition. Equivalently, $\alpha \in \mathcal{G}(\bR^\star, \calU_\COLL, \eta)$. Since $\alpha$ clearly exists, $\mathcal{G}(\bR^\star, \calU_\COLL, \eta) \neq \emptyset$, which means $\bR^\star$ has $(\eta, \calU_\COLL)$-sufficient singular value gap. Thus by Theorem \ref{thm:SW_EA_Model1}, social welfare is improved by the manipulated matrix.
\end{proof}

\subsubsection{Proof of Theorem~\ref{thm:find_eta_algo}}\label{app:users_proofs_find_eta}
\FindEta*
\begin{proof} [Proof of thm \ref{thm:find_eta_algo}]
Let $\eta$ be the returned output of Algorithm \ref{alg:find_eta_aglo}. Note that the index of the picky item is $\barn + 1$ without loss of generality to any $i^\star >\barn$, see \Cref{obs:permutation_invariance}. Thus we will return to $i^\star$ as if it were $\barn + 1$ for the sake of this proof. There are two parts to Theorem \ref{thm:find_eta_algo} that we present as claims and prove sequentially for the cases when $\eta$ returned by Algorithm \ref{alg:find_eta_aglo} is positive or zero. For cleanliness of notation, we will refer to $\AV(\bR^\star, \calU_\COLL)$ as simply $\AV$ and $\|\bR^\star_{i^\star}\|^2_2$ as $\ASV$.

\begin{claim}\label{thm:find_eta_algo_claim_1}
    If $\eta >0$ is returned by Algorithm \ref{alg:find_eta_aglo}, then playing $\eta$ will improve social welfare.
\end{claim}

\begin{proof}[Proof of Claim~\ref{thm:find_eta_algo_claim_1}]
    By Corollary \ref{cor:suff_cond}, it is sufficient to show that $\eta$ (when $\eta \neq 0$) satisfies the following:
\begin{enumerate}
    \item $\alpha^2 < \min \{\sigma_{k_{\maj}}(\bR_{\maj}^\star)^2, \eta^2|\calU_\COLL|+\ASV\}-\eta\sqrt{\bar{n}}\AV$
    \item $\eta > 0$
    \item $\eta < \kappa$
\end{enumerate}
We first focus on the inequality: 
\begin{equation} \label{alpha_ub}
    \alpha^2 < \min \{\sigma_{k_{\maj}}(\bR_{\maj}^\star)^2, \eta^2|\calU_\COLL|+\ASV\}-\eta\sqrt{\bar{n}}\AV
\end{equation}

Where $\eta > 0$, the inequality above is equivalent to both of the following statements holding: 
\begin{equation} \label{eta_ub}
\begin{aligned}
        & \alpha^2 < \sigma_{k_{\maj}}(\bR_{\maj}^\star)^2-\eta\sqrt{\bar{n}}\AV
        & \iff \eta < \frac{\sigma_{k_{\maj}}(\bR_{\maj}^\star)^2-\alpha^2}{\sqrt{\bar{n}}\AV}
    \end{aligned}
\end{equation}
and 
\begin{equation} \label{eta_lb}
\begin{aligned}
        & \alpha^2 < \eta^2|\calU_\COLL|+\ASV-\eta\sqrt{\bar{n}}\AV
        & \iff  \eta^2|\calU_\COLL|-\eta \sqrt{\bar{n}}\AV +\ASV-\alpha^2 >0 
\end{aligned}
\end{equation}
Clearly equation \ref{eta_ub} is an upper bound on $\eta$. We shall analyze equation \ref{eta_lb} to get explicit bounds on $\eta$

Let $f(\eta)$ be the quadratic of equation \ref{eta_lb} in terms of $\eta$ with discriminant $d := \bar{n}\AV^2+4|\calU_\COLL|(\alpha^2-\ASV)$. Now we need to understand for which set of $\eta \in \R$, $f(\eta)> 0$. Notice that, by standard properties of quadratic functions, if $d \geq 0$, $f(\eta) > 0$ where $\eta \in \left[\frac{\sqrt{\bar{n}}\AV -\sqrt{d}}{2|\calU_\COLL|}, \frac{\sqrt{\bar{n}}\AV +\sqrt{d}}{2|\calU_\COLL|}\right]^C$ and if $d < 0$, $f(\eta) > 0 \quad \forall \eta \in \R$. Consequently, the set of feasible $\eta$ for equation \ref{alpha_ub} to hold break into the following cases:
\begin{enumerate}
    \item Case 1: $d < 0$, therefore equation \ref{eta_lb} does not constrain $\eta$ and only equation \ref{eta_ub} and positivity is important: \[\eta \in \bigg(0, \frac{\sigma_{k_{\maj}}(\bR_{\maj}^\star)^2-\alpha^2}{\sqrt{\bar{n}}\AV}\bigg)\]
    \item Case 2: $d \geq 0$, $\eta$ must be feasible for both equation \ref{eta_lb} and \ref{eta_ub} and positive.
    \[\eta \in \left[\frac{\sqrt{\bar{n}}\AV -\sqrt{d}}{2|\calU_\COLL|}, \frac{\sqrt{\bar{n}}\AV +\sqrt{d}}{2|\calU_\COLL|}\right]^C \cap \bigg(0, \frac{\sigma_{k_{\maj}}(\bR_{\maj}^\star)^2-\alpha^2}{\sqrt{\bar{n}}\AV}\bigg)\]
\end{enumerate}
Note that $\frac{\sigma_{k_{\maj}}(\bR_{\maj}^\star)^2-\alpha^2}{\sqrt{\bar{n}}\AV} > 0$ because $\sigma_{k_{\maj}}(\bR_{\maj}^\star) > \alpha$ by setting assumptions.

We can further rewrite Case 2. Notice that by setting $\nabla f = 0$, the minimum of $f(\eta)$ is at $\eta = \frac{\sqrt{\barn} \AV}{2|\calU_\COLL|}$ which is greater than 0 by setting assumptions. Thus, it must be that $\frac{\sqrt{\bar{n}}\AV +\sqrt{d}}{2|\calU_\COLL|} > 0$ because $\frac{\sqrt{\bar{n}}\AV +\sqrt{d}}{2|\calU_\COLL|}$ is the right hand root $f(\eta)$.
\begin{enumerate}
    \item Case 1: $d < 0$ \[\eta \in \bigg(0, \frac{\sigma_{k_{\maj}}(\bR_{\maj}^\star)^2-\alpha^2}{\sqrt{\bar{n}}\AV}\bigg)\]
    \item Case 2: $d \geq 0$,
    \[\eta \in \left(0, \min(\frac{\sqrt{\bar{n}}\AV -\sqrt{d}}{2|\calU_\COLL|},\frac{\sigma_{k_{\maj}}(\bR_{\maj}^\star)^2-\alpha^2}{\sqrt{\bar{n}}\AV} )\right)\cup \left(\frac{\sqrt{\bar{n}}\AV +\sqrt{d}}{2|\calU_\COLL|}, \frac{\sigma_{k_{\maj}}(\bR_{\maj}^\star)^2-\alpha^2}{\sqrt{\bar{n}}\AV}\right)\]
\end{enumerate}

Now we finally add the $\eta < \kappa$ to the sufficient conditions. this becomes a part of both case's upper bounds:
\begin{enumerate}
    \item Case 1: $d < 0$, \[\eta \in \bigg(0, \min(\kappa, \frac{\sigma_{k_{\maj}}(\bR_{\maj}^\star)^2-\alpha^2}{\sqrt{\bar{n}}\AV})\bigg)\]
    \item Case 2: $d \geq 0$,
    \[\eta \in \left(0, \min(\kappa, \frac{\sqrt{\bar{n}}\AV -\sqrt{d}}{2|\calU_\COLL|},\frac{\sigma_{k_{\maj}}(\bR_{\maj}^\star)^2-\alpha^2}{\sqrt{\bar{n}}\AV} )\right)\cup \left(\frac{\sqrt{\bar{n}}\AV +\sqrt{d}}{2|\calU_\COLL|}, \min(\kappa, \frac{\sigma_{k_{\maj}}(\bR_{\maj}^\star)^2-\alpha^2}{\sqrt{\bar{n}}\AV})\right)\]
\end{enumerate}
Note that we have shown that these cases are equivalent to the sufficient conditions we must prove are met.

It is easy to see that in either case, when the relevant space is non-empty, Algorithm \ref{alg:find_eta_aglo} returns an $\eta$ in the space because the algorithm first checks the discriminant and then constructs the relevant range(s) (if nonempty).
\end{proof}

\begin{claim}\label{thm:find_eta_algo_claim_2}
    Algorithm \ref{alg:find_eta_aglo} returns $0$ if and only if there is no $\eta$ correlated strategy that satisfies our feasible conditions.
\end{claim}
\begin{proof}[Proof of Claim~\ref{thm:find_eta_algo_claim_2}]
In our proof of Claim~\ref{thm:find_eta_algo_claim_1}, we showed that an equivalent way to characterize an $\eta$ that satisfies our sufficient conditions for SW improvement is the following:
\begin{enumerate}
    \item Case 1: $d < 0$, 
    \[\eta \in (0, \min(u, \kappa) )\]
    \item Case 2: $d \geq 0$,
    \[\eta \in \left(0, \min(\kappa, r_1,u )\right)\cup \left(r_2, \min(\kappa, u)\right)\]
\end{enumerate}
Where 
\[
d:=\bar{n}\AV^2+4|\calU_\COLL|(\alpha^2-\ASV)
\]
\[
u := \frac{\sigma_{k_{\maj}}(\bR_{\maj}^\star)^2-\alpha^2}{\sqrt{\bar{n}}\AV})
\]
\[
r_1:=\frac{\sqrt{\bar{n}}\AV -\sqrt{d}}{2|\calU_\COLL|}
\]
\[
r_2 :=\frac{\sqrt{\bar{n}}\AV +\sqrt{d}}{2|\calU_\COLL|}
\]
and $\barn, |\calU_EA|, \sigma_{k_\maj}(\bR_\maj), \kappa, \ASV, \AV$ are parameters of the algorithm

Suppose there does not exist an $\eta$ that is feasible according to our cases. It must be the case that $d\geq 0$, because otherwise there is clearly always a feasible $\eta$ as $\kappa, u > 0$ by setting assumptions. Since there is no feasible $\eta$, it must be the case that  \[\eta \in \left(0, \min(\kappa, \frac{\sqrt{\bar{n}}\AV -\sqrt{d}}{2|\calU_\COLL|},\frac{\sigma_{k_{\maj}}(\bR_{\maj}^\star)^2-\alpha^2}{\sqrt{\bar{n}}\AV} )\right)\cup \left(\frac{\sqrt{\bar{n}}\AV +\sqrt{d}}{2|\calU_\COLL|}, \min(\kappa, \frac{\sigma_{k_{\maj}}(\bR_{\maj}^\star)^2-\alpha^2}{\sqrt{\bar{n}}\AV})\right)\]
is an empty space. Algorithm \ref{alg:find_eta_aglo} first checks the LHS set. If it is empty, it checks the RHS, and if that is empty, it returns $0$. Therefore if $\eta$ is infeasible, $0$ is returned.

Suppose $0$ is returned by the algorithm. It clearly could not have been the case that $d < 0$ because given $\kappa, u > 0$, $d<0$ would never result in a returned $0$. Thus we consider $d\geq 0$. In this case, evaluating the If statements, $0$ is clearly returned only if 
\[\eta \in \left(0, \min(\kappa, \frac{\sqrt{\bar{n}}\AV -\sqrt{d}}{2|\calU_\COLL|},\frac{\sigma_{k_{\maj}}(\bR_{\maj}^\star)^2-\alpha^2}{\sqrt{\bar{n}}\AV} )\right)\cup \left(\frac{\sqrt{\bar{n}}\AV +\sqrt{d}}{2|\calU_\COLL|}, \min(\kappa, \frac{\sigma_{k_{\maj}}(\bR_{\maj}^\star)^2-\alpha^2}{\sqrt{\bar{n}}\AV})\right)\]
is empty. Thus we have that if $0$ is returned, there must be no feasible $\eta$ (for our sufficient conditions).
\end{proof}
Having proven both Claims, we have shown both parts of Theorem \ref{thm:find_eta_algo} hold, so we may conclude the full proof.
\end{proof}

\RobustAlgo*
\subsubsection{Proof of Theorem~\ref{thm:perturbed_params}}\label{app:users_proofs_robust}
In order to prove this, we will invoke Lipschitz bounds on a function that is a minimum of two Lipschitz functions, so the following lemma will be helpful:
\begin{lemma}\label{lemma:lipschitz}
    Let $f(\bz) = \min \{f_1(\bz), f_2(\bz)\}$ where $f_1$ and $f_2$ are Lipschitz on a convex region $\calD$ with constants $L_1$ and $L_2$, respectively. Then $f$ is Lipschitz on $\calD$ with constant $L = \max\{L_1,L_2\}$. 
\end{lemma}

\begin{proof}[Proof of Theorem \ref{thm:perturbed_params}]
    Consider two arbitrary points $\bz_1, \bz_2 \in \calD$. Assume without loss of generality that $f(\bz_1)\ge f(\bz_2)$. If $f_1(\bz_1) \ge f_2(\bz_1)$:
    \begin{align*}
        |f(\bz_1) -f(\bz_2)| & = |f_1(\bz_1) - \min\{f_1(\bz_2), f_2(\bz_2)\}|\\
        & = |f_1(\bz_1) + \max\{-f_1(\bz_2), -f_2(\bz_2)\}|\\
        & = | \max\{f_1(\bz_1)-f_1(\bz_2), f_1(\bz_1)-f_2(\bz_2)\}|\\
        & \le | \max\{f_1(\bz_1)-f_1(\bz_2), f_2(\bz_1)-f_2(\bz_2)\}|\\
        & \le \max \{|f_1(\bz_1)-f_1(\bz_2)|, |f_2(\bz_1)-f_2(\bz_2)|\}\\
        & \le \max \{L_1 \|\bz_1-\bz_2\|, L_2 \|\bz_1-\bz_2\|_2\}\\
        & = \max\{L_1, L_2\}\|\bz_1-\bz_2\|_2
    \end{align*}
By making a symmetric argument for $f_1(\bz_1)< f_2(\bz_1)$ we get the same bound. Thus, $f$ is Lipschitz on $\calD$ with constant $L = \max\{L_1, L_2\}$ as desired.  
\end{proof}

With this lemma, we will now proceed with the full proof.

\begin{proof}[Proof of Theorem \ref{thm:perturbed_params}]
For simplicity of notation, we will denote $\sigma_{k_{maj}}(\bR_{\maj}^\star)$ as $\sigma$, $\AV(\bR^\star, \calU_\COLL)$ as simply $\AV$ and $\|\bR^\star_{i^\star}\|^2_2$ as $\ASV$. 

Recall from Corollary \ref{cor:suff_cond}, it suffices to show that 
    \begin{enumerate}
        \item $\hat{\eta} > 0$
        \item $\hat{\eta} < \kappa$
        \item $0 < \min \{(\sigma^\star)^2, \hat{\eta}^2|\calU_\COLL^\star|+\ASV^\star\}-\hat{\eta}\sqrt{\bar{n}^\star}\AV^\star - (\alpha^\star)^2$
    \end{enumerate}
for the true parameters, $\bz^\star$ and the returned $\hat{\eta}$. Because $\hat{\eta}>0$, by assumption, clearly the first condition is satisfied. Also, the 2nd condition must be satisfied because $\hat{\eta}$ was returned by Algorithm \ref{alg:find_eta_aglo}, and Theorem \ref{thm:find_eta_algo} asserts that any nonzero $\eta$ returned by the algorithm satisfies $\hat{\eta} < \kappa$.

Thus all we must prove is that $f(\bz^\star;\hat{\eta}) > 0$ given that $ f(\hat{\mathbf{z}};\hat{\eta})>0$ by Theorem \ref{thm:find_eta_algo}. Note that we can equivalently write $f(\mathbf{z};\eta) = \min \{f_1(\mathbf{z};\eta), f_2(\mathbf{z};\eta)\}$ where
\[f_1(\mathbf{z};\eta) = \sigma^2 -\eta\sqrt{\bar{n}}\AV -\alpha^2\]
\[f_1(\mathbf{z};\eta) = \eta^2 |\calU_\COLL|+\ASV -\eta\sqrt{\bar{n}}\AV -\alpha^2\]
Because $\eta, \AV, \barn, \alpha > 0$ by setting assumptions.

We note that for the remainder of this proof, we are exclusively interested in $f(\cdot;\hat{\eta})$ (i.e. $f$ with the returned $\hat{\eta}$ as the parameter), so for notational simplicity, we often drop parameter, $\hat{\eta}$.

Assume that $f$ is $L$-Lipschitz on some region $\calD$ (we will prove later that this is true for a suitably defined $\mathcal{D}$). Then, given $\|\bz-\hat{\bz}\|_2 < \frac{\Delta(\hat{\bz})}{L}$ and $\hat{\bz}, \bz \in \calD$:
\begin{align*}
    f(\bz) & = f(\hat{\bz}) + (f(\bz) -f(\hat{\bz}))\\
    & \ge f(\hat{\bz}) - |f(\bz)-f(\hat{\bz})|\\
    & \ge f(\hat{\bz}) -L\|\bz-\hat{\bz}\|_2 \quad \text{$f$ is $L$-Lipschitz}\\
    & > f(\hat{\bz}) - f(\hat{\bz}) \quad \text{by assumption}\\
    &=0
\end{align*}
Therefore, to show the desired $f(\hat{\mathbf{z}};\hat{\eta})>0$, it suffices to show that $f$ is $L$-Lipschitz on  $\mathcal{D}$. We will do this by finding a bound on the maximum gradient norm of $f_1$ and $f_2$. First we will compute the gradient of both functions:
\begin{align*}
    \|\nabla f_1(\bz)\|_2^2 & = \|(\frac{\partial f_1}{\partial \sigma}(\mathbf{z}), \frac{\partial f_1}{\partial \alpha}(\mathbf{z}), \frac{\partial f_1}{\partial \bar{n}}(\mathbf{z}), \frac{\partial f_1}{\partial \ASV}(\mathbf{z}), \frac{\partial f_1}{\partial \AV}(\mathbf{z}), \frac{\partial f_1}{\partial |\mathcal{U}_{EA}|}(\mathbf{z}))^\top\|\\
    & = \|(2\sigma, -2\alpha, \frac{\hat{\eta} \AV}{2\sqrt{\bar{n}}},0, -\hat{\eta} \sqrt{\bar{n}},0)^\top\|_2^2 =4\sigma^2 + 4 \alpha^2 + \frac{\hat{\eta} \AV^2}{4\sqrt{\bar{n}}}+\hat{\eta}^2 \bar{n} & \quad \text{Definition of $f_1(\bz)$}\\
\end{align*}
\begin{align*}
    \|\nabla f_2(\bz)\|_2^2 & = \|(\frac{\partial f_2}{\partial \sigma}(\mathbf{z}), \frac{\partial f_2}{\partial \alpha}(\mathbf{z}), \frac{\partial f_2}{\partial \bar{n}}(\mathbf{z}), \frac{\partial f_2}{\partial \ASV}(\mathbf{z}), \frac{\partial f_2}{\partial \AV}(\mathbf{z}), \frac{\partial f_2}{\partial |\mathcal{U}_{EA}|}(\mathbf{z}))^\top\|\\
    & = \|(0, -2\alpha, \frac{\hat{\eta} \AV}{2\sqrt{\bar{n}}}, 1, -\hat{\eta} \sqrt{\bar{n}}, \hat{\eta}^2)^\top\|_2^2 = 4 \alpha^2 + \frac{\hat{\eta} \AV^2}{4\sqrt{\bar{n}}}+1+\hat{\eta}^2 \bar{n}+\hat{\eta}^4& \quad \text{Definition of $f_2(\bz)$}\\
\end{align*}

To bound the gradient norms we will use bounds on the parameters in terms of the primitives of the matrix. Recall that $\hat{\eta}$ is fixed parameter (it is what Algorithm \ref{alg:find_eta_aglo} recommends). Thus, the gradient norm depend on $\alpha, \bar{n}, \AV, \frac{1}{\sqrt{n}}$. Since all variables are positive, it suffices to find an upper-bound in terms of the known values of $n$, $\|\bR^\star\|_1$ and $\|\bR^\star\|_2$. As $1 \le \bar{n} \le n$ we have that $\bar{n} \le n$ and $\frac{1}{\sqrt{n}} \le 1$. By assumption, we have that $\alpha < \sigma_{k_{maj}}(\bR_{\maj}^\star)=\sigma$. Additionally, that $\sigma_{k_{maj}}(\bR_{\maj}^\star) \le \sigma_{1}(\bR^\star) =\|\bR^\star\|_2$. Hence $\alpha < \sigma < \|\bR^\star\|_2$. Lastly, 
\[\AV = \max_{i \in [\bar{n}]}\sum_{u \in \calU_\COLL}r_{u,i} \le \max_{i \in [n]}\sum_{u \in [m]}r_{u,i}=\|\bR^\star\|_1\]

Hence we can define $\mathcal{D}$ to be the space where each of these bounds hold:
\[\mathcal{D} = \{\mathbf{z}: 0 < \alpha < \sigma \le \|\bR^\star\|_2, 1 \le \bar{n} \le n, \AV \le \|\bR^\star\|_1 \}\]
Since it is defined by a set of linear inequalities, $\mathcal{D}$ is convex. 

Further we can bound the norms of the gradients for all $\mathbf{z} \in \mathcal{D}$:
\[\|\nabla f_1(\bz)\|_2^2 = 4\sigma^2 + 4 \alpha^2 + \frac{\hat{\eta} \AV^2}{4\sqrt{\bar{n}}}+\hat{\eta}^2 \bar{n} \le 8\|\bR^\star\|_2^2+\frac{\hat{\eta} \|\bR^\star\|_1^2}{4}+\hat{\eta}^2n\]
\[\|\nabla f_1(\bz)\|_2^2= 4 \alpha^2 + \frac{\hat{\eta} \AV^2}{4\sqrt{\bar{n}}}+1+\hat{\eta}^2 \bar{n}+\hat{\eta}^4 \le  4 \|\bR^\star\|_2^2+ \frac{\hat{\eta} \|\bR^\star\|_1^2}{4}+1+\hat{\eta}^2 n+\hat{\eta}^4\]
Consequently, for $L_1 =\sqrt{8\|\bR^\star\|_2^2+\frac{\hat{\eta} \|\bR^\star\|_1^2}{4}+\hat{\eta}^2n}$ and $L_2 = \sqrt{4 \|\bR^\star\|_2^2+ \frac{\hat{\eta} \|\bR^\star\|_1^2}{4}+1+\hat{\eta}^2 n+\hat{\eta}^4}$ we have that, by the Mean Value Theorem, for any $\bz, \bz' \in \calD$ and $i \in \{1,2\}$:
\[|f_i(\bz) - f_i(\bz^{'})| \le \sup_{\bz \in \calD}\|\nabla f_i(\bz)\|_2\|\bz - \bz'\|_2\le L_i \|\bz - \bz'\|_2\]
That is, $f_i$ is $L_i$-Lipschitz on $\calD$.  

Applying Lemma \ref{lemma:lipschitz}, we can conclude that $f = \min \{f_1, f_2\}$ is Lipschitz with $L = \max\{L_1, L_2\}$. 
\end{proof}

\section{Supplementary material for Section \ref{sec:learner}}
\subsection{Missing proofs}
\subsubsection{Proof of Corollary~\ref{cor:ben_learner}}
\BenLearner*
\begin{proof}[Proof of \Cref{cor:ben_learner}]
    The corollary follows from Theorems \ref{g_maj_b_min2} and \ref{thm:SW_EA_Model1} since from the definition of accuracy-based utilty, $U_{\BEN}^{\TRUE}:= \SW(\bR^\star,\alpha)$ and $U_{\BEN}^{\COLL}:= \SW(\tbR,\alpha)$.
\end{proof}

\subsubsection{Proof of Proposition~\ref{prop:en_learner}}
\EnLearner*
\begin{proof}[Proof of \Cref{prop:en_learner}]
    When agents report truthfully: $r_{ui}^\star = \tilde{r}_{ui}$. Under the collective strategy, the reported preference matrix has all the same elements except $|\calU_{\COLL}|$ ratings (which were $0$s) have become $\eta > 0$. The strict inequality follows from $|\calU_{\COLL}| > 0$ (based on assumptions of Theorem \ref{thm:SW_EA_Model1}).
\end{proof}
\section{Supplementary material for Section \ref{sec:experiment}}\label{app:experiment}
\subsection{Methodological details}\label{app:experiment_details}

\xhdr{Summary.}
We use the Goodreads interactions dataset~\cite{goodreads1, goodreads2}, which consists of users' interactions with books on the platform. Interactions users may have with books are adding to shelf (equivalently, shelving) to indicate interest, marking as read, rating (out of 5), and reviewing. We use the Hugging Face Supervised Fine-Tuning (SFT) Trainer~\cite{vonwerra2020trl} to train four different instances of the Qwen 2.5-Instruct~\cite{qwen2025qwen25technicalreport} 1.5 billion parameter model. Each model, given a user's profile of past interactions with young-adult (YA) books and four new YA book options, recommends two (of the four) new YA books that the user should add to their shelf. That is, the model recommends two new books the user would be interested in, based on the user's past interactions, the book/user information the model has learned from fine-tuning, and the knowledge the model already has from the base Qwen 2.5 model. To study the effect of collective action and whether it achieves the goal of helping less popular authors, we train each of the four models on different versions of $85,238$ users' data. As a baseline, we train a model on users' true history of pre-2015 interactions and use two of each user's actually shelved 2015 books as their ``ground-truth'' recommendation, randomly generating two other 2015 book options such that the model must learn to recommend the two truly shelved books for each user. We train the other three models using the same structure, but with perturbed data that represents collective action to varying levels. We select 10 books from 10 different unpopular authors as the ``targets'' of collection action. In the language of our theory, these targeted books/authors are the picky items, though they do not necessarily satisfy the picky item assumption formally. For all perturbed datasets, we select a random 80\% of users who primarily enjoy very popular books to be in the collective. Each collective user does the following: (1) switches their less popular shelved 2015 book to be one of the target books, representing their boosting of the book directly; (2) adds an interaction with a pre-2015 book by $j$ target authors to represent boosting the authors. To create the three perturbed datasets on which we train the models representing collective action, we vary $j$ across $1, 3,$ and $7$. 

\xhdr{Dataset info.}
We downloaded the Young Adult interactions and books (total of 93,398 books and 34,919,254 interactions) as well as the author meta-data information from the Goodreads Datasets~\cite{goodreads1, goodreads2}. We filter this down to only US and english language books. In order to narrow down to users to whom a model may reasonably recommend books published in 2015, we filter out any users who have fewer than 3 shelved books in 2015 or do not rate any 2015 book at least a 3 of 5.

\xhdr{Baseline model training details.}
From the filtered dataset, we construct the data on which we run supervised fine tuning. We use a prompt-completion structure on the inputs used to train the model. The following is an example of a prompt completion pair:

PROMPT

"Here are a user's 25 most recent pre-2015 interactions:

Added to shelf, read, and rated 5 of 5 Fangirl by Rainbow Rowell. Added to shelf, read, and rated 3 of 5 The Elite (The Selection, \#2) by Kiera Cass. Added to shelf, read, and rated 3 of 5 The One (The Selection, \#3) by Kiera Cass. Added to shelf The 5th Wave (The 5th Wave, \#1) by Rick Yancey. Added to shelf The Infinite Sea (The 5th Wave, \#2) by Rick Yancey.

Using *only* the next options, which 2 2015 books did they add to their shelf (i.e. plan to read)?

All the Bright Places by Jennifer Niven,  All American Boys by Jason Reynolds and Brendan Kiely,  All Fall Down (Embassy Row, \#1) by Ally Carter,  or Ray of Sunlight by Brynn Stein"

COMPLETION

"All the Bright Places by Jennifer Niven,  All American Boys by Jason Reynolds and Brendan Kiely"

To keep prompts of a reasonable token length, we formed the baseline prompt-completion dataset using only the most recent 25 books each user interacted with prior to 2015. The options of potential 2015 shelved books for each user were formed from 2 actually shelved books of that user and then two additional random 2015 book in the dataset that were not shelved. For the prompt and completion, we shuffle the order of the options randomly. The completion are the two actually shelved books. We train on 80\% of the Prompt/Completions using SFTTrainer on Qwen2.5-1.5B-Instruct with the following parameters:

max sequence length: 1024

epochs: 1

batch: 2

gradient accumulation: 8

learning rate: 2e-5

\xhdr{Collective action model training.}
Cutting out any books that appeared as true shelf answers in the baseline model's test set fewer than 5 times, the 10 target books are the most popular of the bottom 20\% popularity books, excluding any books with multiple authors or with authors who do not have an earlier book in the dataset. We define popularity here according to total past interactions with the associated author. It was important to narrow target books to only those that had a single author with a pre-2015 book in the dataset in order to ensure collective action could involve adding an interaction prior to 2015.

We then randomly select the collective participating users as a random 80\% of baseline model train users who have a true shelf answer that is one of the top 30 highest shelving books AND have no shelf answer that corresponds to a target book. For these 80\% of users, we change the (train) completions and options such that the less popular shelf answer becomes a randomly selected target book. Additionally, into the collective users' train prompts, we add ``Added to shelf, read, and rated 4 of 5 <pre-2015 title> by <target author>'' for $j$ of the target books/authors. Random placement of the string was ensured to be between existing interactions, so as not to disturb interpretability. On the perturbed train set, we retrained 3 more models using the same parameters and test/train split.

\xhdr{Testing.}
For all models, we run inference on the full (true) test set and this is presented in Table \ref{tab:accuracies}. We run inference with the same tokenization parameters as from training with the caveat that to the prompt we add, ``"Answer (copy exactly TWO options): '', to encourage model brevity. We do this consistently with all inference.

For the testing of target book/author recommendations, we filter for all test users that have (truly) shelved a target book and replace the two random book options with random popular 2015 books to various thresholds. We refer the reader to \Cref{app:experiment_additional_results} for these results.

\subsection{Additional results}\label{app:experiment_additional_results}

%
%
%
\begin{table}[H]
\centering
\small
\begin{tabular}{lrrrr}
\toprule
Model & Incorrect (\%) & Partial Correct (\%) & \textbf{Correct (\%)}\\
\midrule
Baseline & 1.24 & 31.85 & \cellcolor{gray!20}\textbf{66.91}\\
$j=1$    & 1.41 & 33.82 & \cellcolor{gray!20}\textbf{64.77}\\
$j=3$    & 1.32 & 33.55 & \cellcolor{gray!20}\textbf{65.13}\\
$j=7$    & 1.44 & 34.58 & \cellcolor{gray!20}\textbf{63.98}\\
\bottomrule
\end{tabular}
\caption{Each model outputs a recommendation of \emph{two} books. These are correctness percentages across models on the test set of $21,310$ users. Partial Correct one correct shelving recommendation, Correct is two.}
\label{tab:accuracies}
\end{table}

\begin{figure}[H]
\centering

\begin{subfigure}{0.48\textwidth}
  \centering
  \includegraphics[width=\linewidth]{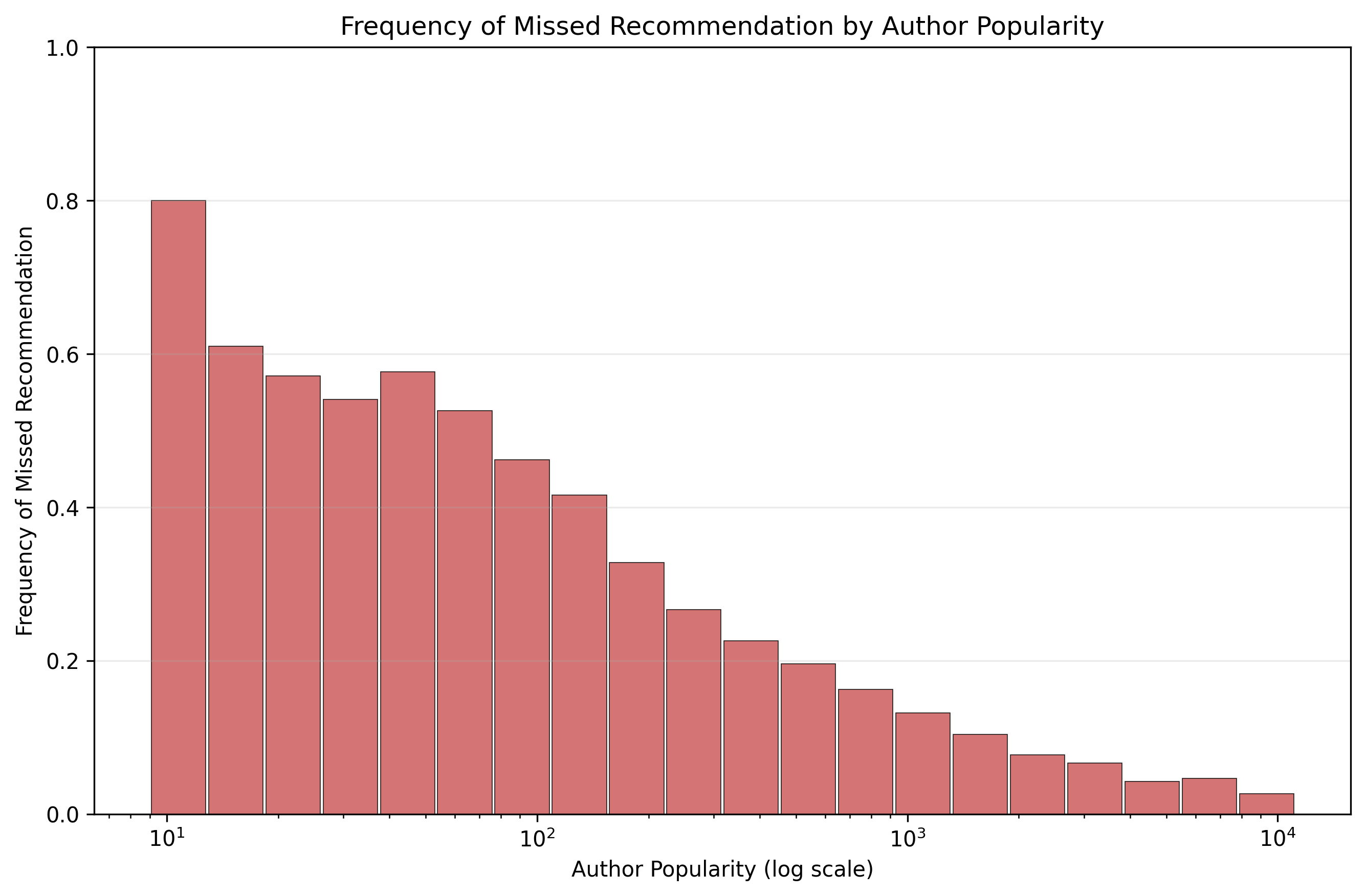}
  \caption{Avg frequency with which books are incorrectly not recommended by popularity of their author}
\end{subfigure}
\hfill
\begin{subfigure}{0.48\textwidth}
  \centering
  \includegraphics[width=\linewidth]{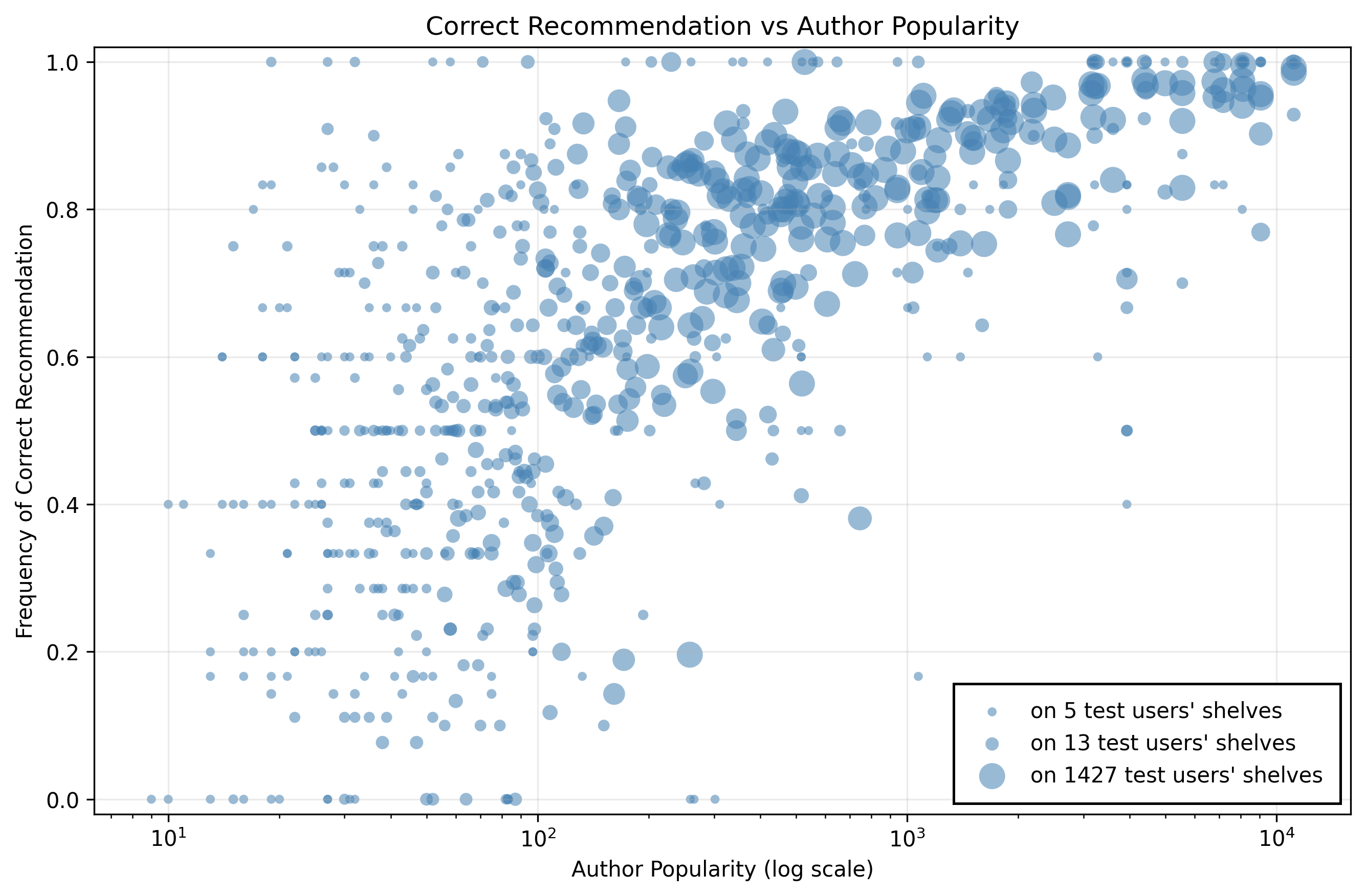}
  \caption{Avg frequency with which books are correctly recommended by popularity of their author}
  \label{fig:sub2}
\end{subfigure}

\caption{[In]correct recommendation of books by their author popularity, measured by how often the author is added to shelf. Books shelved by fewer than $5$ users are dropped.}
\label{fig:baseline_model_acc2}
\end{figure}

\begin{figure}[H]
\centering

\begin{subfigure}{0.48\textwidth}
  \centering
  \includegraphics[width=\linewidth]{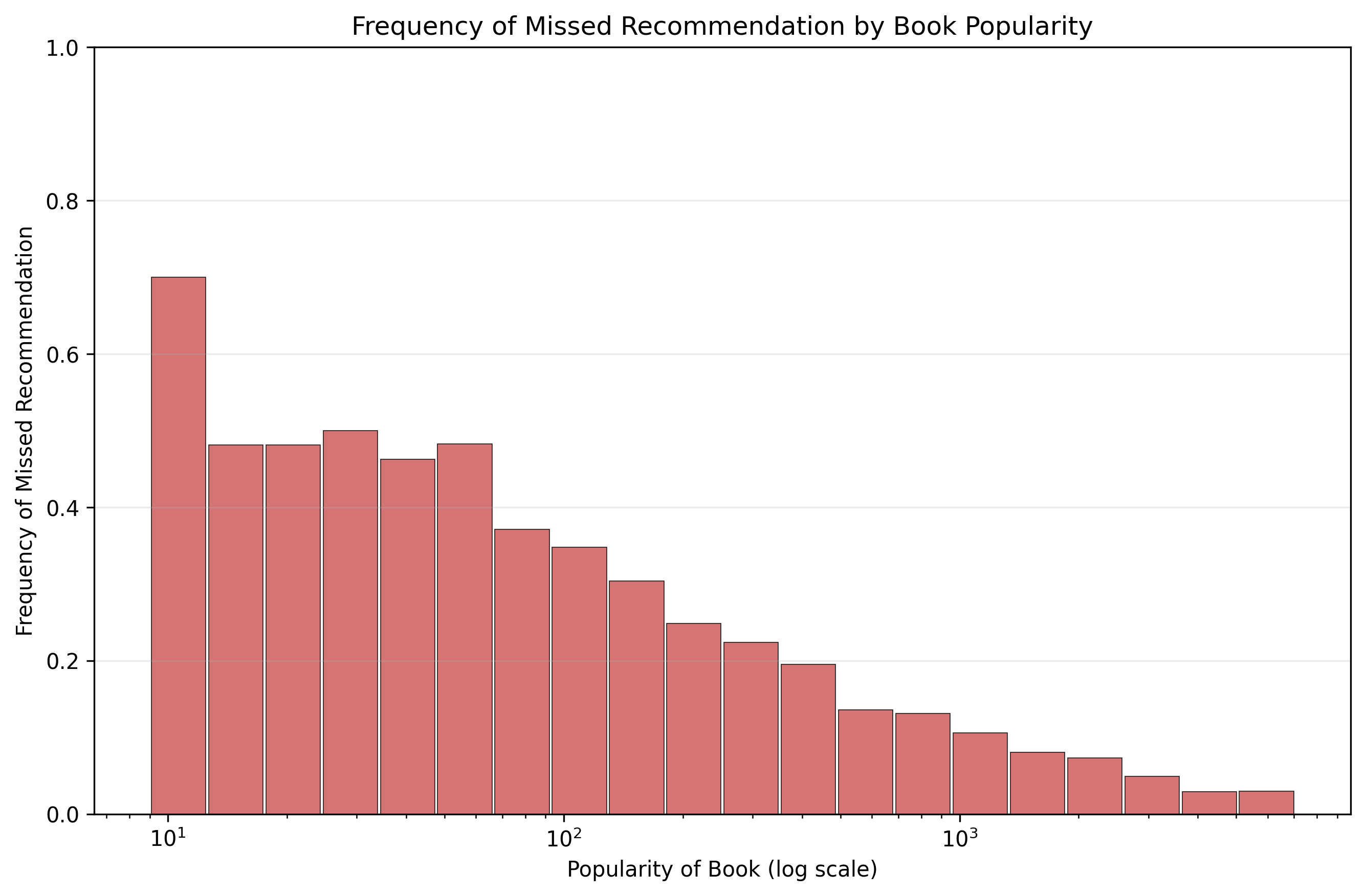}
  \caption{Avg frequency with which books are incorrectly not recommended by popularity}
  \label{fig:sub1}
\end{subfigure}
\hfill
\begin{subfigure}{0.48\textwidth}
  \centering
  \includegraphics[width=\linewidth]{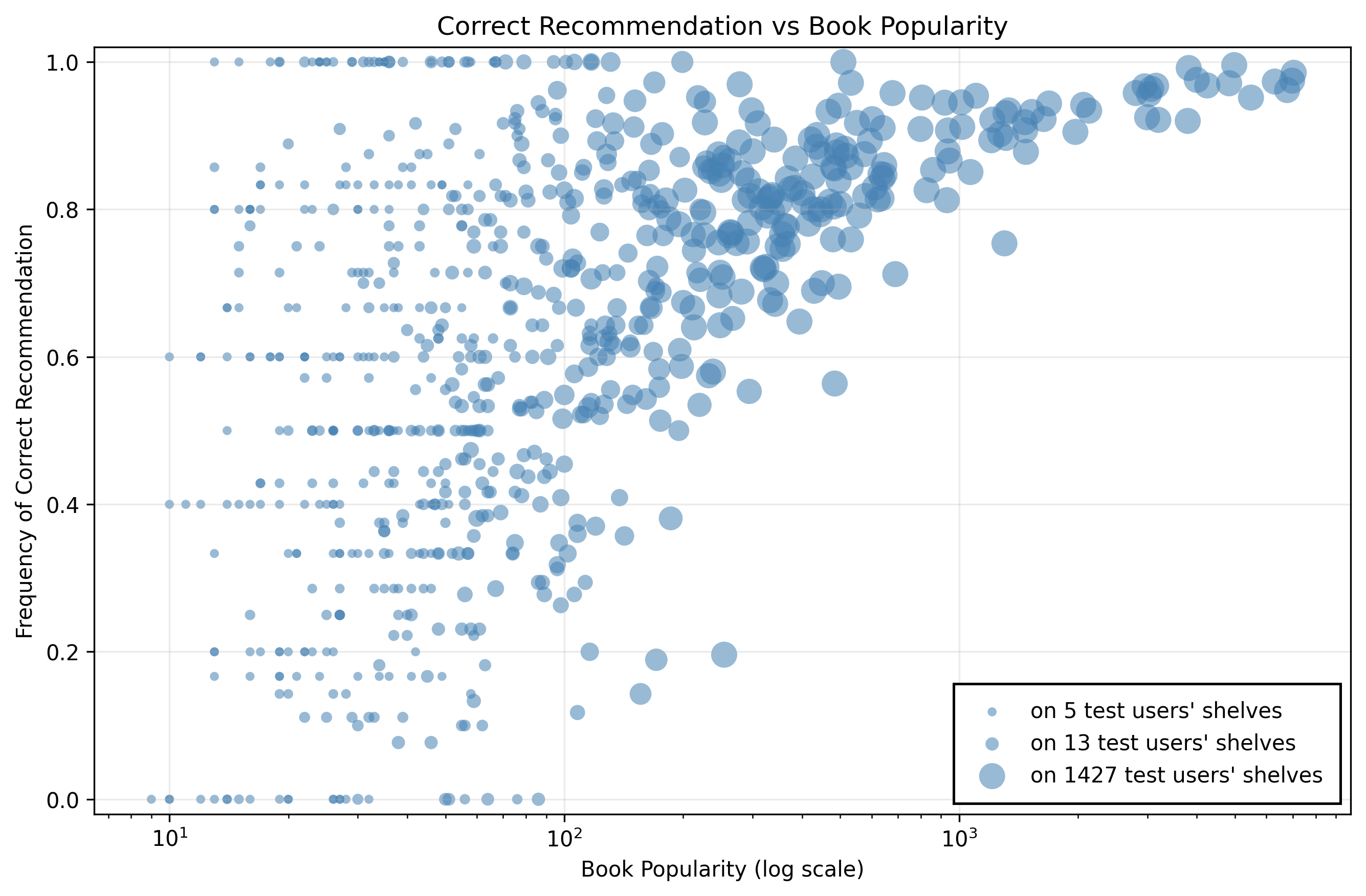}
  \caption{Avg frequency with which books are correctly recommended by popularity}
\end{subfigure}

\caption{[In]correct recommendation of books by their popularity, measured by how often they are added to shelf. Books shelved by fewer than $5$ users are dropped.}
\label{fig:baseline_model_acc3}
\end{figure}

\begin{table}[H]
\centering
\footnotesize
\begin{tabular}{lrrrrr}
\toprule
Target Author & Baseline & $j=1$ & $j=3$ & $j=7$ & Test user \# \\
\midrule
K.A.\ Holt       
& 0.3077 & 0.9231 & 0.9615 & 0.9615 & 26 \\

Allan Stratton   
& 0.0476 & 1.0000 & 1.0000 & 0.9524 & 21 \\

Julie Mayhew     
& 0.2778 & 1.0000 & 1.0000 & 1.0000 & 18 \\

Matthew Crow     
& 0.1333 & 1.0000 & 1.0000 & 0.9333 & 15 \\

Edwidge Danticat 
& 0.1333 & 1.0000 & 0.9333 & 0.9333 & 15 \\

S.\ Walden       
& 0.1538 & 1.0000 & 0.9231 & 0.8462 & 13 \\

Addison Moore    
& 0.0909 & 0.9091 & 0.9091 & 0.9091 & 11 \\

Alan Gratz       
& 0.2000 & 1.0000 & 1.0000 & 1.0000 & 10 \\

Anna Martin      
& 0.5000 & 0.8000 & 0.8000 & 0.8000 & 10 \\

Gabriella Lepore 
& 0.3000 & 0.9000 & 0.9000 & 1.0000 & 10 \\
\bottomrule
\end{tabular}
\caption{Models' frequency of target author recommendation to interested test users when selecting from options consisting of two correct books and two random books from top 30 most-interacted-with authors.}
\label{tab:improve_rec_variant1}
\end{table}

\begin{table}[H]
\centering
\footnotesize
\begin{tabular}{lrrrrr}
\toprule
Target Author & Baseline & $j=1$ & $j=3$ & $j=7$ & Test user \# \\
\midrule
K.A.\ Holt       
& 0.5385 & 0.9231 & 0.8846 & 0.9231 & 26 \\

Allan Stratton   
& 0.0476 & 1.0000 & 1.0000 & 0.9524 & 21 \\

Julie Mayhew     
& 0.2778 & 0.9444 & 0.9444 & 0.9444 & 18 \\

Matthew Crow     
& 0.0000 & 1.0000 & 0.9333 & 0.7333 & 15 \\

Edwidge Danticat 
& 0.2000 & 1.0000 & 1.0000 & 0.9333 & 15 \\

S.\ Walden       
& 0.3077 & 0.9231 & 0.9231 & 1.0000 & 13 \\

Addison Moore    
& 0.7273 & 1.0000 & 1.0000 & 1.0000 & 11 \\

Alan Gratz       
& 0.3000 & 1.0000 & 1.0000 & 1.0000 & 10 \\

Anna Martin      
& 0.5000 & 1.0000 & 1.0000 & 1.0000 & 10 \\

Gabriella Lepore 
& 0.5000 & 1.0000 & 1.0000 & 1.0000 & 10 \\
\bottomrule
\end{tabular}
\caption{Models' frequency of target author recommendation to interested test users when selecting from options consisting of two correct books and two random books from top 500 most-interacted-with authors.}
\label{tab:improve_rec_variant2}
\end{table}

\begin{table}[H]
\centering
\footnotesize
\begin{tabular}{lrrrrr}
\toprule
Target Author & Baseline & $j=1$ & $j=3$ & $j=7$ & Test user \# \\
\midrule
K.A.\ Holt       
& 0.5000 & 0.9615 & 0.9615 & 0.9615 & 26 \\

Allan Stratton   
& 0.1429 & 0.9524 & 0.9524 & 0.9524 & 21 \\

Julie Mayhew     
& 0.4444 & 0.9444 & 1.0000 & 0.9444 & 18 \\

Matthew Crow     
& 0.2000 & 0.9333 & 0.9333 & 0.8667 & 15 \\

Edwidge Danticat 
& 0.2000 & 1.0000 & 1.0000 & 1.0000 & 15 \\

S.\ Walden       
& 0.2308 & 1.0000 & 1.0000 & 1.0000 & 13 \\

Addison Moore    
& 0.5455 & 0.9091 & 0.9091 & 0.9091 & 11 \\

Alan Gratz       
& 0.2000 & 1.0000 & 1.0000 & 1.0000 & 10 \\

Anna Martin      
& 0.4000 & 1.0000 & 1.0000 & 1.0000 & 10 \\

Gabriella Lepore 
& 0.8000 & 1.0000 & 1.0000 & 1.0000 & 10 \\
\bottomrule
\end{tabular}
\caption{Models' frequency of target author recommendation to interested test users when selecting from options consisting of two correct books and two random books from top 1000 most-interacted-with authors.}
\label{tab:improve_rec_variant3}
\end{table}

\section{Additional results}\label{app:model2}
In this section, we derive analogous social welfare results for a more complicated class of preference matrices and collective strategies than those which are in the main body of this manuscript. This preference matrix class is not limited to majority minority groups with exclusive preferences and collective strategies are not limited to uprating by a single value, $\eta$. Because of these complexities, the proof techniques necessary are significantly different than in the main body, but the majority of results have analogies. Additionally, we have not been able to provide a simple algorithm for the computation of an effective collective strategy as we do in the simpler case.

\subsection{Technical connections to the PCA item-fairness of \citet{pca_fairness_liu}}
While conceptual connections between our main body results and \citet{pca_fairness_liu} are already discussed in Section \ref{sec:related_work} and Appendix \ref{app:connection_fair_pca}, we note in this section that the setting presented in this appendix is related both in spirit \emph{and} in technical tools. Reading this section is not necessary to understand the results of this Appendix, but it may be useful for readers already familiar with the theoretical results of \citet{pca_fairness_liu}. 

In the definition of majority and minority users, we define \emph{popular} and \emph{unpopular} items. \citet{pca_fairness_liu} similarly define such items, though because they consider a sequence of preference matrices growing to size $\infty \times \infty$, they define [un]popular values in terms of growth in singular values over the sequence. Our definitions of popular and unpopular items are essentially finite analogies to theirs, defined by a relevant \emph{gap} in singular values. As a result, the proof techniques found in this section, particularly in the use of a variant the Davis-Kahan theorem \cite{davis_kahan} from \citet{yu_useful_2015} will be familiar to a reader with knowledge of \citet{pca_fairness_liu}. Whereas \citet{pca_fairness_liu} use Davis-Kahan to show that the rank reduction of PCA approaches a version of the ratings matrix where all unpopular items have been zeroed out, because we focus on \emph{recommendations}, we are able to make additional minor assumptions such that the rank reduction is guaranteed to have the popular items as the highest estimated rated item for every user, thus ensuring that all users (even those that prefer unpopular items) end up with popular items. Therefore, for our purposes, while unpopular items are not necessarily estimated as 0 rating, they are functionally 0 in that no user will actually receive them.

\subsection{A class of popularity gap matrices}\label{pref_class2}
In the main body, we have a fairly simplistic idea of majority and minority agent groupings. The groups' preferences are disjoint, which makes the idea of which items ``belong'' to each group very simple. Further, we can separate groups into their own smaller preference matrices and define what makes a group majority in terms of matrix singular values. If we break disjointedness, as we do in for these Appendix results, we must redefine what makes a majority/minority group and what makes corresponding ``popular'' and ``unpopular'' items such. In this subsection, we introduce how we can do this. Additionally, note that we present all of our definitions using a generic ratings/preference matrix $\bR$. For \emph{theorems, propositions, etc}, we specify genuine personal interest matrices, $\bR^\star$ and revealed preference matrices $\widetilde{\bR}$ where it is needed.

Consider a tuple: $(\mathbf{R}, \bar{n})$ where: $\mathbf{R} \in [0,1]^{m\times n}$ is a [normalized] preference matrix and $\bar{n}$ is some integer value $0< \bar{n}< n$. Call the first $\bar{n}$ items (columns) of $\mathbf{R}$ the \emph{popular} items and the remaining, the \emph{unpopular} items. 
Define $\kappa_{(\bfR, \barn)} \in \R_{\geq 0}$ to be $\max_{i' \in \{(\bar{n}+1), \dots, n\}}\|\bR_{i'}\|_1$, the greatest L-1 norm for any unpopular item. The $\bR$ matrix but with the preferences for unpopular items zeroed out is important for the remainder of our analysis, so we define this as follows.
\begin{definition}[Popular Preferences Matrix, $\bR'(\barn)$]
Let $\mathbf{R}'(\bar{n}) \in [0,1]^{m\times n}$ be a matrix s.t.:
\[
r'_{u,i} =
\begin{cases}
r_{u,i}, & \text{if $i \leq \bar{n}$}\\
0, & \text{otherwise}
\end{cases}
\]
Thus $\mathbf{R}'({\bar{n}})$ is a block matrix where $\mathbf{A}({\bar{n}}) \in [0,1]^{m\times \bar{n}}$ is the popular item block of $\bf{R}$.:
\begin{equation}
    \mathbf{R}'({\bar{n}}) = 
\begin{pmatrix} 
\mathbf{A}({\bar{n}}) & \mathbf{0}^{m \times (n-\bar{n})} 
\end{pmatrix}
\end{equation}
\end{definition}

Because the learner is interested in recovering a best item for each user, we define a set of top item(s) for a user $u$ with respect to the matrix, $\bR$:
\begin{definition}[User $u$'s Top Item(s)]
    Define $\topitems(\bR, u)$ to be a set of top items for a user $u$ according to preference matrix $\bR$: 
    \[\topitems(\bR, u) := \mathrm{arg}\max_{i\in[n]}r_{u,i}\]
    Clearly, $|\topitems(\bR, u)| \geq 1$.
\end{definition}
Note that by definition, $\Topk(u) \subseteq \topitems(\widehat{\bR}, u)$, but it is not necessarily the case that they are \emph{exactly} the same because $\Topk(u)$ is the singular item that user $u$ is actually recommended.

\subsubsection{User groupings}
We define user groups based on whether a user's top item(s) is(are) popular or unpopular. Intuitively, this is like saying a user is majority if the thing she likes most is mainstream item and she is minority if the thing she likes most is something niche. Notably, we are defining users' groups using their \emph{most liked} item. This means that the full profile of preferences are not disjoint. 
\begin{definition}[Majority User]
    A majority user for a particular $\bf{R}$ preference matrix and $\barn$ is one who has top rated item $i \in \topitems(\bR, u)$, such that $i\leq \bar{n}$, meaning $i$ is one of the popular items. Formally, we define the set of majority users for a particular $\bf{R}$ preference matrix:
    \[\calUmaj := \{u: \exists i \in \topitems(\bR, u)\quad s.t. \quad i \in [\bar{n}]\}\]
\end{definition}

\begin{definition}[Minority User]
    A minority user for a particular $\bf{R}$ preference matrix and $\barn$ is one who has a top-rated item $i \in \topitems(\bR, u)$, such that $i> \bar{n}$, meaning $i$ is one of the unpopular items. Formally, we define the set of minority users for a particular $\bf{R}$ preference matrix:
    \[\calUmin := \{u: \exists i \in \topitems(\bR, u)\quad s.t. \quad i \in \{(\bar{n} + 1), \dots, n\} \}\]
\end{definition}
For some results, it is useful to make assumptions that a tuple, $(\bfR, \barn)$ is such that there exist nonempty user majority/minority groups and they are exclusive. 
\begin{assumption}[Minority/Majority User assumptions]\label{maj_min_users_ass}
Preference matrix $\mathbf{R}$ and $\bar{n}$ is such that the following is true of majority and minority user groups:
    \begin{enumerate}
        \item Majority and Minority user sets are exclusive: 
        \[\calUmin \cap \calUmaj= \emptyset\]
        \item There is at least one minority user:
        \[\forall m: |\calUmin| > 0\]
    \end{enumerate}
\end{assumption}
\begin{remark}
    Note that the majority/minority exclusivity of assumption \ref{maj_min_users_ass} is a weaker assumption than the exclusivity assumption in the main body of the paper as assumption \ref{maj_min_users_ass} exclusivity does not imply that majority and minority users' preferences are entirely exclusive, only that their \textbf{top items} are exclusive. Formally:
    \begin{equation*}
        \begin{aligned}
            \forall u \in \calUmaj, \forall i \in \{\barn + 1, \dots, n\}, &\quad i \notin \topitems(\bR, u)\\
            \forall u' \in \calUmin, \forall i' \in [\barn], &\quad i' \notin \topitems(\bR, u')
        \end{aligned}
    \end{equation*}
    \end{remark}
    
We now construct a class of preference matrix, $\bfR$ and popular item index $\barn$ tuples. In order to do this, we define the assumptions that tuples belonging to this class must satisfy. These assumptions are be about the difference between particular ratings of the preference matrix, so before proceeding we define $\Delta(\mathbf{R}, \barn)$. This has a similar function to the singular value gaps of the class of matrices used in the main body in that we use $\Delta$ to impose a gap in popularity between items while in the main body it denoted a gap between singular values of the majority/minority.

\begin{definition}[Sufficient Ratings Gap, $\Delta(\mathbf{R}, \barn)$]
The sufficient ratings gap is a function of the ratings of $\bfR$ and which items are popular, $\barn$
\begin{equation}
    \Delta(\mathbf{R}, \barn):= \frac{2^{\frac{5}{2}}\kappa_{(\bfR, \bar{n})} n^{\frac{3}{2}}}{\left[\sigma_{\bar{n}}(\mathbf{R}'(\bar{n}))\right]^2}
\end{equation}
\end{definition}

\begin{assumption}[Majority users' top item(s) are sufficiently highly rated] \label{maj_users_high}
    Majority users don't care about all items equally: $[n]\setminus \topitems(\bR, u) \neq \emptyset$ and there is a sufficient gap between a majority user's top rating (which may appear on multiple items) and her other ratings:
    \begin{equation}
        \max_{i \in [n]\setminus \topitems(\bR, u)}r_{u,i} < \max_{i \in [n]}r_{u,i} - \Delta(\bfR, \barn) \quad \forall u \in \calUmaj
    \end{equation}
\end{assumption}

\begin{assumption}[Minority users have sufficient preference for a popular item] \label{min_users_high}
    Each minority user likes at least one popular item by a sufficient amount:
    \begin{equation}
    \exists i \in \bar{n}\quad  s.t. \quad
        r_{u,i} > \Delta(\bfR, \barn) \quad \forall u \in \calUmin
    \end{equation}
\end{assumption}
Note that because assumption \ref{min_users_high} states that each minority user likes at least one popular item by some small amount, the majority-minority matrices discussed in the main body \emph{cannot} satisfy this assumption as they are block matrices that impose complete exclusivity in preference. Thus, while the class we construct here does not technically include matrices analyzed in the main body, the appendix class can be viewed as ``more general" because it handles the settings in which preferences matrices do not have the $\bR_{\minor}, \bR_{\maj}$ block structure that creates exclusivity between all items a majority user likes and all items a minority user likes. The results of this appendix are basically complementary to those in the main body. One way to think about this is that Assumption \ref{min_users_high} is reminiscent of a non-zero support assumption on minority users' preference over popular items while the main body imposes zero-support over the same space.
We can now define the key class of preference matrix and popular item index tuples that we use for the remainder of the results in this section.
\begin{definition}[Popularity Gap Class, $\calM$]
    The following is an important class of tuples where the popular items (whose indices lies in $[\barn]$) are sufficiently more highly rated by a variety of users than the unpopular items:
    \begin{equation}
        \mathcal{M} := \{(\mathbf{R}, \bar{n}):\text{Assumption \ref{maj_users_high}, \ref{min_users_high} hold.}\}
    \end{equation}
\end{definition}

\begin{remark}[The Meaning of Popularity]\label{rem:pop_gap_class_why}
    We note that for any valid $\mathbf{R} \in [0,1]^{m\times n}$ and $\barn$ where assumptions \ref{maj_users_high} and \ref{min_users_high} are true, it must be the case that the following is true:
    \begin{equation}\label{eqn:rem_model_2}
        2^{\frac{5}{4}}n^{\frac{3}{4}}\sqrt{\kappa_{(\bfR, \bar{n})}} < \sigma_{\bar{n}}(\mathbf{R}'(\bar{n}))
    \end{equation}

    Intuitively, equation \ref{eqn:rem_model_2} means that popular items (those in $[\bar{n}]$) are sufficiently well-liked by enough users such that their associated singular values are big relative to the magnitude of minority items' ratings (whose l1-norms are upper bounded by $\kappa_{(\bfR, \bar{n})}$). 
\end{remark}
We call $\calM$ the ``Popularity Gap Class" following the intuition detailed in remark \ref{rem:pop_gap_class_why}. That is, in order for the assumptions to be potentially satisfied, it must be the case that the singular value of the associated Popular Preferences Matrix, $\bR'(\barn)$, dominates over the a function of the unpopular ratings. Much like the singular value gap assumption of the main body, ensuring that equation \ref{eqn:rem_model_2} holds ensures that the items labeled as popular by $\barn$ are actually mathematically popular.
\subsection{Learner's selection of optimal truncation rank}
In the section above, we present a class of $(\mathbf{R},\bar{n})$.
However, recall that the $\alpha$-loss tolerant learner gets an approximated version of the received preference matrix, $\tbR$, such this approximated version, $\widehat{\bR}$ has a particular ``optimal'' rank, $k^\star$. 
Given that the tuple $(\tbR,\bar{n}) \in \mathcal{M}$, in this section we show that $k^\star = \bar{n}$ for learners whose $\alpha$ total variance loss budget is within a particular range. This result is this setting's version of \Cref{rhat_equals_nbar1}.

\subsubsection{Preliminary: useful singular value bounds}
Recall that $\alpha$-loss tolerant learners are defined in terms of how large the next singular value after truncation would be. Thus, it is useful to have bounds on the singular values of $\bfR$ and $\bfR'(\barn)$.
\begin{corollary}[Corollary of Lemma \ref{thm:remove_cols}] \label{relations_singular_vals}
    Let matrix $\tilde{\mathbf{A}} \in \mathbb{C}^{m\times (n-j)}$ where $j < n$. Define $\sigma_r(\hat{\mathbf{A}}) = 0$ for singular values lost to column deletion. Then the following relation of singular values holds:
    \begin{equation*}
        \sigma_i(\mathbf{A}) \geq \sigma_i(\tilde{\mathbf{A}}) \geq \sigma_{i + j}(\mathbf{A})
    \end{equation*}
\end{corollary}
\begin{proof}
    This is easily seen by induction on $j$ using the Horn and Johnson lemma as a $j =1$ base case.
\end{proof}

\begin{proposition}[$\barn$ and $\barn + 1$ singular value bounds] \label{next_singular_val_low}
    If a tuple $(\mathbf{R}, \bar{n}) \in \mathcal{M}$ then the singular values of $\mathbf{R}$ satisfy the following relations:
    \begin{equation*}
        \sigma_{\bar{n}}(\mathbf{R}) \geq 2^{\frac{5}{4}}n^{\frac{3}{4}}\sqrt{\kappa_{(\bfR, \bar{n})}}
    \end{equation*}
    \begin{equation*}
        \sigma_{\bar{n}+1}(\mathbf{R}) \leq \sqrt{(n-\bar{n})\kappa_{(\bfR,\barn)}}
    \end{equation*}
\end{proposition}
\begin{proof}
The first inequality comes from the fact that assumptions \ref{maj_users_high} and \ref{min_users_high} hold. Clearly $\frac{2^{\frac{5}{2}}\kappa_{(\bfR, \bar{n})} n^{\frac{3}{2}}}{\left[\sigma_{\bar{n}}(\mathbf{R}'(\bar{n}))\right]^2} \leq 1$ if assumption \ref{maj_users_high} is true because otherwise the minimum difference between a top items and next ratings is greater than what $\mathbf{R} \in [0,1]^{m\times n}$ allows. This yields $\sigma_{\bar{n}}(\mathbf{R}'({\bar{n}})) \geq 2^{\frac{5}{4}}n^{\frac{3}{4}}\sqrt{\kappa_{(\bfR, \bar{n})}}$ Because $\mathbf{R}'({\bar{n}})$ is the same as $\mathbf{R}$ with the unpopular columns removed (and replaced with zeros, which does not affect singular values) we can invoke Corollary \ref{relations_singular_vals} to get the desired inequality in terms of $\mathbf{R}$.\\
Now we show the second inequality. Define matrix $\mathbf{B} \in [0,1]^{m \times (n-\bar{n})}$ to be matrix $\mathbf{R}$ but where popular item columns have been removed. We have the following:
\begin{equation*}
    \begin{aligned}
        \sigma_{1 + \bar{n}}(\mathbf{R}) &\leq \sigma_1(\mathbf{B})\quad \text{Corollary \ref{relations_singular_vals}}\\
        &= ||\mathbf{B}||_{2} \quad \text{Def of spectral norm}\\
        &\leq ||\mathbf{B}||_{F} \quad \text{Matrix Norm Equivalences}\\
        &\leq \sqrt{(n-\bar{n})\kappa_{(\bfR,\barn)}}
    \end{aligned}
\end{equation*}

To get the last inequality, recall that $||\mathbf{X}||_{F} = \sqrt{\mathrm{tr}(\mathbf{X}^\top \mathbf{X})}$ and note that 
\begin{equation*}
    (\mathbf{B}^\top\mathbf{B})_{i,j}\leq \max_{i \in [n-\bar{n}]}\mathbf{B_i}^\top\mathbf{1} = \max_{i \in [n-\bar{n}]}||\mathbf{B_i}||_{1} = \kappa_{(\bfR,\barn)}
\end{equation*}
where $\mathbf{B_i}$ is the $i$th column vector of $\mathbf{B}$. $\mathrm{tr}(\mathbf{X}^\top \mathbf{X})$ must thus be upper bounded by $(n-\barn)\kappa_{(\bfR,\barn)}$ because there are $n-\barn$ diagonal elements of $\mathbf{B}^\top \mathbf{B}$ each upper bounded by $\kappa_{(\bfR,\barn)}$.
\end{proof}
\subsubsection{The important class of learners}

Intuitively, we have shown in the above preliminaries that if $(\mathbf{R}, \bar{n}) \in \mathcal{M}$, then the $\bar{n}$-th singular value of $\mathbf{R}$ must be relatively big while the next singular values must be quite small. This should mean that retaining singular values 1 through $\bar{n}$ is ``important" while the remainder of the singular values do not contribute very much. Notably this is another way in which to understand intuitively why an item is ``popular'' or ``unpopular''!
\begin{definition}[$(\barn, \bR)$-Singular Value Gap] For any $\bR \in [0,1]^{m\times n}$ and $\barn \in [n]$ s.t. $(\bR, \barn) \in \calM$, define the space 
\begin{equation*}
    \calG(\barn, \bR):= \{y \in \R: y \in  \left(\sqrt{(n-\bar{n})\kappa_{(\bfR,\barn)}}, 2^{\frac{5}{4}}n^\frac{3}{4}\sqrt{\kappa_{(\bfR,\barn)}}\right)\}
\end{equation*}
    
\end{definition}

Much like in the main body of the paper, a learner whose $\alpha$ parameter falls into this gap selects to truncate exactly to dimension $\barn$. Formally:
\begin{proposition}[$k^\star=\barn$ for the $\alpha$-loss tolerant learner] \label{rhat_equals_nbar}
    For all $\tbR \in [0,1]^{m\times n}$ such that $(\tbR, \bar{n}) \in \mathcal{M}$, 
    If the $\alpha$ loss tolerant learner is such that $\alpha \in \calG(\barn, \tbR)$,
    then it must be the case that $k^\star = \bar{n}$.
\end{proposition}
\begin{proof}
By proposition \ref{next_singular_val_low}, $\sigma_{\bar{n}}>\alpha$ while $\sigma_{\bar{n}+1}<\alpha$. By properties of singular values, $\sigma_j \geq \sigma_{\bar{n}}\forall j \leq \bar{n}$, thus $\bar{n}$ is the minimum $k$ such that $\sigma_{k + 1} < \alpha$
\end{proof}
\begin{remark}[$\calG(\barn, \bR)$ is nonempty]\label{calg_nonmpety_appendix}
    Note that for any $\bar{n} \geq 1$, $2^{\frac{5}{4}}n^\frac{3}{4} > \sqrt{n-\bar{n}}$ therefore the space $\calG(\barn, \tbR)$ is not empty for any reasonable tuple.
\end{remark}
Obviously, this range limits the type of learners we discuss, however it is importantly, there are always some learners who fall into this range! And for large $n$, this range for $\alpha$ is also very big, therefore this is a non negligible space of general $\alpha$-loss tolerant learners.
\subsection{Recommendations and social welfare results}
We are now ready to derive the analogous results to Section \ref{sec:users} for this setting. We compute the recommendations made and resulting social welfare when the received preference matrix is such that $(\tbR, \barn)\in \calM$ and the $\alpha$-loss tolerant learner is parametrized such that $k^\star = \barn$. Recall that the learner gets access to a SVD truncation of the the matrix $\tbR$. Before deriving the main results, we need some useful lemmas about SVD in our setting. This involves similar techniques to those used in \citet{pca_fairness_liu}, though we use them for our finite matrices rather than for an infinite sequence.
\subsubsection{Preliminary: SVD truncation error bounds}
First we remind the reader of an equivalent representation of truncated SVD and derive a useful lemma to upper bound the approximation error.

Recall that for a preference matrix, $\bfR \in [0,1]^{m\times n}$, a $k^\star$-truncated SVD approximation is equivalent to solving the following optimization problem:
\begin{equation} \label{eqn:svd_opt}
    \begin{aligned}
        \text{minimize}_{\mathbf{\Pi}\in \R^{n \times n}} &\quad \| \bR - \bR \mathbf{\Pi}\|^2_F \\
        \text{subject to} &\quad \mathbf{\Pi} = \mathbf{U}\mathbf{U}^\top \\
        &\quad \mathbf{U} \in \R^{n \times k^\star} \\
        &\quad \mathbf{U}^\top \mathbf{U} = \mathbf{I}_{k^\star}
    \end{aligned}
\end{equation}
Where $\mathbf{I}_{k^\star}$ is a $k^\star$-dimensional identity and clearly $\mathbf{\Pi}$ is a projection matrix. 

We can define $\widehat{\bR} = \bR \mathbf{\Pi}^\star$ where $\mathbf{\Pi}^\star$ is the minimizer and this $\widehat{\bR}$ is equivalent to the $k^\star$-truncated SVD. We can derive a bound on how close the optimal projection matrix, $\mathbf{\Pi}^\star$, is to $\mathbf{I}_{n,\barn}$, a ``partial" identity matrix where only the first $\barn$ diagonal elements are $1$s and everything else is a 0. Functionally, because $\widehat{\bR} = \bR \mathbf{\Pi}^\star$, this is an upper bound on how close $\widehat{\bR}$ is to just being the first $\barn$ columns of $\bR$ with the remaining columns zeroed out (recall we called this matrix $\bR'(\barn)$). We note that this bound (and its proof) is a version of Theorem 1 from \cite{pca_fairness_liu}.
\begin{proposition}\label{projection_identity}
    Let $\mathbf{\Pi}^\star_{\bar{n}} \in \R^{n\times n}$ be the optimal projection operator of $\mathbf{R}$ to its $\bar{n}$-truncated SVD. Assume that $\sigma_{\barn}(\bR'(\barn))>0$. We have the following:
    \begin{equation}
        ||\mathbf{\Pi}^\star_{\bar{n}} - \mathbf{I}_{n,\bar{n}}||_{F} \leq \frac{\Delta(\bfR, \barn)}{2\sqrt{n}},
    \end{equation}
    where $\mathbf{I}_{n,\bar{n}}$ is a $n\times n$ matrix where the first $\bar{n}$ diagonal entries are $1$ and all other entries are $0$.
\end{proposition}

In order to show this, we invoke a well-known matrix theory result that we define as a lemma and prove for completeness below. As a preliminary, recall from matrix analysis that between two matrices, we may compare their subspaces using \emph{principal angles}. In particular, between two matrices, $\mathbf{U}, \mathbf{U}' \in \R^{m\times \barn}$ made up of orthonormal columns, the vector of principal angles is $\mathbf{d}:= (\cos^{-1}\sigma_1, \dots \cos^{-1}\sigma_{\barn})$ where $\sigma_i$ is the $i$th singular value of $\mathbf{U}^\top\mathbf{U}'$. We denote $\sin \Theta(\mathbf{U}, \mathbf{U}'):= \mathrm{diag}(\mathbf{d})$, i.e., a matrix with elements of $\mathbf{d}$ on the diagonal and zeros everywhere else.

\begin{lemma}\label{lemma:1.5} Let $\mathbf{U}, \mathbf{U}' \in \mathbb{R}^{m \times \bar{n}}$ be matrices with orthonormal columns. 
\[\|\sin \Theta(\mathbf{U},\mathbf{U}')\|_F = \frac{1}{\sqrt{2}}\|\mathbf{U}\mathbf{U}^\top - \mathbf{U}'\mathbf{U}'^\top\|_F.\]
\end{lemma}
\begin{proof}
Let $\mathbf{\Pi}:=\mathbf{U}\mathbf{U}^\top$ and $\mathbf{\Pi}':= \mathbf{U}'\mathbf{U}'^\top$ notice that these are projection matrices. 
\begin{equation*}
\begin{aligned}
\|\mathbf{U}\mathbf{U}^\top - \mathbf{U}'\mathbf{U}'^\top\|_F^2 &= \| \mathbf{\Pi} - \mathbf{\Pi}'\|_F^2\\
&= \mathrm{Tr}\left((\mathbf{\Pi}-\mathbf{\Pi}')^\top (\mathbf{\Pi}-\mathbf{\Pi}')\right) &\quad \text{def of Frobenius norm}\\
&= \mathrm{Tr}\left(\mathbf{\Pi}^\top\mathbf{\Pi} + \mathbf{\Pi}'^\top\mathbf{\Pi}' - \mathbf{\Pi}^\top\mathbf{\Pi}' - \mathbf{\Pi}'^\top\mathbf{\Pi}\right)\\
&=\mathrm{Tr}(\mathbf{\Pi}^2) + \mathrm{Tr}(\mathbf{\Pi}'^2) - \mathrm{Tr}(\mathbf{\Pi}\mathbf{\Pi}') - \mathrm{Tr}(\mathbf{\Pi}'\mathbf{\Pi}) &\quad \text{projection symmetric, trace linear}\\
&=\mathrm{Tr}(\mathbf{\Pi}^2) + \mathrm{Tr}(\mathbf{\Pi}'^2) - 2\mathrm{Tr}(\mathbf{\Pi}\mathbf{\Pi}') &\quad \text{trace cyclic}\\
&=\mathrm{Tr}(\mathbf{\Pi}) + \mathrm{Tr}(\mathbf{\Pi}') - 2\mathrm{Tr}(\mathbf{\Pi}\mathbf{\Pi}') &\quad \text{projection idempotent}\\
&= 2\barn -2\mathrm{Tr}(\mathbf{\Pi}\mathbf{\Pi}') &\quad \text{trace of projection = rank}\\
&= 2\barn -2\mathrm{Tr}\left((\mathbf{U}^\top\mathbf{U}')^\top(\mathbf{U}^\top\mathbf{U}')\right)\\
&=2\barn - 2\sum_{i \in [\barn]}\cos^2(d_i)&\quad\sigma_i(\mathbf{U}^\top \mathbf{U}') = \cos(d_i)\\
&=2\left(\sum_{i \in [\barn]}1-\cos^2(d_i)\right)\\
&=2\left(\sum_{i \in [\barn]}\sin^2(d_i)\right)&\quad\text{trig identity}\\
&=2\|\sin \Theta(\mathbf{U},\mathbf{U}')\|_F^2
\end{aligned}
\end{equation*}
Taking square root and dividing by $\sqrt{2}$ on both sides gives the desired identity.
\end{proof}

Now we are ready for the proof of \Cref{projection_identity}.
\begin{proof}[Proof of Proposition \ref{projection_identity}]
Let $\mathbf{C} = \mathbf{R}^\top \bR$ and $\mathbf{C}' = \mathbf{R}'({\bar{n}})^\top\mathbf{R}'({\bar{n}})$, thus $\mathbf{C},\mathbf{C}'$ $\in \R^{n \times n}$. Let $\mathbf{U}, \mathbf{U}'\in \R^{n \times \bar{n}}$ be matrices whose columns correspond to the $\bar{n}$ normalized eigenvectors of the $\bar{n}$ largest eigenvalues of $\mathbf{C}, \mathbf{C}'$.\\
We complete this proof by going through the following claims:
\begin{enumerate}
    \item $\mathbf{U}'\mathbf{U}'^{\top} = \mathbf{I}_{n, \barn}$
    \item $\frac{1}{\sqrt{2}}\|\mathbf{U}\mathbf{U}^\top -\mathbf{I}_{n, \bar{n}}\|_F=\|\sin \Theta(\mathbf{U},\mathbf{U}')\|_F$
    \item $\|\mathbf{U}\mathbf{U}^\top -\mathbf{I}_{n, \bar{n}}\|_F \le \frac{2\sqrt{2}n\kappa}{[\sigma_{\barn}(\bR'(\barn))]^2}$
\end{enumerate}
For some notational cleanliness, for this proof, we refer to $\kappa_{(\bR, \barn)}$ as simply $\kappa$.\\
\begin{claim}\label{projection_identity_claim_1}
    \[
\mathbf{U}'\mathbf{U}'^{\top} = \mathbf{I}_{n, \barn}
\]
where $\mathbf{I}_{n,\bar{n}}$ is a $n\times n$ matrix where the first $\bar{n}$ diagonal entries are $1$ and all other entries are $0$.
\end{claim}
\begin{proof}[Proof of Claim~\ref{projection_identity_claim_1}]

Notice that because \begin{equation}
    \mathbf{R}'({\bar{n}}) = 
\begin{pmatrix} 
\mathbf{A}({\bar{n}}) & \mathbf{0}^{m \times (n-\bar{n})} 
\end{pmatrix}
\end{equation}
We have that:
\begin{equation}
    \mathbf{C}' = 
\begin{pmatrix} 
\tilde{\mathbf{C}} & 0 & \hdots & 0 \\
0 & 0 & 0 & 0\\
\vdots & \vdots & \vdots & \vdots\\
0 & 0 & \hdots & 0
\end{pmatrix}
\end{equation}
Where $\tilde{\mathbf{C}}\in \R^{\barn \times \barn}$. 

Notice that $\tilde{\mathbf{C}}$ is a symmetric $\barn \times \barn$ matrix, thus it has an orthonormal eigenbasis. Consider orthonormal matrix $\mathbf{V'}\in \R^{\barn \times \barn}$ s.t. that it's columns are eigenvectors of $\tilde{\mathbf{C}}$. By orthonormality, $\mathbf{V}' \mathbf{V}'^\top = \mathbf{I}_{\barn}$. Define $\tilde{\mathbf{V}}'\in \R^{n\times \barn}$ such that it is $\mathbf{V}'$ padded with zeros. Clearly, $\tilde{\mathbf{V}}'$ is a normalized matrix of the $\barn$ eigenvectors corresponding to the $\barn$ largest eigenvectors of $\mathbf{C}'$, therefore $\mathbf{U}'$ is simply $\tilde{\mathbf{V}}'$. Thus we can see that $\mathbf{U}'\mathbf{U}'^{\top} = \mathbf{I}_{n, \barn}$.
\end{proof}
\begin{claim}\label{projection_identity_claim_2}
   \[ \frac{1}{\sqrt{2}}\|\mathbf{U}\mathbf{U}^\top -\mathbf{I}_{n, \bar{n}}\|_F=\|\sin \Theta(\mathbf{U},\mathbf{U}')\|_F\]
\end{claim}
\begin{proof}[Proof of Claim~\ref{projection_identity_claim_2}]
By Claim~\ref{projection_identity_claim_1}, 
\[\mathbf{U}\mathbf{U}^\top -\mathbf{I}_{n, \bar{n}} = \mathbf{U}\mathbf{U}^\top -\mathbf{U}'\mathbf{U}'^\top.\]
Further, since $\bf{U}, \bf{U}'$ are composed of normalized orthogonal columns, by Lemma \ref{lemma:1.5}, we have that 
\[\frac{1}{\sqrt{2}}\|\mathbf{U}\mathbf{U}^\top-\mathbf{U}'\mathbf{U}'^\top\|_F=\|\sin \Theta(\mathbf{U},\mathbf{U}')\|_F.\]
Put together, we have the desired equality:
\[\frac{1}{\sqrt{2}}\|\mathbf{U}\mathbf{U}^\top -\mathbf{I}_{n, \bar{n}}\|_F = \frac{1}{\sqrt{2}}\|\mathbf{U}\mathbf{U}^\top-\mathbf{U}'\mathbf{U}'^\top\|_F=\|\sin \Theta(\mathbf{U},\mathbf{U}')\|_F.\]
\end{proof}

\begin{claim}\label{projection_identity_claim_3}
    \[ \|\mathbf{U}\mathbf{U}^\top -\mathbf{I}_{n, \bar{n}}\|_F \le \frac{2\sqrt{2}n\kappa}{(\sigma_{\barn}(\bR'(\barn)))^2}\]
\end{claim}
\begin{proof}[Proof of Claim~\ref{projection_identity_claim_3}]

From \citet{yu_useful_2015} Theorem 2 using $r=1$ and $s=\bar{n}$, we have:
\[
||\sin \Theta(\mathbf{U}, \mathbf{U}')||_F \leq \frac{2 \min(\sqrt{\bar{n}}||\mathbf{C} - \mathbf{C}'||_{op}, ||\mathbf{C} - \mathbf{C}' ||_F)}{\min(\lambda_0 - \lambda_1(\mathbf{C}'), \lambda_{\bar{n}}(\mathbf{C}')-\lambda_{\bar{n}+1})}
\]
where $\lambda_0 = \infty$ and $\lambda_{\bar{n}+1}=0$ by construction. Note that by assumption, $\sigma_{\barn}(\bR'(\barn))>0$, so the denominator is well-defined. 
 
Thus we have:
\begin{equation*}
\begin{aligned}
\|\mathbf{U}\mathbf{U}^\top - \mathbf{I}_{n,\bar{n}}\|_F &\leq \frac{2\sqrt{2} \min(\sqrt{\bar{n}}||\mathbf{C} - \mathbf{C}'||_{op}, ||\mathbf{C} - \mathbf{C}' ||_F)}{\min(\lambda_0 - \lambda_1(\mathbf{C}'), \lambda_{\bar{n}}(\mathbf{C}')-\lambda_{\bar{n}+1})} &\quad \text{Claim 2}\\
&= \frac{2\sqrt{2} \min(\sqrt{\bar{n}}||\mathbf{C} - \mathbf{C}'||_{op}, ||\mathbf{C} - \mathbf{C}' ||_F)}{\min(\infty - \lambda_1(\mathbf{C}'), \lambda_{\bar{n}}(\mathbf{C}'))} &\quad \text{by construction}\\
&\leq \frac{2\sqrt{2}\|\mathbf{C} - \mathbf{C}' \|_F}{\min(\infty - \lambda_1(\mathbf{C}'), \lambda_{\bar{n}}(\mathbf{C}'))}\\
&= \frac{2\sqrt{2}\|\mathbf{C} - \mathbf{C}' \|_F}{\lambda_{\bar{n}}(\mathbf{C}')}\\
&=\frac{2\sqrt{2}\|\mathbf{C} - \mathbf{C}' \|_F}{(\sigma_{\barn}(\bR'(\barn)))^2} &\quad \text{Def of singular value}\\
&\leq \frac{2\sqrt{2}n\kappa}{(\sigma_{\barn}(\bR'(\barn)))^2}
\end{aligned}
\end{equation*}
Where the last inequality is as follows. Notice that
\begin{align*}
    \|\mathbf{C}-\mathbf{C}'\|_F & = \|\bR^\top \bR-\bR'(\barn)^\top \bR'(\barn)\|_F\\
    & = \left\|\begin{pmatrix}\mathbf{A}(\barn)^\top \mathbf{A}(\barn) & \mathbf{A}(\barn)^\top \mathbf{B}(\barn)\\
    \mathbf{B}(\barn)^\top \mathbf{A}(\barn) & \mathbf{B}(\barn)^\top \mathbf{B}(\barn)\end{pmatrix}-\begin{pmatrix}\mathbf{A}(\barn)^\top \mathbf{A}(\barn) & 0\\
    0 & 0\end{pmatrix}\right\|_F \\
    & = \left\|\begin{pmatrix}0 & \mathbf{A}(\barn)^\top \mathbf{B}(\barn)\\
    \mathbf{B}(\barn)^\top \mathbf{A}(\barn) & \mathbf{B}(\barn)^\top \mathbf{B}(\barn)\end{pmatrix}\right\|_F
\end{align*}
Where $\mathbf{A}(\barn)\in [0,1]^{m\times \barn}$ are the popular columns of the preference matrix and $\mathbf{B}(\barn) \in [0,1]^{m\times (n-\barn)}$ are the unpopular columns of the matrix. 
Thus we can upper bound every element of $\mathbf{C} - \mathbf{C}'$:
\[(\mathbf{C}-\mathbf{C}')_{ij} \le \max_{i > \bar{n}}\mathbf{R_i}^\top \mathbf{1} = \max_{i > \bar{n}}\|\mathbf{R_i}\|_1 = \kappa.\]
There are $n^2$ elements in $(\mathbf{C}-\mathbf{C}')$
Therefore
\[\|\mathbf{C} - \mathbf{C}' \|_F \le \sqrt{n^2\kappa^2}\le n\kappa.\]
\end{proof}
To reach the final statement of the theorem, notice that $\mathbf{U}\mathbf{U}^\top$ creates the projection matrix, $\mathbf{\Pi}^\star$ that minimizes the optimization problem \ref{eqn:svd_opt} and 
recall that $\Delta(\mathbf{R}, \barn):= \frac{4\sqrt{2}\kappa_{(\bfR, \bar{n})} n\sqrt{n}}{[\sigma_{\bar{n}}(\mathbf{R}'(\bar{n}))]^2}$
\end{proof}

\subsubsection{Recommendations and social welfare bounds}
Now we have some idea of what $\widehat{\bR}$ will be from proposition \ref{projection_identity}. In particular, \Cref{projection_identity} explains (in terms of Frobenius error) how far $\widehat{\bR}$ is from $\bR'(\barn)$! Leveraging this knowledge, we can make statements about the recommendations and resulting social welfare particular $\alpha$ learners give when $(\tbR, \barn) \in \calM$ because if $\widehat{\bR}$ is sufficiently similar to $\bR'(\barn)$, then we know that minority users are not be served their actual top items, which are all indexed $i' > \barn$!. The following theorem (combined with \Cref{myo_si_ub_finite}) is the analog to \Cref{g_maj_b_min2}, but in this appendix setting.

\begin{theorem}[Recommendations are good for majority, bad for minority]\label{g_maj_b_min}
Let $\bR^\star$ be a personal preference matrix such that, for some $\barn \in [n]$, $(\bR^\star, \barn) \in \calM$. Assume that agents all report ratings according to their personal interest such that $\tbR = \bR^\star$. Let there also be an $\alpha$ loss tolerant learner s.t. $k^\star = \barn$, or sufficiently, $\alpha \in \calG(\barn, \tbR)= \calG(\barn, \bR^\star)$. 

Then, all majority users are accurately given one of their top popular items, while minority users are given a popular item. Formally, top-1 item recommendation on $\widehat{\bR}$ satisfies the following two properties:
    \begin{equation} \label{maj_user_sat}
        \mathrm{arg}\max_{i \in [n]}\hat{r}_{u,i} \subseteq \topitems(\tbR, u) \cap [\bar{n}] \quad \forall u \in \calUmaj
    \end{equation}
    and 
    \begin{equation}\label{min_user_unsat}
         \mathrm{arg}\max_{i \in [n]}\hat{r}_{u,i} \subseteq [\bar{n}] \quad \forall u \in \calUmin
    \end{equation}
\end{theorem}
\begin{proof}
For notational cleanliness in this proof, we write $\kappa$ to refer to $\kappa_{(\tbR, \barn)}$ and $\tbR'$ to refer to $\tbR'(\barn)$.

Note that if we consider $\alpha \in \calG(\barn, \tbR)$, by proposition \ref{rhat_equals_nbar}, $k^\star = \barn$ thus, $\widehat{\bR}:= \tbR\mathbf{\Pi}_{\bar{n}}^\star$. By Proposition \ref{projection_identity}, if $\frac{2\sqrt{2}\kappa n}{\sigma_{\bar{n}}(\tbR')} = 0$ then $\mathbf{\Pi}^\star_{\bar{n}} = \mathbf{I}_{n,\bar{n}}$. When this is the case, $\widehat{\bR}= \tbR'$ and the optimal solution for our top-1 item selection problem is such that properties \ref{min_user_unsat} and \ref{maj_user_sat} hold. We want to show that when the Frobenius norm difference of Proposition \ref{projection_identity} is small, under assumptions \ref{maj_users_high} and \ref{min_users_high}, the top-1 item selection problem's solution is as if the Frobenius norm difference were $0$, so properties \ref{min_user_unsat} and \ref{maj_user_sat} still hold.
Let the rows of $\tbR$ be $\tbr_u^\top \in [0,1]^{n}$ for $u \in [m]$ and the columns of $\mathbf{\Pi}^\star_{\bar{n}}$ be $\mathbf{v}_i \in \R^{n}$ for $i \in [n]$. We the consider the problem as follows:
\begin{equation}
\begin{pmatrix} 
\tilde{r}_{1,1}, & \hdots & \tilde{r}_{1,n}\\
\vdots& \hdots & \vdots \\
\tilde{r}_{m,1}, & \hdots & \tilde{r}_{m,n}
\end{pmatrix}
\begin{pmatrix} 
1 + \eps_{1,1}& 0+ \eps_{1,2} &\hdots & 0+ \eps_{1,\bar{n}}& \hdots &0 + \eps_{1,n}\\
0+ \eps_{2,1}& 1+ \eps_{2,2} &\hdots & 0+ \eps_{2,\bar{n}}& \hdots &0+ \eps_{2,n}\\
\vdots& \vdots &\ddots & \vdots& \vdots &\vdots\\
0+ \eps_{i,1}& 0+ \eps_{i,2} &\hdots & 1+ \eps_{i,\bar{n}}& \hdots &0+ \eps_{i,n}\\
0+ \eps_{i+1,1}& 0+ \eps_{i+1,2} &\hdots & 0+ \eps_{i+1,\bar{n}}& \hdots &0+ \eps_{i+1,n}\\
\vdots& \vdots &\vdots & \vdots& \vdots &\vdots\\
0+ \eps_{n,1}& 0+ \eps_{n,2} &\hdots & 0+ \eps_{n,\bar{n}}& \hdots &0+ \eps_{n,n}
\end{pmatrix}
= 
\begin{pmatrix} 
\tbr_1^\top\mathbf{v}_1 & \hdots & \tbr_1^\top\mathbf{v}_n\\
\vdots & \hdots & \vdots\\
\tbr_m^\top\mathbf{v}_1 & \hdots & \tbr_m^\top\mathbf{v}_n
\end{pmatrix}
\end{equation}
Where $\mathbf{\Pi}^\star_{\bar{n}}$ is some perturbed $\mathbf{I}_{n,\bar{n}}$ matrix such that $\sqrt{\sum_{i\in [n]}\sum_{i' \in [n]}|\eps_{i,i'}|^2} \leq \frac{2\sqrt{2}n\kappa}{(\sigma_{\bar{n}}(\tbR'))^2}$. To show the properties \ref{min_user_unsat} and \ref{maj_user_sat} hold for an $\tbR$ that satisfies our $\calM$ assumptions, it is useful to define some bounds on how far off $\tbr_u^\top\mathbf{v}_i$ may be from $\tilde{r}_{u,i}$ with respect to a bound on the norm of the perturbation.\\

\begin{claim}\label{g_maj_b_min_claim_1}
    Fix any $u \in [m]$ and any $i\in [\bar{n}]$ and define an upper bound $x \geq ||\mathbf{\eps}_i||_2$ where $\mathbf{\eps}_{i} \in \R^n$ is the $i$th column of perturbations. The estimate of $\tilde{r}_{u,i}$, $\tbr_u^\top \mathbf{v}_i$, is lower bounded: $\tilde{r}_{u,i} - \sqrt{n}x \leq \tbr_u^\top \mathbf{v}_i$.
\end{claim}
\begin{proof}[Proof of Claim~\ref{g_maj_b_min_claim_1}]
    We can rewrite $\tbr_u^\top \mathbf{v}_i = \tilde{r}_{u,i} + \tilde{r}_{u,i}\eps_{i,i} + \sum_{i' \in [n]\setminus i}\tilde{r}_{u,i'}\eps_{i',i}$. Construct a vector $\mathbf{a}\in \R^n$ such that when $\eps_{i',i} \leq 0$, $a_{i'} = -\eps_{i',i}$ otherwise $a_{i'} = 0$. Because $\tbR \in [0,1]^{m\times n}$ we have that 
    $\tbr_u^\top \mathbf{v}_i \geq \tilde{r}_{u,i} - \tilde{r}_{u,i}a_{i,i} - \sum_{i' \in [n]\setminus i}\tilde{r}_{u,i'}a_{i',i}$. Invoking the upper bound of 1 on $\tilde{r}_{u',i'}$: $\tbr_u^\top \mathbf{v}_i \geq \tilde{r}_{u,i} - \sum_{i' \in [n]}a_{i',i}$. By construction, $a_{i'} \geq 0$ $\forall i' \in [n]$. Thus equivalently: $\tbr_u^\top \mathbf{v}_i \geq \tilde{r}_{u,i} - ||\mathbf{a}||_1$. Because $a_{i'}$ are equal to $-\eps_{i',i}$ or 0, $||\mathbf{a}||_{1} \leq ||\mathbf{\eps}_{i}||_1$ where $\mathbf{\eps}_{i} \in \R^n$ is the $i$th column of perturbations and we can replace $\mathbf{a}$: $\tbr_u^\top \mathbf{v}_i \geq \tilde{r}_{u,i} - ||\mathbf{\eps}_{i}||_1$. From the l1-l2 norm inequality and the l2 norm bound on the perturbation column: $||\mathbf{\eps}_{i}||_1 \leq \sqrt{n}||\mathbf{\eps}_{i}||_2 \leq \sqrt{n}x$. We have: $\tbr_u^\top \mathbf{v}_i \geq \tilde{r}_{u,i} - \sqrt{n}x$.
\end{proof}
Additionally, we use an analogous proof to show the following claim as well:
\begin{claim}\label{g_maj_b_min_claim_2}
    Fix any $u \in [m]$ and any $i\in [\bar{n}]$ and define an upper bound $x \geq ||\mathbf{\eps}_i||_2$ where $\mathbf{\eps}_{i} \in \R^n$ is the $i$th column of perturbations. The estimate of $\tilde{r}_{u,i}$, $\tbr_u^\top \mathbf{v}_i$, is upper bounded: $\tilde{r}_{u,i} + \sqrt{n}x \geq \tbr_u^\top \mathbf{v}_i$.
\end{claim}
\begin{claim}\label{g_maj_b_min_claim_3}
    Fix any $u \in [m]$ and any $i\in \{\bar{n}+1,\dots, n\}$ and define an upper bound $x \geq ||\mathbf{\eps}_i||_2$ where $\mathbf{\eps}_{i} \in \R^n$ is the $i$th column of perturbations. The estimate of $\tilde{r}_{u,i}$, $\tbr_u^\top \mathbf{v}_i$, is upper bounded: $\sqrt{n}x \geq \tbr_u^\top \mathbf{v}_i$.
\end{claim}
\begin{proof}[Proof of Claim~\ref{g_maj_b_min_claim_3}]
    We can rewrite $\tbr_u^\top \mathbf{v}_i = \sum_{i' \in [n]}\tilde{r}_{u,i'}\eps_{i',i}$.Construct a vector $\mathbf{a}\in \R^n$ such that when $\eps_{i',i} \geq 0$, $a_{i'} = \eps_{i',i}$ otherwise $a_{i'} = 0$. Because $\tbR \in [0,1]^{m\times n}$ we have that 
$\tbr_u^\top \mathbf{v}_i \leq \sum_{i' \in [n]}\tilde{r}_{u,i'}a_{i',i}$. Invoking the upper bound of 1 on $\tilde{r}_{u',i'}$: $\tbr_u^\top \mathbf{v}_i \leq \sum_{i' \in [n]}a_{i',i}$. By construction, $a_{i'} \geq 0$ $\forall i' \in [n]$. Thus equivalently: $\tbr_u^\top \mathbf{v}_i \leq ||\mathbf{a}||_1$. Because $a_{i'}$ are equal to $\eps_{i',i}$ or 0, $||\mathbf{a}||_{1} \leq ||\mathbf{\eps}_{i}||_1$ where $\mathbf{\eps}_{i} \in \R^n$ is the $i$th column of perturbations and we can replace $\mathbf{a}$: $\tbr_u^\top \mathbf{v}_i \leq ||\mathbf{\eps}_{i}||_1$. From the l1-l2 norm inequality and the l2 norm bound on the perturbation column: $||\mathbf{\eps}_{i}||_1 \leq \sqrt{n}||\mathbf{\eps}_{i}||_2 \leq \sqrt{n}x$. We have: $\tbr_u^\top \mathbf{v}_i \leq \sqrt{n}x$.
\end{proof}

Now using the above claims we show that given any reported preference matrix $\tbR$ that satisfies assumption \ref{maj_users_high} and \ref{min_users_high}, if learner does $\bar{n}$-truncated SVD, properties \ref{maj_user_sat} and \ref{min_user_unsat} hold.\\
We first show that property \ref{maj_user_sat} holds. By property \ref{projection_identity}, the Frobenius norm for the whole perturbation is upper bounded. Thus, the $L_2$ norm for any individual perturbation vector is also upper bounded by the same value. Thus we invoke Claims \ref{g_maj_b_min_claim_1}, \ref{g_maj_b_min_claim_2}, and \ref{g_maj_b_min_claim_3} with the Frobenius norm bound: $x:= \frac{\Delta(\tbR, \barn)}{2\sqrt{n}}$. For a given majority user $u \in \calUmaj$, we want lower bounds on estimates for the popular top items and upper bounds on estimates for other items:
\begin{enumerate}
    \item $\forall i \in (\topitems(\tbR, u) \cap [\bar{n}])$ , $\tilde{r}_{u,i} - \frac{\Delta(\tbR, \barn)}{2} \leq \tbr_u^\top \mathbf{v}_i$
    \item $\forall i' \in (\topitems(\tbR, u)^{C} \cap [\bar{n}])$, $\tilde{r}_{u,i'} + \frac{\Delta(\tbR, \barn)}{2} \geq \tbr_u^\top \mathbf{v}_{i'}$
    \item $\forall i'' \in \{\bar{n}+1, \dots n\},\quad \frac{\Delta(\tbR, \barn)}{2} \geq \tbr_u^\top \mathbf{v}_{i''}$
\end{enumerate}
By definition of a majority user, $(i_{(top)}^{(u)} \cap [\bar{n}])$ is not empty. But by assumption \ref{maj_users_high}, $\nexists i, i', i''$ such that $\tbr_u^\top \mathbf{v}_i \leq \tbr_u^\top \mathbf{v}_{i'}$ or $\tbr_u^\top \mathbf{v}_i \leq \tbr_u^\top \mathbf{v}_{i''}$ for any $u \in \calUmaj$. It must be the case that $\mathrm{arg}\max_{i \in [n]}\hat{r}_{u,i} \subseteq (i_{(top)}^{(u)} \cap [\bar{n}]) \quad \forall u \in \calUmaj$ and property \ref{maj_user_sat} holds.
\\
Now we show that property \ref{min_user_unsat} holds. Invoking Claims \ref{g_maj_b_min_claim_1} and \ref{g_maj_b_min_claim_3} using property \ref{projection_identity} for the $L_2$ bound, for a given $u \in \calUmin$ we want the lower bound on estimates of popular items to compare to the upper bound on estimates for unpopular items:
\begin{enumerate}
    \item $\forall i \in [\bar{n}],\quad \tilde{r}_{u,i} - \frac{\Delta(\tbR, \barn)}{2} \leq \tbr_u^\top \mathbf{v}_i $
    \item $\forall i' \in \{\bar{n}+1, \dots n\},\quad \frac{\Delta(\tbR, \barn)}{2} \geq \tbr_u^\top \mathbf{v}_{i'}$
\end{enumerate}
But by Assumption \ref{min_users_high}, there exists at least one $i$ such that $\tbr^{\top}_u\mathbf{v}_i > \tbr^{\top}_u\mathbf{v}_{i'} \quad \forall i'$. Therefore $\mathrm{arg}\max_{i \in [n]}\hat{r}_{u,i} \subseteq [\bar{n}] \quad \forall u \in \calUmin$ and property \ref{min_user_unsat} holds.
\end{proof}

Now we have the key idea: in this regime, majority users get an actual top (popular) item, while minority users also get a popular item. Recall though, that tuples $\in \calM$ do not by definition satisfy the assumption that majority/minority groups do not share users Assumption (\ref{maj_min_users_ass}). Therefore, adding this final additional assumption, we have the following corollary directly about social welfare:

\begin{corollary}[Upper Bound on Social Welfare with Truthful Users] \label{myo_si_ub_finite}
    Additionally, if assumption \ref{maj_min_users_ass} holds, %
    we have:
    \begin{equation}
        \SW(\bR^\star, \alpha) \leq |\calUmin| \Rlower + \sum_{u\in \calUmaj} \max_{i\in[n]}r^\star_{u,i} < \text{Max SW Possible}
    \end{equation}
    Where $\Rlower := \max_{u\in \calUmin}\max_{i' \in [\bar{n}]}r^\star_{u,i'}$
\end{corollary}
\begin{proof}
Let $\Topst(u) \in \topitems(\bR^\star, u)$ be some (truthfully) top item for a user $u$. Recall from the notation in the main body of our paper that $\Topk(u)$ represents the recommended item to user $u$ (according to $\widehat{\bR}$).
\begin{align*}
        \SW(\bR^\star, \alpha) &= \sum_{u\in [m]}r^\star_{u,\Topk(u)}\\
        &= \sum_{u\in \calUmin}r^\star_{u,\Topk(u)} + \sum_{u\in \calUmaj}r^\star_{u,\Topk^\star(u)} \tag{property \ref{maj_user_sat}}\\
        &= \sum_{u\in \calUmin}r^\star_{u,\mathrm{arg}\max_{i\in[\bar{n}]}\hat{r}_{u,i}} + \sum_{u\in \calUmaj}r^\star_{u,\Topk^\star(u)} &\tag{\text{property \ref{min_user_unsat}}}\\
        & \leq |\calUmin| \Rlower + \sum_{u\in \calUmaj}r^\star_{u,\Topk^\star(u)}
\end{align*}
By Assumption \ref{maj_min_users_ass}, the strict inequality holds as well.
\end{proof}
Intuitively, Theorem \ref{g_maj_b_min} and Corollary \ref{myo_si_ub_finite} highlight something concerning: when users are interact according to their personal interests with this type of $\alpha$ learner, majority users get their best recommendations, while minority users do not, instead they get recommended some popular item which does not reflect their greatest preferences. Therefore, we see that minority users specifically are missing out on perfect recommendations and maximized social welfare just like we've seen in online MC with \Cref{thm:rank_min_bad} and the main body results \Cref{g_maj_b_min2}!
\subsection{Improving top-1 social welfare via collective action}
Much like in the main body, we are interested in whether and how agents who are in the majority defined by $(\bfR^\star, \barn)$ could improve top-1 social welfare given the learner is $\alpha$-loss tolerant such that $\alpha \in \calG(\bR^\star, \barn)$ by falsifying ratings on just one minority item.

We consider strategic collective misreporting. Much like the main body of the paper, this collective action transforms $\mathbf{R}^\star$ into $\tbR$, which is the same matrix except the $\bar{n}+1$th vector has been changed. However, to generalize the ``uprating" in that section, we now no longer limit collective agents to only \emph{increasing} their rating. Rather, the relevant vector $\br^\star$ is changed to any $\tilde{\mathbf{r}} \in [0,1]^m$, with the constraint that only the changed elements are those corresponding to majority users participating in the collective.
\begin{remark}[Minority item reordering]
    WLOG, we discuss item $\barn + 1$ when it comes to collective strategy. Minority items can be reordered with no consequence.
\end{remark}
\begin{definition}[General Collective Rating]\label{def:gen_alt} 
Consider a ground truth preference matrix $\bR^\star$ such that for some $\barn$, $(\bR^\star, \barn) \in \calM$. WLOG, we consider a [general] collective rating strategy to be one in which the $\barn + 1$th column vector, $\bR^\star_{\barn + 1}$, is replaced with $\tbr$ under the constraints:
\begin{equation*}
    \tilde{r}_u = r^\star_{u, \barn + 1}\quad \forall u \in \calUmin, \quad \tbr \in [0,1]^m
\end{equation*}

Such uprating results in the learner receiving a strategically manipulated preference matrix, $\tbR \in [0,1]^{m\times n}$ instead of the true matrix $\bfR^\star$.
\end{definition}

Naturally, because the goal of manipulating ratings of item $\barn + 1$ is to help minority users who like it, it is important to establish that there exists enough [true] preference for item $\barn + 1$ such that it is ``worthwhile" manipulating. Recall that in the main body we have an assumption that collective uprating goes toward a ``picky'' item and minority users because these groups are particularly in danger of being poorly served by the recommender. For this appendix, we generalize this concept as well and define the relevant targeted item formally:

\begin{assumption}[Manipulated item is sufficiently liked] \label{ass:ea_item_liked}
For a given $(\bR^\star, \barn) \in \calM$, define $\calUswitch:= \{u: u\in \calUmin, (\bar{n}+1) \in \topitems(\bR^\star, u)\} \subseteq \calUmin$ to be the set of minority users who have a top item which is item $\bar{n}+1$ and assume $|\calUswitch| \neq 0$.
Assume the following is true of ground truth preferences:

$\exists \delta \in \R_{>0}$ such that:
    \begin{enumerate}
        \item For minority users whose top item is not $\barn + 1$, the variation of popular item ratings is not too large: 
        \begin{equation}\label{delta_cond_1}
            \sum_{u \in \calUmin\setminus \calUswitch}\max_{i \in [\bar{n}+1]}r^\star_{u,i} \leq \sum_{u \in \calUmin\setminus \calUswitch}\min_{i \in [\bar{n}+1]}r^\star_{u,i}  + \delta
        \end{equation}
        \item The switch users like item $\barn + 1$ sufficiently more than popular items: 
        \begin{equation}\label{delta_cond_2}
            \sum_{u \in \calUswitch}r^\star_{u,(\barn +1)}> \sum_{u \in \calUswitch}\max_{i \in [\bar{n}]}r^\star_{u,i} + \delta
        \end{equation}
    \end{enumerate}
\end{assumption}
Intuitively, the idea of Assumption \ref{ass:ea_item_liked} making it ``worthwhile'' to do collective action on item $\barn + 1$ is that the existence of the positive value $\delta$ formally quantifies the following: (1) the remaining minority agents won't be significantly hurt by the change of $\barn +1$ becoming a popular item and (2) switch users' experience a sufficient increase in utility from receiving item $\barn+1$ because, note, that by definition item $\barn + 1$ is the switch users' top item, but it is not necessarily a top item by a significant margin.

Just like the main body of the paper, we now derive sufficient conditions on $\tbr$ in order to improve social welfare beyond the personal interest baseline. Intuitively, these sufficient conditions represent the following: generalized collective rating increases the $\barn + 1$ singular value and thus ensures that $(\tbR, \barn + 1) \in \calM$ now allowing some minority users to become a part of the majority, giving these ``switch" users all the benefits of being majority under Theorem \ref{g_maj_b_min}.

Because we want conditions for $(\tbR, \barn + 1) \in \calM$ while assuming that $(\bR^\star, \barn) \in \calM$ we need $\Delta(\tbR, \barn + 1)$ based on the given $\bR^\star, \barn$. By extension, we need $\sigma_{\barn + 1}(\tbR'(\barn + 1))$. Rather than use this singular value directly, we use a lower bound, which may be calculated without taking the SVD of $\tbR'(\barn+1)$.

\begin{definition}[$\hat{\sigma}_{\barn + 1}(\tbR'(\barn + 1)$] 
For a given $(\bR^\star, \barn) \in \calM$, we estimate the revealed popular preferences matrix's $\barn +1$th singular value:
\begin{equation*}
    \hat{\sigma}_{\barn + 1}(\tbR'(\barn + 1)):=\sqrt{\min(\tilde{\mathbf{r}}^{\top}\tilde{\mathbf{r}}, \left[\sigma_{\bar{n}}(\bR^{\star'}({\bar{n}}))\right]^2) - ||\tilde{\mathbf{r}}^{\top}\mathbf{A}(\barn)||_2}
\end{equation*}
Where $\mathbf{A}(\barn) \in [0,1]^{m \times \barn}$ are the first $\barn$ columns of $\bR^\star$
\end{definition}

Note that this looks familiar to an estimate used in the $(\calU_\COLL, \eta)$-sufficient singular value gap of the main body. This is because this estimate comes from the same underlying mathematical machinery. This is made evident in the following proofs, but it is not necessary to understand for the theorems conceptually.

With this estimate in hand, we can proceed with our estimate of $\Delta(\tbR, \barn + 1)$, which we call $\Delta(\tbr; \bR^\star, \barn)$. 
\begin{definition}[Collective Sufficient Ratings Gap, $\Delta(\tbr; \bR^\star, \barn)$]
For a given $(\bR^\star, \barn) \in \calM$, we define the sufficient ratings gap needed for a particular collective strategy, $\tbr$, to be:
\begin{equation*}
\Delta(\tbr; \bR^\star, \barn) := \frac{2^{\frac{5}{2}}n^{\frac{3}{2}}\kappa_{(\bfR^\star, \bar{n}+1)}}{\left[\hat{\sigma}_{\barn + 1}(\tbR'(\barn + 1))\right]^2}
\end{equation*}
Where $\kappa := \max_{i' \in \{(\bar{n}+2), \dots, n\}}\|\bR^\star_{i'}\|_1$
\end{definition}

Our sufficient conditions for $\tbr$ ensures that the collective strategy is such that $(\tbR, \barn + 1) \in \calM$ and that learner selects $k^\star = \barn + 1$. The sufficient conditions can be evaluated without actually calculating any resulting singular values of $\tbR$, which may be expensive to do over the entire space of feasible $\tbr$ 

\begin{proposition}[Sufficient Conditions for Effective Collective Strategy]\label{suff_ea_M}
Let there be some personal interest preference matrix, $\bR^\star$, such that for some $\barn \in [n]$, $(\bR^\star, \barn) \in \calM$ and assumptions \ref{maj_min_users_ass} and \ref{ass:ea_item_liked} hold. Also let there be an $\alpha$-loss tolerant learner such that $\alpha \in \calG(\bR^\star, \barn)$. 
    
The following are sufficient conditions on $\tbr$ (definition \ref{def:gen_alt}) to ensure $\SW(\tbR, \alpha)> \SW(\bfR^\star, \alpha)$:
\begin{enumerate}
    \item Collective strategy is significant enough for recommender to learn another dimension of info
    \begin{enumerate}
        \item $\alpha < \hat{\sigma}_{\barn + 1}(\tbR'(\barn + 1))$
    \end{enumerate}
    \item Collective strategy does not overhype the niche item s.t. it is served to majority users.
    \begin{enumerate}
        \item $\forall u \in \calUmaj :\quad \tilde{r}_u < \max_{i \in [n]}r^\star_{u,i} - \Delta(\tbr; \bR^\star, \barn)$
    \end{enumerate}
    \item Majority users remain majority
    \begin{enumerate}
        \item $\forall u \in \calUmaj :\quad \max_{i \in [n]\setminus \topitems(\bR^\star, u)}r^\star_{u,i} < \max_{i \in [n]}r^\star_{u,i} - \Delta(\tbr; \bR^\star, \barn)$
    \end{enumerate}
    \item Switch users become majority
    \begin{enumerate}
        \item $\forall u \in \calUswitch:\quad \max_{i \in [n]\setminus \topitems(\bR^\star, u)}r^\star_{u,i} < r^\star_{u,\bar{n}+1} - \Delta(\tbr; \bR^\star, \barn)$
    \end{enumerate}
    \item Remaining minority users remain minority
    \begin{enumerate}
        \item $\forall u \in \calUmin\setminus \calUswitch: \quad 0 < \max_{i \in [\bar{n}+1]}r^\star_{u,i} - \Delta(\tbr; \bR^\star, \barn)$
    \end{enumerate}
\end{enumerate}
\end{proposition}
Before we present the proof, we discuss the meaning of the \Cref{suff_ea_M}. Because this setting is much mathematically richer, this version of the main body's \Cref{cor:suff_cond} is more complex, but we note that fundamentally they are conceptually similar. \Cref{suff_ea_M}'s 1st condition is this setting's equivalent to $\alpha < \sqrt{\min \{\sigma_{k_{\maj}}(\bR_{\maj}^\star)^2, \eta^2|\calU_\COLL|+\|\bR^\star_{i^\star}\|^2_2\}-\eta\sqrt{\bar{n}}\AV(\bR^\star, \calU_\COLL)}$ and serves the same functionality. The 2nd condition is the analog to $\eta < \kappa$. Finally, conditions 3 through 5 just ensure that the collective strategy is such that $(\bR^\star, \barn) \in \calM$, which is very likely not \emph{necessary}, but we impose it because we do not have guarantees about more general matrices that do not satisfy our mathematical conditions. In the main body, conditions like 3 through 4 are not needed for \Cref{cor:suff_cond} because the majority-minority matrix structure combined with the picky item is so simple that it continues to be very easy to mathematically reason about.

\begin{proof}[Proof of \Cref{suff_ea_M}]
We break this proof into the following claims:
\begin{claim}\label{suff_ea_M_claim_1}
    $\SW(\bfR^\star, \alpha)$ is upper bounded by: \[\sum_{u\in \calUmin}\max_{i\in[\bar{n}]}r^\star_{u,i} + \sum_{u\in \calUmaj}\max_{i\in[n]}r^\star_{u,i}\]
\end{claim}
\begin{claim}\label{suff_ea_M_claim_2}
    The collectively transformed matrix and $\barn+1$ index falls into the popularity gap class, $\calM$. Formally: \[(\widetilde{\mathbf{R}}, \bar{n}+ 1) \in \mathcal{M}\]
\end{claim}
\begin{claim}\label{suff_ea_M_claim_3}
    $\SW(\tbR, \alpha)$ is bounded from below by:
    \[\sum_{u\in \calUmin\setminus \calUswitch}\min_{i\in[\bar{n}+1]}r^\star_{u,i} + \sum_{u\in \calUmaj\cup \calUswitch}\max_{i\in[n]}r^\star_{u,i}\]
    and thus yields desired (strict) inequality.
\end{claim}
\begin{proof}[Proof of Claim~\ref{suff_ea_M_claim_1}]

By assumption, $(\mathbf{R}^\star, \bar{n}) \in \mathcal{M}$ and $\alpha \in \calG(\bR^\star, \barn)$ thus
by proposition \ref{rhat_equals_nbar}, if the users were to submit preferences truthfully according to personal interest, the learner reduces to rank $\bar{n}$. Thus, by proposition \ref{g_maj_b_min}:
\begin{equation*}
        \mathrm{arg}\max_{i \in [n]}\hat{r}_{u,i} \subseteq \topitems(\bR^\star, u) \cap [\bar{n}] \quad \forall u \in \calUmaj
    \end{equation*}
\begin{equation*}
         \mathrm{arg}\max_{i \in [n]}\hat{r}_{u,i} \subseteq [\bar{n}] \quad \forall u \in \calUmin
    \end{equation*}
Which directly yields the desired social welfare upper bound because majority users get their actual maximum value, while minority users cannot do any better than their maximum value amongst the popular items, which is strictly less than the actual maximum value over all items by assumption \ref{maj_min_users_ass}.
\end{proof}
\begin{proof}[Proof of Claim~\ref{suff_ea_M_claim_2}]
Recall that a preference matrix $\mathbf{R}$ such that $(\mathbf{R}, \bar{n})\in \mathcal{M}$ looks like this:
\begin{equation}
    \mathbf{R} = 
\begin{pmatrix} 
\mathbf{P} & \mathbf{U}
\end{pmatrix}
\end{equation}
Where $\mathbf{P} \in [0,1]^{m\times \bar{n}}$ and $\mathbf{U} \in [0,1]^{m\times (n-\bar{n})}$ are the matrices of popular and unpopular item ratings respectively.
Construct the following:
\begin{equation}
    \mathbf{X} = 
\begin{pmatrix} 
\tilde{\mathbf{r}} & \mathbf{P}
\end{pmatrix}
\end{equation}
Where $\tilde{\mathbf{r}} \in [0,1]^{m\times 1}$ is the $\bar{n}+1$ modified column vector of $\mathbf{R}^\star$ (ie. the 1st column of $\mathbf{U}$) to represent majority users' collective strategy. Thus $\mathbf{X} \in [0,1]^{m \times (\bar{n}+1)}$. 
Let $\mathbf{A}:= \mathbf{X}^{\top}\mathbf{X}$. Thus we clearly have 
\[\mathbf{A} = \begin{pmatrix}c & \mathbf{a}^\top\\
    \mathbf{a} & \mathbf{M}\end{pmatrix}\]
Where:
\begin{enumerate}
    \item $\mathbf{M} = \mathbf{P}^{\top}\mathbf{P} \in \R ^{\bar{n}\times \bar{n}}$
    \item $\mathbf{a}^\top := \tilde{\mathbf{r}}^{\top} \mathbf{P} \in [0,1]^{1 \times \bar{n}}$
    \item $c := \tilde{\mathbf{r}}^{\top}\tilde{\mathbf{r}} \in \R$
\end{enumerate}
Note that the eigenvalues of $\mathbf{A}$ would be the same as the squared nonzero singular values of $\tbR'({\bar{n}+1})$. Thus we can use the lower bound given by Lemma \ref{thm:singular_val_LB1} to get a lower bound on $\bar{n}+1$th singular value of $\tilde{\mathbf{R}}'({\bar{n}+1})$. From Lemma \ref{thm:singular_val_LB1}:
\begin{equation*}
\left[\sigma_{\bar{n}+1}(\tbR'({\bar{n}+1})\right]^2 \geq \min(\tilde{\mathbf{r}}^{\top}\tilde{\mathbf{r}}, \left[\sigma_{\bar{n}}(\bR^{\star'}({\bar{n}}))\right]^2) - ||\tilde{\mathbf{r}}^{\top}\mathbf{P}||_2 = \left[\hat{\sigma}_{\barn + 1}(\tbR'(\barn + 1))\right]^2
\end{equation*}
Note that this means that our estimate of delta:
\begin{equation*}
    \Delta(\tbr; \bR^\star, \barn) \geq \Delta(\tbR, \barn + 1)
\end{equation*}
So now that we've established that our estimate is an upper bound on the true $\Delta$ our sufficient conditions clearly ensure that assumptions \ref{maj_users_high}, \ref{min_users_high} would hold on $(\tilde{\mathbf{R}}, \bar{n}+1)$ using the real $\Delta(\tbR, \barn + 1)$. We write this out explicitly below:\\
Assumption \ref{maj_users_high}:
This holds because the users who are the new majority under $(\tilde{\mathbf{R}}, \bar{n}+1)$ are now $u \in \calUmaj \cup \calUswitch$. 
\begin{enumerate}
    \item $\forall u \in \calUmaj :\quad \tilde{r}_u < \max_{i \in [n]}r^\star_{u,i} - \Delta(\tbr; \bR^\star, \barn) \leq \max_{i \in [n]}r^\star_{u,i} - \Delta(\tbR, \barn + 1)$
    \item $\forall u \in \calUmaj :\quad \max_{i \in [n]\setminus \topitems(\bR^\star, u)}r^\star_{u,i} < \max_{i \in [n]}r^\star_{u,i} - \Delta(\tbr; \bR^\star, \barn) \leq \max_{i \in [n]}r^\star_{u,i} - \Delta(\tbR, \barn + 1)$
    \item $\forall u \in \calUswitch:\quad \max_{i \in [n]\setminus \topitems(\bR^\star, u)}r^\star_{u,i} < r^\star_{u,\bar{n}+1} - \Delta(\tbr; \bR^\star, \barn) \leq r^\star_{u,\bar{n}+1} - \Delta(\tbR, \barn + 1)$
\end{enumerate}

Assumption \ref{min_users_high}:
Minority users is slightly more subtle because on $(\tilde{\mathbf{R}}, \bar{n}+1)$ the new minority group is $\subseteq \calUmin \cup \calUswitch$ (recall that majority and minority groups are not necessarily exclusive unless stated). However, once again by construction, the properties $\tilde{\mathbf{r}}$ satisfies ensure that assumption \ref{min_users_high} is satisfied on all $u \in \calUmin \cup \calUswitch$. Because $\forall u \in \calUswitch:$
\begin{equation*}
        \max_{i \in [n]\setminus \topitems(\bR^\star, u)}r^\star_{u,i} < r^\star_{u,\bar{n}+1} - \Delta(\tbr; \bR^\star, \barn) \leq r^\star_{u,\bar{n}+1} - \Delta(\tbR, \barn + 1)
\end{equation*}
Which automatically implies
\begin{equation*}
0<r^\star_{u,\bar{n}+1} - \Delta(\tbr; \bR^\star, \barn) \leq r^\star_{u,\bar{n}+1} - \Delta(\tbR, \barn + 1)
\end{equation*}
And then satisfaction of assumption \ref{min_users_high} for $u \in \calUmin\setminus \calUswitch$ follows directly from the the $\tilde{\mathbf{r}}$ properties again because $\Delta(\tbR, \barn + 1)$ is upper bounded by $\Delta(\tbr; \bR^\star, \barn)$

Of course it is also the case that $\tbR\in [0,1]^{m\times n}$ Thus we have $(\tbR, \bar{n}+1) \in \mathcal{M}$ as desired.
\end{proof}

\begin{proof}[Proof of Claim~\ref{suff_ea_M_claim_3}]
We now leverage the fact that $(\tilde{\mathbf{R}}, \bar{n}+1) \in \mathcal{M}$ to get a lower bound on social welfare. First, we prove that if this $\alpha$ learner sees $\tbR$, he reduces to rank $\barn + 1$. We have from assumptions that $\alpha \in (\sqrt{(n-\bar{n})\kappa_{(\bfR^\star, \bar{n})}}, \sqrt{\min(\tilde{\mathbf{r}}^{\top}\tilde{\mathbf{r}}, \left[\sigma_{\bar{n}}(\bR^{\star'}({\bar{n}}))\right]^2) - ||\tilde{\mathbf{r}}^{\top}\mathbf{P}||_2})$. We need to prove that this guarantees we also have: 
\[\alpha \in (\sigma_{\bar{n}+2}(\tilde{\mathbf{R}}),\sigma_{\bar{n}+1}(\tilde{\mathbf{R}}))\]
We start with the LHS:
\begin{align*}
    \sqrt{\min(\tilde{\mathbf{r}}^{\top}\tilde{\mathbf{r}}, \left[\sigma_{\bar{n}}(\bR^{\star'}({\bar{n}}))\right]^2) - ||\tilde{\mathbf{r}}^{\top}\mathbf{P}||_2} &\leq \sigma_{\bar{n}+1}(\tilde{\mathbf{R}}'({\bar{n}+1})) \tag{Claim~\ref{suff_ea_M_claim_2}}\\
    &\leq \sigma_{\bar{n}+1}(\tilde{\mathbf{R}}) \tag{Corollary \ref{relations_singular_vals}}
\end{align*}
Now the RHS: \\
\begin{align*}
        \sqrt{(n-\bar{n})\kappa_{(\bfR^\star, \bar{n})}} &\geq \sqrt{(n-\bar{n})\kappa_{(\tilde{\bfR}, \bar{n}+1)}} \tag{\text{Definition of $\kappa$}}\\
        &\geq \sigma_{\bar{n}+2}(\tilde{\mathbf{R}}) \tag{\text{Proposition \ref{next_singular_val_low}}}
\end{align*}
Thus the learner reduces to rank $\bar{n}+1$. Because $(\tilde{\bfR}, \barn +1)\in \calM$ and the learner rank reduces to $\barn$, we can now invoke proposition \ref{g_maj_b_min}:
\begin{equation*}
        \mathrm{arg}\max_{i \in [n]}\hat{r}_{u,i} \subseteq \topitems(\tbR, u) \cap [\bar{n}+1] \quad \forall u \in \calUmaj \cup \calUswitch
    \end{equation*}
\begin{equation*}
         \mathrm{arg}\max_{i \in [n]}\hat{r}_{u,i} \subseteq [\bar{n}+1] \quad \forall u \in \calUmin
    \end{equation*}
From this we get the lower bound we want because users $\in \calUmaj \cup \calUswitch$ receive their top item since we ensure  $\calUmaj$ top items are unchanged by the sufficient conditions that guarantee $\forall u \in \calUmaj :\quad \tilde{r}_u < \max_{i \in [n]}r^\star_{u,i}$ and ratings for users $\in \calUswitch$ are unchanged.

Users $u \in \calUmin \setminus \calUswitch$ might receive something as bad as their worst $i \in [\bar{n}+1]$ item:
\[\SW(\tbR, \alpha) \geq \sum_{u\in \calUmin\setminus \calUswitch}\min_{i\in[\bar{n}+1]}r^\star_{u,i} + \sum_{u\in \calUmaj\cup \calUswitch}\max_{i\in[n]}r^\star_{u,i}\]
Thus we have the following $\rho$ bound by invoking the assumption that the switch users sufficiently like their top item (Assumption \ref{ass:ea_item_liked}) (colored for clarity):
\begin{equation*}
    \begin{aligned}
        \rho &\geq \frac{LB(\SW(\tbR, \alpha))}{UB(\SW(\bR^\star, \alpha)}\\
        &= \frac{\sum_{u\in \calUmin\setminus \calUswitch}\min_{i\in[\bar{n}+1]}r^\star_{u,i} + \sum_{u\in \calUmaj\cup \calUswitch}\max_{i\in[n]}r^\star_{u,i}}{\sum_{u\in \calUmin}\max_{i\in[\bar{n}]}r^\star_{u,i} + \sum_{u\in \calUmaj}\max_{i\in[n]}r^\star_{u,i}}\\
        &= \frac{\sum_{u\in \calUmin\setminus \calUswitch}\min_{i\in[\bar{n}+1]}r^\star_{u,i} + {\color{orange}\sum_{u\in \calUswitch}\max_{i\in[n]}r^\star_{u,i} }+ \sum_{u\in \calUmaj}\max_{i\in[n]}r^\star_{u,i}}{{\color{red}\sum_{u\in \calUmin\setminus \calUswitch}\max_{i\in[\bar{n}]}r^\star_{u,i}} + \sum_{u\in \calUswitch}\max_{i\in[\bar{n}]}r^\star_{u,i} +\sum_{u\in \calUmaj}\max_{i\in[n]}r^\star_{u,i}}\\
        &>\frac{\sum_{u\in \calUmin\setminus \calUswitch}\min_{i\in[\bar{n}+1]}r^\star_{u,i} +{\color{orange} \delta + \sum_{u\in \calUswitch}\max_{i\in[\bar{n}]}r^\star_{u,i}} + \sum_{u\in \calUmaj}\max_{i\in[n]}r^\star_{u,i}}{{\color{red}\sum_{u\in \calUmin\setminus \calUswitch}\max_{i\in[\bar{n}]}r^\star_{u,i}} + \sum_{u\in \calUswitch}\max_{i\in[\bar{n}]}r^\star_{u,i} +\sum_{u\in \calUmaj}\max_{i\in[n]}r^\star_{u,i}} \quad {\color{orange}\text{equation \ref{delta_cond_2}}}\\
        &\geq \frac{\sum_{u\in \calUmin\setminus \calUswitch}\min_{i\in[\bar{n}+1]}r^\star_{u,i} + {\color{orange}\delta + \sum_{u\in \calUswitch}\max_{i\in[\bar{n}]}r^\star_{u,i}} + \sum_{u\in \calUmaj}\max_{i\in[n]}r^\star_{u,i}}{{\color{red}\delta + \sum_{u\in \calUmin\setminus \calUswitch}\min_{i\in[\bar{n}+1]}r^\star_{u,i} }+ \sum_{u\in \calUswitch}\max_{i\in[\bar{n}]}r^\star_{u,i} +\sum_{u\in \calUmaj}\max_{i\in[n]}r^\star_{u,i}}\quad {\color{red}\text{equation \ref{delta_cond_1}}}\\
        &= 1
    \end{aligned}
\end{equation*}
Which yields the desired $\rho > 1$
\end{proof}
This concludes the full proof as well.
\end{proof}
\subsubsection{Algorithms to find effective strategies?}
Now, we have shown sufficient conditions under which, in this more complex setting, collective strategy improves the system! For the real-world phenomenon, we can see that under these conditions (and likely broader ones, because these are not strictly necessary), \emph{generalized} collective strategy is theoretically supported. This collective strategy does not require \emph{all} users to uprate by a single shared value $\eta$, but rather allows for broader levels of participation across users. On its face, this is a ``better'' result than that which is in the main body, but it comes at a sacrifice--we do not have natural and simple algorithms to find an effective $\tilde{\br}$. Of course, this does not mean that real users do not have a reasonable way of playing an effective collaborative rating strategy in the real-world, after all, our conditions are only sufficient, so simpler strategies may work, and, as we show in our experiment \ref{sec:experiment}, this is the case to a large extent. That said, we do not have provable guarantees that a computationally efficient algorithm does not exist, so we leave this as an open question.

\subsection{Learner welfare}
Finally,to conclude with all of our analogous results, we discuss the learner welfare. Because the sufficient conditions that we present in this case yield to better user social welfare, all the questions about the learner from the main body still stand. That is, what does potentially incentivizing collective strategization mean for platform design? We briefly cover the analogous results, leaving most of the interpretation to the main body because it is largely the same.

\subsubsection{Personalization accuracy learner}
Recall that the utility for the personalization accuracy learner is based on how well it provides recommendations based on personal interest. Because it is also equivalent to our measure of social welfare and we have obviously already proven that collective strategy improves social welfare under the sufficient conditions, the following corollary follows directly:
\begin{corollary}[Collective increases utility for personalization accuracy learner 2]\label{cor:ben_learner2}
    Under the sufficient conditions for \Cref{suff_ea_M}, an $\alpha$-loss tolerant learner with a personalization accuracy-based utility function would achieve improved utility, i.e., $U_{\BEN}^{\COLL} > U_{\BEN}^{\TRUE}$.
\end{corollary}

\begin{proof}
    This follows directly from the definition of the personalization accuracy learner.
\end{proof}

\subsubsection{Engagement learner}
For the engagement learner, we need to make a slightly more nuanced argument. Recall that here we consider \emph{generalized} collective action strategies (\Cref{def:gen_alt}), whose only constraint is that minority ratings for the chosen item stay the same. Thus, unlike our previous constraint that $\eta > 0$, we cannot guarantee that every strategy which satisfies the sufficient conditions of \Cref{suff_ea_M} increases total ratings (and thus engagement). However, we should acknowledge that this definition purposefully generalizes \emph{mathematically} as much as possible. That is, it is done in such as way to make the most general statement in \Cref{suff_ea_M} possible. Realistically, in the actual movements we study and cite in Section \ref{sec:intro}~\cite{tiktok_blm, mccall_booktoks_2022}, participating users seem to only \emph{increase} interactions with the target content (e.g., leave more comments, like more posts, etc), though to different degrees. Thus, considering the following realistic general collective strategy instead:
\begin{definition}[Realistic General Collective Rating]\label{def:gen_alt_real} 
Consider a ground truth preference matrix $\bR^\star$ such that for some $\barn$, $(\bR^\star, \barn) \in \calM$. WLOG, we consider a [general] collective rating strategy to be one in which the $\barn + 1$th column vector, $\bR^\star_{\barn + 1}$, is replaced with $\tbr$ under the constraints:
\begin{equation*}
    \tilde{r}_u = r^\star_{u, \barn + 1}\quad \forall u \in \calUmin, \quad \tbr \in [0,1]^m, and \quad \tilde{r}_u \geq r^\star_{u, \barn + 1}\quad \forall u \in [m]
\end{equation*}
\end{definition}
clearly \Cref{suff_ea_M} still goes through because any such realistic strategy is also feasible by \Cref{def:gen_alt}, but these realistic strategies are still much more general than those of the main body because users do not have to choose a single uprating value $\eta$! Using this version of general collective rating, we easily see the following:
\begin{proposition}[Collective increases utility for the engagement learner 2]\label{prop:en_learner2}
Under the sufficient conditions for \Cref{suff_ea_M} using \emph{realistic} general collective rating (\Cref{def:gen_alt_real}), an $\alpha$-loss tolerant learner with engagement-based utility would achieve improved utility: $U_{\EN}^{\COLL} \geq U_{\EN}^{\TRUE}$.
\end{proposition}
\begin{proof}
    When agents report truthfully according to personal engagement, $r_{ui}^\star = \tilde{r}_{ui}$. Under the collective strategy, $\tilde{r}_u \geq r^\star_{u, \barn + 1}$.
\end{proof}
\subsubsection{Implications}
As in the main body, mathematically, neither of these results are at all novel beyond what has already been presented in the earlier sections of this appendix. Rather, the point is simply to formalize the intuition that even though these collectives are fundamentally protests against the platform's algorithm, they can serve a positive role in the platform's agenda. On the personalization accuracy front, they improve platform accuracy for at no cost (to the platform) and in terms of engagement they, when implemented realistically, just increase the platform's ad space.

\subsection{Extra results for the popularity gap class}
In the previous section we concluded all of our analogous results using this setting. However, because, to our knowledge, our modeling of majority/minority users, popular/unpopular items, and the class $\calM$ is unique and potentially of future interest, we provide one additional result that we do not explicitly use for collective strategy narrative, but is potentially useful for any future work.

While we have explained how the $\calM$ class of tuples represents matrices with a popularity gap, it is not intuitively clear exactly what combinations of $\bR$ and $\barn$ may work. Can one $\bR$ have multiple values of $\barn$ such that $(\bR, \barn), (\bR, \barn') \in \calM$?
We can show conditions that, given $(\mathbf{R}, \bar{n}) \in \mathcal{M}$, for $\bar{n}' > \bar{n}$, $(\mathbf{R},\bar{n}') \notin \mathcal{M}$. There are perhaps other interesting propositions about this class we leave to future work.
\begin{proposition}[Greater $\bar{n}$ does not satisfy assumptions for $\mathcal{M}$] \label{greater_bar_n_ass}
Define 
    \[
    \underline{\kappa}_{(\bfR, \barn)}:=\min_{i' \in \{(\bar{n}+1), \dots, n\}}||\mathbf{R_{i'}}||_1
    \]
If a tuple $(\mathbf{R}, \bar{n})\in \calM$ and $\underline{\kappa}_{(\bfR, \barn)} > \frac{(n - \bar{n})\kappa_{(\bfR,\barn)}}{4\sqrt{2}n\sqrt{n}}$ then $\nexists \bar{n}' \in \{\bar{n} +1, \dots, n\}$ s.t. $(\mathbf{R}, \bar{n}')$ satisfies the assumptions of the previous subsection.
\end{proposition}
\begin{proof}
Define $\bar{n}' \in \{\bar{n} +1, \dots, n\}$, $\mathbf{R}'(\bar{n}') \in [0,1]^{m \times \bar{n}'}$ to be the matrix $\mathbf{R}$, but with all columns $i > \bar{n}'$ set to be 0 vectors, and $\kappa_{(\bfR, \bar{n}')}:=\max_{i' \in \{(\bar{n}'+1), \dots, n\}}||\mathbf{R_{i'}}||_1$. We want to show that $\sigma_{\bar{n}'}(\mathbf{R}'(\bar{n}')) < \sqrt{4\sqrt{2}n\sqrt{n}\kappa_{(\bfR, \bar{n}')}}$. If this is the case, it definitely cannot be true that $(\mathbf{R}, \bar{n}')$ satisfies the assumptions of the previous subsection because it would require the difference between top and next rating to be greater than 1.

By corollary \ref{relations_singular_vals} and proposition \ref{next_singular_val_low}, $\sigma_{\bar{n}+1}(\mathbf{R}'_{\bar{n}'}) \leq \sqrt{\kappa_{(\bfR,\barn)}(n-\bar{n})}$. Note that this is in terms of $\kappa_{(\bfR,\barn)}$ and not in terms of $\kappa_{(\bfR, \bar{n}')}$. Using the asssumption, we have that: $\frac{\underline{\kappa}_{(\bfR, \barn)}(4\sqrt{2}n\sqrt{n})}{n-\bar{n}}>\kappa_{(\bfR,\barn)}$. Thus we can write:
\[\sigma_{\bar{n}+1}(\mathbf{R}'_{\bar{n}'}) \leq \sqrt{\kappa_{(\bfR,\barn)}(n-\bar{n})} <\sqrt{4\sqrt{2}n\sqrt{n}\underline{\kappa}_{(\bfR, \barn)}} \leq \sqrt{4\sqrt{2}n\sqrt{n}\kappa_{(\bfR, \bar{n}')}}\]
Where the last inequality follows because $\underline{\kappa}_{(\bfR, \barn)}$ is minimum.
\end{proof}
\section{Generalization to top-k}\label{app:topk}
In the main body of the paper, we present all of our results specifically using top-$1$, but as we show in this section, this is primarily to simplify the mathematical formalization needed. Conceptually our primary results go through even for more general top-$k$ regimes. Here we provide some formal statements to support this claim.

\subsection{Top-k regime}
Instead of the recommendation phase as detailed in Protocol \cref{prot:learner}, let the learner recommend the top $k$ items to each user from $\widehat{\bR}$ rather than only the top 1. Assume the learner breaks ties analogously to the top-$1$ setting. 

Formally, the learner serves each user her (approximately) most preferred set of $k$ items: $\mathcal{J}_\texttt{topk}(u):= \text{arg}\max_{\mathcal{I}_k\subseteq [n], |\mathcal{I}_k| = k}\sum_{i \in \mathcal{I}_{k}} \hat{r}_{u,i}$. If there is only one top $k$ set, $\mathcal{J}_\texttt{topk}(u)$, then $\Topk(u)= \mathcal{J}_\texttt{topk}(u)$. If $|\mathcal{J}_\texttt{topk}(u)|>1$, i.e., there are multiple sets of $k$ items most preferred, the learner recommends uniformly randomly from the most popular of the top-$k$ sets: 
\begin{center}
        $\Topk(u) \sim \text{Unif}(\mathcal{J}^\texttt{pop}_\texttt{topk}(u)), \quad \text{where} \quad \mathcal{J}^\texttt{pop}_\texttt{topk}(u):= \arg \max_{\mathcal{I} \in \mathcal{J}_\texttt{topk}(u)} \sum_{i \in \mathcal{I}}\|\hat{\mathbf{R}}_i\|_1$
\end{center}
Using this analogous top-k version of the top-1 setting, we present a version of of existing theorems when $k \leq rank(\mathbf{R}^\star_{\maj})$

\subsection{Results extended to top-k}
\begin{corollary}[Extension of \Cref{g_maj_b_min2}]\label{cor:topk_g_maj_b_min}
    Assume $k \leq rank(\mathbf{R}^\star_{\maj})$. Let the ground truth preference matrix $\bR^\star$ be a majority-minority matrix satisfying Assumption~\ref{assumption:singular_value_gap} and $\bR^{\star} = \tbR$ (i.e., simple personal interest rating). If $\alpha$ is in the singular value gap (i.e. $\alpha \in \mathcal{G}(\bR^\star)$), then all majority users are accurately given their top $k$ items, while minority users are given popular items they do not like. Formally:
   \begin{center} 
        $\Topk(u) \in \arg\max_{\mathcal{I}_k\subseteq [n], |\mathcal{I}_k| = k}\sum_{i \in \mathcal{I}_{k}} r^*_{u,i} \quad \forall u \in \mathcal{U}_{\maj}
        \quad and \quad
        \Topk(u) \subseteq [\bar{n}] \quad \forall u \in \mathcal{U}_{\minor}$
    \end{center}
    Thus,
    \begin{center}
        $\SW(\mathbf{R}^\star, \alpha) = \sum_{u\in \mathcal{U}_{\maj}}\max_{\mathcal{I} \subseteq [n], |\mathcal{I}| = k}\sum_{i \in \mathcal{I}}r^*_{u,i}$
    \end{center}
\end{corollary}
\begin{proof}
In the proof of \Cref{g_maj_b_min2} we show that 
$\widehat{\mathbf{R}} = \begin{pmatrix}U_{\maj}\Sigma_{\maj}V_{\maj}^\top & \mathbf{0} \\
    \mathbf{0} & \mathbf{0} \end{pmatrix} = \begin{pmatrix} \mathbf{R}_{\maj}^\star & \mathbf{0} \\
\mathbf{0} & \mathbf{0}\end{pmatrix}$
Clearly because the rows of the of the matrix representing majority users are preserved, for all $u \in \mathcal{U}_{\maj}$, $\Topk(u)$ reflects the the true $k$ most-preferred items for each user. Note that, due to ties, it may be the case that there are multiple (truly) top $k$ sets and the learner will recommend one of these. 

For minority users, all items are tied to an estimated value of 0. Thus, like in the main theorem, we consider which $k$ sets are most popular. Recall from linear algebra, that the row vectors of $\mathbf{R}^\star_{\maj}$ span $rank(\mathbf{R}^\star_{\maj})$ dimensions. Thus, there must be at least $rank(\mathbf{R}^\star_{\maj})$ different items with some nonzero rating. Thus, at least $rank(\mathbf{R}^\star_{\maj})$ of the first $\bar{n}$ columns of $\widehat{\mathbf{R}}$ will have a positive l-1 norm. Because $k \leq rank(\mathbf{R}^\star_{\maj})$, this means that $\mathcal{J}^{\texttt{pop}}_{\texttt{topk}}(u)$ will consist of $k$ sets made up only of majority items.
\end{proof}

In the main body of the paper we frequently reference $\kappa:= \min_{u \in \mathcal{U}_{maj}}\max_{i \in [n]}r_{u,i}^\star$, this ensures that the collective agents do not uprate so much that they accidentally get recommended the minority item instead of their true favorite. For general top-k recommendation, we define a similar constant, 
\begin{center}
    $\kappa(k):= \min_{u \in \mathcal{U}_{maj}}r_{u,i}^{\star(k)}$ 
\end{center}
where $r_{u,i}^{\star(k)}$ uses the standard ordered statistic notation. Recall that for a given $u \in [m]$, there are $n$ ratings, so this can be stated as a vector $(r^\star_{u,1}, r^\star_{u,2}, \dots, r^\star_{u,n})$, the ordered statistic of this vector is $(r^{\star(1)}_{u,i}, r^{\star(2)}_{u,i'}, \dots, r^{\star(n)}_{u,i''})$

\begin{corollary}[Extension of \Cref{thm:SW_EA_Model1}]\label{thm:topk_collective}
 Assume $k \leq rank(\mathbf{R}^\star_{\maj})$. Let $\bR^\star$ be a majority-minority matrix with a picky item $i^\star>\bar{n}$ and suppose there is some collective uprating to value $\eta$ with collective $\calU_\COLL$ such that $\bR^\star$ gets a $(\eta, \calU_\COLL)$-sufficient singular value gap. If $\eta < \kappa(k)$ and $\alpha \in \mathcal{G}(\bR^\star, \calU_\COLL, \eta)$, then we have that
    \begin{center}$
    \Topk(u) \in \arg\max_{\mathcal{I} \subseteq [n], |\mathcal{I}| = k} \sum_{i \in \mathcal{I}}r_{u,i}^\star\quad \forall u \in \mathcal{U}_{\maj} \cup \mathcal{U}_{i^\star}$
    \end{center}
    \begin{center}
    and
    \end{center}
    \begin{center}$
    \Topk(u) \subseteq [\bar{n} + 1]\quad \forall u \in \mathcal{U}_{\minor}\setminus \mathcal{U}_{i^\star}$
    \end{center}
    Thus,
    \begin{center}$
    \SW(\widetilde{\mathbf{R}}, \alpha) = \sum_{u \in (\mathcal{U}_{\maj} \cup \mathcal{U}_{i^\star})}\max_{\mathcal{I} \subseteq [n], |\mathcal{I}| = k} \sum_{i \in \mathcal{I}}r_{u,i}^\star$
    \end{center}
\end{corollary}
\begin{proof}
From Claim \ref{claim:SW_EA_Model1_1} and \ref{claim:SW_EA_Model1_2} in the proof of \Cref{thm:SW_EA_Model1} we see that the approximation matrix, $\widehat{\mathbf{R}}$ will be as follows:
$\hat{r}_{u,i}=\begin{cases} r_{u,i}^\star & u \in \mathcal{U}_{\maj}, i \in [\bar{n}]\\
 \tilde{r}_{u,i} & u \in (\mathcal{U}_{\maj}\cup \mathcal{U}_{i^\star}), i=i^\star \\
 0 & ow
 \end{cases}$
equivalently: 
$\widehat{\mathbf{R}} = \begin{pmatrix}\widetilde{\mathbf{R}}_{\maj'} & \mathbf{0} \\ \mathbf{0} & \mathbf{0} \end{pmatrix}$

Because $\eta < \kappa(k)$, for all majority users, the ordering of the top $k$ elements are preserved, so clearly their recommended set reflects a set a that is in the true arg max. 

For the $\mathcal{U}_{i^\star}$ users, the ordering of their preferences is similarly preserved recall that the definition of picky users and items indicates that picky users only have a rating for the picky item and zeros elsewhere, so their recommended set must also be accurate to true preferences.

Finally, for users in the $\mathcal{U}_{\minor}\setminus \mathcal{U}_{i^\star}$, all items are tied to an estimated value of 0. Thus, like in the main theorem, we consider which $k$ sets are most popular. Notice from Claim \ref{claim:SW_EA_Model1_2} that $\widehat{R}$ still contains $\mathbf{R}^\star_{\maj}$ in the top left corner. Recall from linear algebra, that the row vectors of $\mathbf{R}^\star_{\maj}$ span $rank(\mathbf{R}^\star_{\maj})$ dimensions. Thus, there must be at least $rank(\mathbf{R}^\star_{\maj})$ different items with some nonzero rating. Thus, at least $rank(\mathbf{R}^\star_{\maj})$ of the first $\bar{n}$ columns of $\hat{\mathbf{R}}$ will have a positive l-1 norm. The $\bar{n} + 1$ column will also clearly have a postive l-1 norm. Because $k \leq rank(\mathbf{R}^\star_{maj})$, this means that $\mathcal{J}^\texttt{pop}_\texttt{topk}(u)$ will consist of $k$ sets made up only of the $\bar{n} + 1$th items. 
\end{proof}

\subsubsection{Conclusions}
We elaborate slightly further on what these results mean and what they imply for all our other statements about algorithmic computation of $\eta$ and learner welfare.

\xhdr{Why limit $k$?}
Of course, these results still do not apply to all possible $k \leq n$. First of all, its easy to see that $k$ being approximately $n$ would be silly, so it is not necessarily a reasonable goal to strive for in our theoretical results anyways. Yes, there are $n$ items, but if the recommender simply presents every user with $n$ items, then this task is trivial. If $k$ is slightly less than $n$ then the the task is no longer \emph{trivial}, but it does reflect a recommender system that is barely narrowing down items and $n$ may be very very large meaning that users are functionally getting no use out of interacting with the platform at all. Intuitively, we consider $k \leq rank(\mathbf{R}^\star_{\maj})$ to be a limitation on the number of recommended items that is more reasonable than both $1$ and $n$. Mathematically, this constraint ensures clean statements about minority users \emph{definitely} getting only popular items. More generally, we could get even looser with $k$, but it would require probabilistic statements because then, the number of recommended items would be so large that while minority users' recommendations would be \emph{dominated} by majority items, they would, by chance, contain certain minority items. Then we could still make similar statements about collective action, though in terms of \emph{expected} social welfare and recommendation, but we avoid doing so in favor of mathematical simplicity.

\xhdr{Implications for other results}
We have results on algorithms to find effective $\eta$s and the learner welfare. Because these extended theorems are so similar to the original statements, there is actually very little change to these results, so rather than providing formal statements we only allude in words. For the algorithm, notice that the main relevant change in sufficient conditions is that $\kappa\rightarrow \kappa(k)$. In terms of $\eta$ and $\calU_\COLL$ this is just a different constant, so all that would be necessary is to switch the parameter to be this new $\kappa$ as a function of $k$. For the learner welfare, also nothing would change. The first learner utility function is equivalent to a social welfare and under these top-$k$ collectives, we still see increased social welfare, so that result still holds. The second learner utility function comes as a result of $\eta > 0$ and for the top-$k$ regime, we have not changed what it means to do collective rating, so this result also holds.
\section{Survey}\label{app:survey}
Our survey supplements existing evidence of altruistic behavior in RecSys from grassroots movements and HCI research. Given small sample size and our focus on theoretical modeling, we do not aim to conduct a robust statistical analysis, but rather present examples of strategies and preliminary results on the scale of altruism. All methodological details and an overview of results are below.
\subsection{Methodology}
Our survey was IRB-exempt as an online survey to adults in the US.

We ran our survey to 100 US-based Prolific users on May 7th, 2025. Prior to the finalized version, we ran two pilot studies each of 5 users (all of whom were excluded from the final study) to ensure questions and format was understandable. Each participant was compensated \$$2.70$. Participants were pre-screened to ensure residence in the United States, a Prolific approval rate $\geq 95$, and a Prolific join date no later than Sept. 1st 2024.

Survey questions were divided into 5 sections: 1. Demographics, 2. Recommender System Use, 3. Self-interested Strategization, 4. Altruistic Strategization, and 5. Fairness Beliefs about Recommender Systems. The order in which participants received sections 3 and 4 were randomized. A full list of questions can be found in Appendix \ref{app:empirical_survey_questions}. To understand users' knowledge and theories about concepts such as collaborative filtering, we asked participants whether they believed their interactions with content affect their own and others' recommendations and how much. To understand whether users use this knowledge in order to interact strategically, we asked participants whether they ever \emph{intentionally} interact(avoid) content with the purpose of increasing(decreasing) its recommendation to themselves/others. To understand whether any strategic behavior may be driven by altruistic beliefs, we ask participants whether they believe accuracy of recommendations and promotion of content is fair across different user groups. We manually add theme/topic codes to textual responses. Authors manually added $\{0,1\}$ codes to indicate textual responses that mention boosting/promoting specific creators, charity, politics, harmful/misinformative content, and privacy.

\subsection{Results}
Survey time took participants an average of 8 minutes and 33 seconds and of 100 responders, all documented using some streaming, social networking, and/or music platform(s) that use recommendation algorithms. The vast majority of participants believed their own interactions would affect their future recommendations while a majority (though smaller) also believed their interactions affected others (Table \ref{tab:recommendation_effect}). 
\begin{table}[ht]
\caption{Number of participants (out of 100) and response to algorithmic impact questions.}\label{tab:recommendation_effect}
\centering
\begin{tabular}{lccc}
\toprule
Do your interactions affect... & Yes & No & Unsure \\
\midrule
your own future recommendations?&  92&  6&  2\\
other people's recommendations?&  57&  15& 28\\
\bottomrule
\end{tabular}
\end{table}

We asked about participants' underlying reasoning (if any) when interacting with content on platforms. Each participant was asked whether they have \emph{intentionally} interacted or avoided interacting with platform content with the explicit purpose of affecting their own or others' recommended feeds (Table \ref{tab:strategic_actions}).\footnote{See Appendix \ref{app:empirical_survey_questions} for definitions and wording as they appeared in the survey.}

\begin{table}[ht]
\caption{Number of participants (out of 100) and response to strategic interaction questions.}\label{tab:strategic_actions}
\centering
\begin{tabular}{lccc}
\toprule
Have you intentionally...  & Yes & No & Unsure \\
\midrule
interacted to affect \textbf{your} recommendations? &  68&  27&  5\\
avoided interacting to affect \textbf{your} recommendations? &  62&  30&  8\\
interacted to affect \textbf{others'} recommendations? &  20&  79&  1\\
avoided interacting to affect \textbf{others'} recommendations? &  20&  75&  5\\
\bottomrule
\end{tabular}
\end{table}

It is unsurprising that our results indicate a majority of surveyed users are strategic for their own recommendations as previous large-scale experiments already indicate this is a prominent behavior~\cite{cen2024measuringstrategizationrecommendationusers}. However, while HCI and real-world evidence also indicates the existence of \emph{interpersonal} strategic interactions~\cite{melo_booktok, karizat_algo_resistance, tiktok_blm, mccall_booktoks_2022}, it is heretofore unclear how common these are across all users. In our study, 32 users had intentionally interacted or avoided to influence other feeds (\Cref{tab:strategic_actions}). This is a surprisingly large group. Those who have intentionally interacted to affect others were fairly consistent in reasoning; 16/20 discussed promoting content from specific sources they liked or morally supported and 6/20 mentioned some form of charity. E.g., one user wanted to increase the reach of small creators:

\textit{``I just do this to support creators and help them grow"}

and another wanted to increase the reach of other users in need:

\textit{``I remember seeing a woman on TikTok who was raising money for a personal cause through GoFundMe, and she asked for support in getting her video on more people’s For You pages. I intentionally interacted with her post by liking, sharing, and leaving supportive comments. I also saved the video to my collection to boost its visibility."}

Of users who have intentionally avoided interactions to affect others' feeds, underlying reasoning is more varied. We evaluated textual responses for mentions of relevance to politics, the spread harmful information, or personal privacy. 7/20 textual responses mention political motivations. 

\textit{``Anything political, I 100\% REFUSE to interact with anything political other than to ad it to my filters or block because engaging in politics is just too dangerous."}

8/20 responses describe wanting to act as a filter for various types of content the user explicitly deems harmful

\textit{``I once avoided liking or commenting on a sensational news post on Facebook because I didn’t want to boost its visibility or contribute to spreading misinformation. I knew that interacting with it would make it more likely to appear in others’ feeds. By ignoring it, I hoped the platform’s algorithm would deprioritize it for others as well."}

2/20 users describe avoiding content so that others will not know that they are recommended this:

\textit{``I do not want people to see what I am interested in in my mental health feeds"}

\subsection{Prolific details and survey questions}
\textbf{Study label}: Survey\\
\textbf{Study name}: Users' Interactions with Content Recommendation Algorithms\\
\textbf{Study Description}: This INSTITUTION\_NAME research survey is a part of an experiment to understand people's interactions with algorithms on social media, streaming, and music platforms. You will be asked about your behavior and underlying reasoning when engaging with content recommended to you by these platforms' algorithms.
\subsubsection{Online survey questions}\label{app:empirical_survey_questions}
The survey consists of five blocks of questions: 
\begin{enumerate}
    \item The first block elicited demographic characteristics: age, employment status, education level, race, state of residence, ethnicity, and gender. 
    \item The second block asks questions about recommendation system use. We broke down recommender systems into three types: social media (e.g., Twitter, Instagram), streaming (e.g., Netflix, Hulu), and music platforms (e.g., Spotify, Sound Cloud). For each type, we asked (1) what specific platforms respondents used in the last week, and (2) how often they currently use those platforms. 
    \item The third block pertains to self-interested strategization on recommender systems. First, we elicit information about \emph{awareness}: do respondents believe that their interactions with platform contents impact their future recommendations both in general and for specific types of interactions (e.g., likes, comments, or subscriptions). Second, we ask about \emph{strategic behavior}: have respondents have  ever intentionally interacted (or avoided interacting) with content in order to impact their future recommendations. If they answer yes, we ask for frequency of this behavior and for them to provide an example. 
    \item The fourth block mirrors the previous block but asks about \emph{altruistic} strategization. We first ask about awareness: do respondents believe that their interactions with platform contents impact \emph{other people's} future recommendations both in general and for specific types of interactions. Second, we ask about strategic behavior: have respondents have  ever intentionally interacted (or avoided interacting) with content in order to impact future recommendations for \emph{other people}. If they answer yes, we ask for frequency of this behavior and for them to provide an example. 
    \item The last block asks about people's beliefs on the accuracy and fairness of recommender systems. We ask if they think (a) the accuracy of content recommendations and (b) the amount of content promotion is fair or not fair for different types of users, and if they think that companies should undertake efforts to increase fairness across users.  
\end{enumerate}
The following pages are the exact questions on the survey including survey logic. We have redacted information that may de-anonymize authors for the reviewing process
\includepdf[pages=-, scale=0.75]{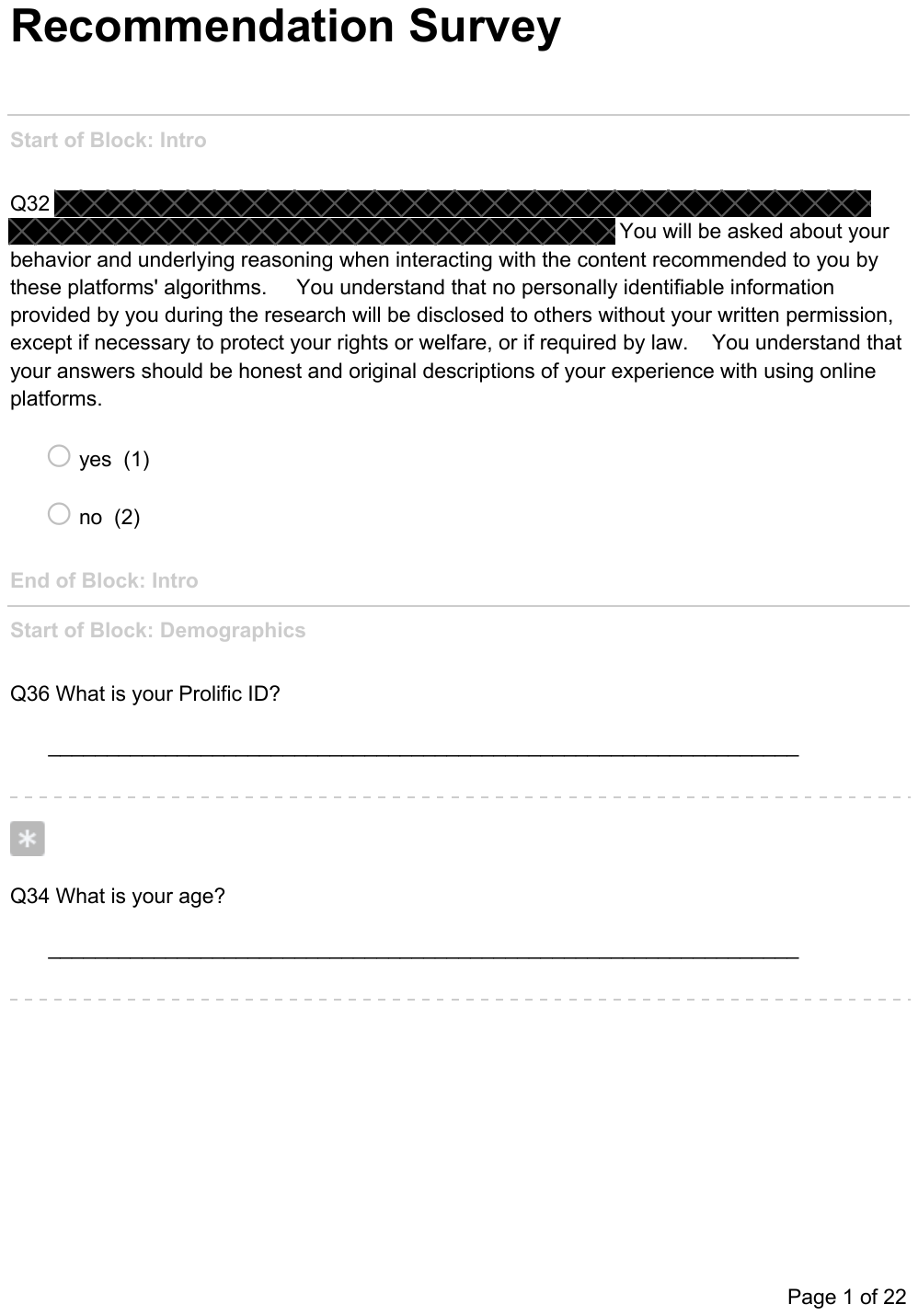}

\end{document}